\newif\iflongversion
\longversiontrue

\documentclass[10pt,twoside]{article}
\usepackage{fullpage}
\usepackage{times}
\usepackage{doi}

\usepackage{graphicx}

\usepackage{fancyhdr}
\usepackage{graphicx}
\usepackage{hyperref}%

\usepackage{xcolor}
\usepackage{paralist}

\usepackage{wrapfig}
\usepackage{blindtext}
\usepackage{etoolbox}
\BeforeBeginEnvironment{wrapfigure}{\setlength{\intextsep}{0pt}}

\usepackage{bm}
\usepackage{marvosym}

\usepackage[bbsets,Dfprime,setrelation]{math}
\usepackage{manalysis}

\usepackage{wasysym}
\usepackage[prefixflatinterpret,bracketmodalinterpret,fixformat,silentconst,sidenotecalculus,longseqcontext,seqinsist]{logic}
\usepackage[prefixflatinterpret,bracketmodalinterpret,fixformat,silentconst,differentialdL,simplenames]{dL}
\def\leftrule{L}%
\def\rightrule{R}%
\newcommand{\bebecomes}{\mathrel{::=}}
\newcommand{\alternative}{~|~}

\newcommand{\neighborhood}{\neighbourhood}
\newcommand{\cneighborhood}[2][]{\closure{\mathcal{U}_{#1}(#2)}}

\newcommand{\States}{\mathbb{S}}

\newcommand{\I}{\dLint[const=I,state=\nu]}

\usepackage{ upgreek }
\newcommand{\solvar}{\boldsymbol\upvarphi}

\usepackage{prettyref}
\newcommand{\rref}[2][]{\prettyref{#2}}
\newrefformat{sec}{Section\,\ref{#1}}
\newrefformat{subsec}{Section\,\ref{#1}}
\newrefformat{def}{Def.\,\ref{#1}}
\newrefformat{thm}{Theorem\,\ref{#1}}
\newrefformat{prop}{Proposition\,\ref{#1}}
\newrefformat{rem}{Remark\,\ref{#1}}
\newrefformat{lem}{Lemma\,\ref{#1}}
\newrefformat{cor}{Corollary\,\ref{#1}}
\newrefformat{ex}{Example\,\ref{#1}}
\newrefformat{cex}{Counterexample\,\ref{#1}}
\newrefformat{tab}{Table\,\ref{#1}}
\newrefformat{fig}{Fig.\,\ref{#1}}
\newrefformat{case}{case\,\ref{#1}}
\newrefformat{foot}{Footnote\,\ref{#1}}

\iflongversion
\newrefformat{app}{Appendix\,\ref{#1}}
\else
\newrefformat{app}{\cite{DBLP:journals/corr/abs-2010-13096}}
\fi

\usepackage{amsthm}

\newcommand{\cmp}{\succcurlyeq}

\newsavebox{\Lightningval}%
\sbox{\Lightningval}{\mbox{\lightning}}

\newsavebox{\Rval}%
\sbox{\Rval}{$\scriptstyle\mathbb{R}$}
\relax

  \newdimen\linferenceRulehskipamount%
  \newdimen\lcalculuscollectionvskipamount%
\renewcommand{\linferenceRuleNameSeparation}{\hspace{4pt}}

\definecolor{vblue}{rgb}{.1,.15,.62}

\definecolor{vgreen}{rgb}{.1,.5,0}
\definecolor{vgray}{rgb}{.35,.35,.35}
\definecolor{vred}{rgb}{.7,0,0}

\renewcommand*{\irrulename}[1]{\text{\textcolor{vgray}{\upshape#1}}}%
\renewcommand{\linferenceRuleNameSeparation}{\hspace{5pt}}
\renewcommand{\linferenceRuleSidenoteSeparation}{\quad}
\renewcommand{\linferenceRuleRefSeparation}{, }

\renewcommand{\axkey}[1]{#1}

\renewcommand*{\lie}[3][]
{\mathcal{L}_{#2}^{\ifthenelse{\equal{#1}{}}{}{^{\left(#1\right)}}}(#3)}
\renewcommand*{\lied}[3][]{\overset{\bm .}{#3}\ifthenelse{\equal{#1}{}}{}{{}^{(#1)}}}

\renewcommand{\Dostar}[1]{\ifthenelse{\equal{#1}{}}{(*)}{-(*)}}

\newcommand{\argx}{(x)}

\DeclareMathOperator{\stable}{Stab}
\newcommand{\stabode}[1]{\stable(#1)}

\DeclareMathOperator{\stablePR}{Stab^{P}_{R}}
\newcommand{\stabodePR}[3]{\stablePR(#1,#2,#3)}

\DeclareMathOperator{\asymptotically}{Asym}
\newcommand{\asymode}[2]{\asymptotically(#1,#2)}

\DeclareMathOperator{\attractive}{Attr}
\newcommand{\attrode}[1]{\attractive(#1)}

\DeclareMathOperator{\astable}{AStab}
\newcommand{\astabode}[1]{\astable(#1)}

\DeclareMathOperator{\attractivePR}{Attr^{P}_{R}}
\newcommand{\attrodePR}[3]{\attractivePR(#1,#2,#3)}

\DeclareMathOperator{\attractiveP}{Attr^{P}}
\newcommand{\attrodeP}[2]{\attractiveP(#1,#2)}

\DeclareMathOperator{\expstable}{EStab}
\newcommand{\expstabode}[1]{\expstable(#1)}

\DeclareMathOperator{\expstableP}{EStab^{P}}
\newcommand{\expstabodeP}[2]{\expstableP(#1,#2)}

\DeclareMathOperator{\distance}{dist}
\newcommand{\dist}[2]{\distance(#1,#2)}

\newcommand{\progxy}{\alpha_{xy}}

\newcommand{\expendulum}{\ensuremath{\alpha_l}}
\newcommand{\expendulumforced}{\ensuremath{\alpha_{i}}}
\newcommand{\expendulumnonlin}{\ensuremath{\alpha_p}}

\newcommand{\exrigid}{\ensuremath{\alpha_r}}
\newcommand{\exrigidfriction}{\ensuremath{\alpha_f}}

\newcommand{\exmoore}{\ensuremath{\alpha_m}}

\newcommand{\fvarA}{\phi}
\newcommand{\fvarB}{\psi}
\newcommand{\rfvar}{P}

\newcommand{\invvar}{I}
\newcommand{\rcfvar}{C}
\newcommand{\rgvar}{G}
\newcommand{\rrfvar}{R}
\newcommand{\rsfvar}{S}

\newcommand{\lterm}{v}
\newcommand{\ptermA}{p}
\newcommand{\ptermB}{q}
\newcommand{\cofterm}{g}

\newcommand{\dotp}[2]{#1 \stimes #2}

\newcommand{\bdr}[1]{\partial #1}

\newtheorem{theorem}{Theorem}
\newtheorem{lemma}[theorem]{Lemma}

\newtheorem{corollary}[theorem]{Corollary}

\newtheorem{definition}[theorem]{Definition}

\theoremstyle{remark}
\newtheorem{example}[theorem]{Example}
\newtheorem{remark}[theorem]{Remark}
\newtheorem{counterexample}[theorem]{Counterexample}{\itshape}{}

\let\oldbibliography\thebibliography
\renewcommand{\thebibliography}[1]{%
  \oldbibliography{#1}%
  \setlength{\itemsep}{0pt}%
}

\makeatletter
   \def\@citecolor{blue}%
   \def\@urlcolor{blue}%
   \def\@linkcolor{blue}%

\def\orcidID#1{\href{http://orcid.org/#1}{\protect\raisebox{-1.25pt}{\protect\includegraphics{orcid.eps}}}}
\makeatother

\begin{document}

\title{Deductive Stability Proofs for Ordinary Differential Equations}

\author{Yong Kiam Tan \and
Andr\'e Platzer
\thanks{
  Computer Science Department, Carnegie Mellon University, Pittsburgh, USA
  {\{yongkiat$|$aplatzer\}@cs.cmu.edu}
  }
}
\date{}

\maketitle              %
\begin{abstract}
Stability is required for real world controlled systems as it ensures that those systems can tolerate small, real world perturbations around their desired operating states.
This paper shows how stability for continuous systems modeled by ordinary differential equations (ODEs) can be formally verified in differential dynamic logic (\dL).
The key insight is to specify ODE stability by suitably nesting the dynamic modalities of \dL with first-order logic quantifiers.
Elucidating the logical structure of stability properties in this way has three key benefits:
\begin{inparaenum}[\it i)]
\item it provides a flexible means of formally specifying various stability properties of interest,
\item it yields rigorous proofs of those stability properties from \dL's axioms with \dL's ODE safety and liveness proof principles, and
\item it enables formal analysis of the relationships between various stability properties which, in turn, inform proofs of those properties.
\end{inparaenum}
These benefits are put into practice through an implementation of stability proofs for several examples in \KeYmaeraX, a hybrid systems theorem prover based on \dL.

\end{abstract}

\section{Introduction}
\label{sec:introduction}

\begin{wrapfigure}[22]{r}{0.30\textwidth}
\includegraphics[width=0.26\textwidth,trim=0 8 0 10,clip]{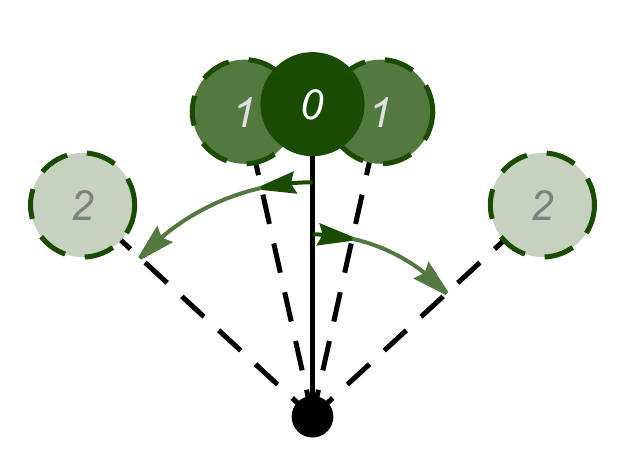}
\includegraphics[width=0.26\textwidth,trim=0 10 0 8,clip]{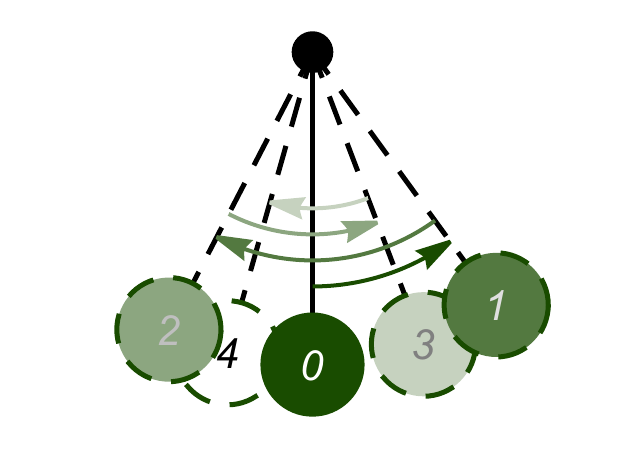}
\vspace{-\baselineskip}
\caption{A pendulum (in green) hung by a rigid rod from a pivot (in black) perturbed from its resting state (bottom) and from its inverted, upright position (top).
Perturbed states (with dashed boundaries) are faded out to show the progression of time.}
\label{fig:pendulumintro}
\end{wrapfigure}
The study of stability has its roots in efforts to understand mechanical systems, particularly those arising in celestial mechanics~\cite{hirsch1984,Liapounoff1907,Poincare92}.
Today, it is an important part of numerous applications in dynamical systems~\cite{MR3837141} and control theory~\cite{10.2307/j.ctvcm4hws,MR1201326}.
This paper studies proofs of stability for continuous dynamical systems described by \emph{ordinary differential equations} (ODEs), such as those used to model  feedback control systems~\cite{10.2307/j.ctvcm4hws,MR1201326}.
For such systems, ODE stability is a key correctness requirement~\cite{10.2307/j.ctt17kkb0d} that deserves fully rigorous proofs \emph{alongside} other key properties such as safety and liveness of those ODEs~\cite{DBLP:journals/jacm/PlatzerT20,DBLP:journals/fac/TanP}.
Despite this, formal stability verification has received less attention compared to proofs of safety and liveness, e.g., through reachability or deductive techniques~\cite{DoyenFPP18}.

Stability for a continuous system (or ODEs) requires that
\begin{inparaenum}[\it i)]
\item its system state always stays close to some desired operating state(s) when initially slightly perturbed from those operating state(s), and
\item those perturbations are eventually dissipated so the system returns to a desired operating state.
\end{inparaenum}
These properties are especially crucial for engineered systems because they must be robust to real world perturbations deviating from idealized system models.
Simple pendulums provide canonical examples of stability phenomena: they are always observed to settle in the rest position of~\rref{fig:pendulumintro} (bottom) after some time regardless of how they are initially released.
In contrast, the inverted pendulum in~\rref{fig:pendulumintro} (top) is \emph{theoretically} also at a resting position but can only be observed transiently in practice because the slightest real world perturbation will cause the pendulum to fall due to gravity.
Stability explains these observations---the resting position is (asymptotically) stable while the inverted position is unstable and requires active control to ensure its stability.
Proofs of safety and liveness properties are still required for the inverted pendulum under control, e.g., its controller must never generate unsafe amounts of torque and the pendulum must eventually reach the inverted position.
The \emph{triumvirate} of safety, liveness, and stability is required for holistic correctness of the inverted pendulum controller.

\begin{wrapfigure}[12]{r}{0.30\textwidth}
\includegraphics[width=0.26\textwidth,trim=0 2 0 0,clip]{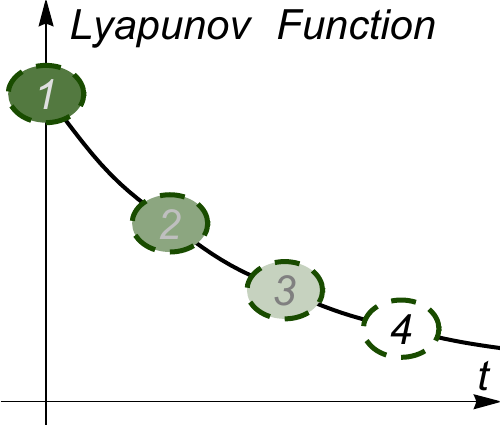}
\vspace{-\baselineskip}
\caption{A Lyapunov function that decreases along the pendulum trajectory shown in~\rref{fig:pendulumintro} (bottom).}
\label{fig:pendulumlyap}
\end{wrapfigure}
The classical way of distinguishing the aforementioned stability situations is by designing a \emph{Lyapunov function}~\cite{Liapounoff1907}, i.e., an energy-like auxiliary measure satisfying certain \emph{arithmetical conditions}~\cite{10.2307/j.ctvcm4hws,MR1201326,MR0450715} which implies that the auxiliary energy decreases along system trajectories towards local minima at the stable resting state(s), see~\rref{fig:pendulumlyap}.
Prior approaches~\cite{DBLP:conf/tacas/AhmedPA20,DBLP:conf/cav/GaoKDRSAK19,DBLP:conf/hybrid/KapinskiDSA14,DBLP:journals/mics/LiuZZ12,DBLP:conf/nolcos/Sankaranarayanan0A13} have emphasized the need to formally verify those arithmetical conditions in order to guarantee that a conjectured Lyapunov function correctly implies stability for a given system.

This paper shows how deductive proofs of ODE stability can be carried out in differential dynamic logic (\dL)~\cite{DBLP:conf/lics/Platzer12a,DBLP:journals/jar/Platzer17,Platzer18}, a logic for \emph{deductive verification} of hybrid systems.\footnote{Hybrid systems are mathematical models describing discrete and continuous dynamics, and interactions thereof. This paper's formal understanding of ODE stability is crucial for subsequent investigation of hybrid systems stability~\cite{DBLP:books/sp/necs2005/Branicky05,10.2307/j.ctt7s02z,DBLP:books/sp/Liberzon03}.}
The key insight is that stability properties can be specified by suitably nesting the dynamic modalities of \dL with quantifiers of first-order logic.
The resulting specifications are amenable to rigorous proof by combining \dL's ODE safety~\cite{DBLP:journals/jacm/PlatzerT20} and liveness~\cite{DBLP:journals/fac/TanP} proof principles with real arithmetic and first-order quantifier reasoning.
This makes it possible to \emph{syntactically derive} stability for a given system from the small set of \dL axioms which, in turn, enables trustworthy stability proofs in the \KeYmaeraX theorem prover for hybrid systems~\cite{DBLP:conf/cade/FultonMQVP15,DBLP:journals/jar/Platzer17}.
Notably, this approach directly verifies \emph{stability specifications}, which goes beyond verifying arithmetic that imply those specifications~\cite{DBLP:conf/tacas/AhmedPA20,DBLP:conf/cav/GaoKDRSAK19,DBLP:conf/hybrid/KapinskiDSA14,DBLP:journals/mics/LiuZZ12,DBLP:conf/nolcos/Sankaranarayanan0A13}.
This is crucial for advanced stability notions because those variations generally require subtle twists to the required arithmetical conditions on their Lyapunov functions~\cite{10.2307/j.ctvcm4hws}; proofs of stability specifications alleviate the onus on system designers to correctly pick and check the appropriate conditions for their applications.
\rref{sec:asymstability} shows how various stability properties for ODE equilibria can be formally specified and proved in \dL with Lyapunov function techniques.
\rref{sec:genstability} generalizes those stability specifications, yielding unambiguous formal specifications of advanced stability properties from the literature~\cite{10.2307/j.ctvcm4hws,MR1201326}, along with their derived proof rules.
These specifications also provide rigorous insights into the logical relationship between various stability notions, which are used to inform their respective proofs.
\rref{sec:casestudies} illustrates the practicality of this paper's \dL approach through several stability case studies formalized in \KeYmaeraX. %

\iflongversion
All omitted definitions and proofs are available in Appendices~\ref{app:proofcalc}--\ref{app:cex}.
\fi

\section{Background: Differential Dynamic Logic}
\label{sec:background}

This section briefly recalls the syntax and semantics of \dL, focusing on its continuous fragment which has a complete axiomatization for ODE invariants~\cite{DBLP:journals/jacm/PlatzerT20}.
Full presentations of \dL, including its discrete fragment, are available elsewhere~\cite{DBLP:journals/jar/Platzer17,Platzer18}.

\paragraph{Syntax and Semantics.}

The grammar of \dL terms is as follows, where $x \in \allvars$ is a variable and $c \in \rationals$ is a rational constant.
These terms are polynomials over $\allvars$
(extensions with Noetherian functions~\cite{DBLP:journals/jacm/PlatzerT20} such as $\text{exp},\sin,\cos$ are possible):
\[
	\ptermA,\ptermB~\bebecomes~x \alternative c \alternative \ptermA + \ptermB \alternative \ptermA \cdot \ptermB
\]

The grammar of \dL formulas is as follows, where $\sim {\in}~\{=,\neq,\geq,>,\leq,<\}$ is a comparison operator and $\alpha$ is a hybrid program:
\begin{align*}
  \fvarA,\fvarB~\bebecomes&~\ptermA \sim \ptermB \alternative \fvarA \land \fvarB \alternative \fvarA \lor \fvarB \alternative \lnot{\fvarA} \alternative \lforall{v}{\fvarA} \alternative \lexists{v}{\fvarA} \alternative \dbox{\alpha}{\fvarA} \alternative\ddiamond{\alpha}{\fvarA}
\end{align*}

This grammar features atomic comparisons ($\ptermA \sim \ptermB$), propositional connectives ($\lnot$, $\land$, $\lor$), first-order quantifiers over the reals ($\lforall{}{}$, $\lexists{}{}$), and the box ($\dbox{\alpha}{\fvarA}$) and diamond ($\ddiamond{\alpha}{\fvarA}$) modality formulas which express that all or some runs of hybrid program $\alpha$ satisfy $\fvarA$, respectively.
The modalities $\dibox{\cdot},\didia{\cdot}$ can be freely nested with first-order and modal connectives, which is crucial for the specification of stability properties in~Sections~\ref{sec:asymstability} and~\ref{sec:genstability}.
Formulas not containing the modalities are formulas of first-order real arithmetic and are written as $\rfvar,\ivr,\rrfvar$.

This paper focuses on the \emph{continuous} fragment of hybrid programs $\alpha \mnodefequiv \pevolvein{\D{x}=\genDE{x}}{\ivr}$, where $\D{x}=\genDE{x}$ is an $n$-dimensional system of ordinary differential equations (ODEs), $\D{x_1}{=}f_1(x), \dots, \D{x_n}{=}f_n(x)$, over variables $x = (x_1,\dots,x_n)$, the LHS $\D{x_i}$ is the time derivative of $x_i$ and the RHS $f_i(x)$ is a polynomial over variables $x$.
The evolution domain constraint $\ivr$ specifies the set of states in which the ODE is allowed to evolve continuously.
When $\ivr$ is the formula $\ltrue$, the ODE is also written as $\D{x}=\genDE{x}$.
For $n$-dimensional vectors $x,y$, the dot product is $\dotp{x}{y} \mdefeq \sum_{i=1}^{n}{x_i y_i} $ and $\norm{x}^2\mdefeq \sum_{i=1}^{n}{x_i^2}$ denotes the squared Euclidean norm.
Variables $z \in \allvars \setminus \{x\}$ not occurring on the LHS of ODE $\D{x}=\genDE{x}$ are \emph{parameters} that remain constant along ODE solutions.
The following parametric ODE model of a simple pendulum is used as a running example.

\begin{example}[Pendulum model]
\label{ex:pendulum}
The ODE $\expendulumnonlin \mnodefequiv \D{\theta} = \omega,~\D{\omega} = {-}\frac{g}{L}\sin(\theta) -b\omega$ models a pendulum (illustrated below) suspended from a pivot by a rod of length $L$, where $\theta$ is the angle of displacement, $\omega$ is the angular velocity of the pendulum, and $g > 0$ is the gravitational constant.
Parameter $a \mnodefeq \frac{g}{L}$ is a positive scaling constant and parameter $b \geq 0$ is the coefficient of friction for angular velocity.
{\makeatletter%
\let\par\@@par%
\par\parshape0%
\everypar{}\begin{wrapfigure}[10]{r}{0.18\textwidth}%
\vspace{-0.7\baselineskip}%
\includegraphics[width=0.14\textwidth]{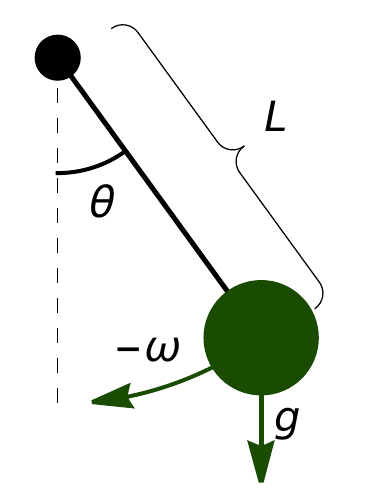}
\vspace{-\baselineskip}%
\end{wrapfigure}
The symbolic parameters $a, b$ make analysis of $\expendulumnonlin$ apply to a range of concrete values, e.g., pendulums that are suspended by a long rod (with large $L$) are modeled by small positive values of $a$, while frictionless pendulums have $b=0$.
A simplification of $\expendulumnonlin$ is used because stability analyses often concern the behavior of the pendulum near its resting (or inverted) state where $\theta = 0$.
For such nearby states with $\theta \approx 0$, the small angle approximation $\sin(\theta) \approx \theta$ yields a linear ODE:\footnote{This linearization is justified by the Hartman-Grobman theorem~\cite{Chicone2006}. A nonlinear polynomial approximation, such as $\sin(\theta) \approx \theta - \frac{\theta^3}{6}$, can also be used.}
\begin{align}
\expendulum \mnodefequiv \D{\theta} = \omega,~\D{\omega} = - a\theta -b\omega
\label{eq:pendulum}
\end{align}
\par}%

{\makeatletter%
\let\par\@@par%
\par\parshape0%
\everypar{}\begin{wrapfigure}[7]{r}{0.18\textwidth}%
\vspace{-1.5\baselineskip}%
\includegraphics[width=0.14\textwidth]{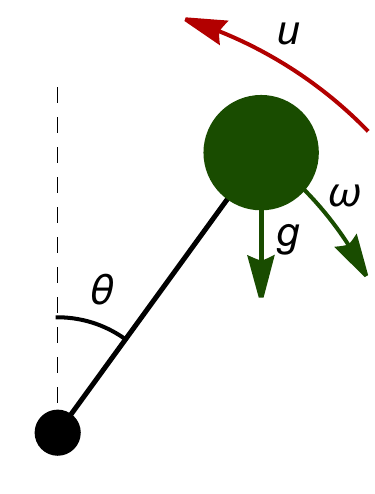}
\vspace{-\baselineskip}%
\end{wrapfigure}

An \emph{inverted} pendulum is modeled by a similar ODE (illustrated on the right) under a change of coordinates.
Such a pendulum requires an external torque input $u(\theta,\omega)$ to maintain its stability; $u(\theta,\omega)$ is determined and proved correct in~\rref{sec:casestudies}.
\begin{align}
\expendulumforced \mnodefequiv \D{\theta} = \omega,~\D{\omega} = a\theta -b\omega - u(\theta,\omega)
\label{eq:pendulumforced}
\end{align}
\par}%

\end{example}

States $\iget[state]{\I} : \allvars \to \reals$ assign real values to each variable in $\allvars$; the set of all states is $\States$.
The semantics of \dL formula $\fvarA$ is the set of states $\imodel{\I}{\fvarA} \subseteq \States$ in which $\fvarA$ is true~\cite{DBLP:journals/jar/Platzer17,Platzer18}, where the semantics of first-order logical connectives are defined as usual, e.g., $\imodel{\I}{\fvarA \land \fvarB} = \imodel{\I}{\fvarA} \cap \imodel{\I}{\fvarB}$.
For ODEs, the semantics of the modal operators is as follows.\footnote{The semantics of \dL formulas is defined compositionally elsewhere~\cite{DBLP:journals/jar/Platzer17,Platzer18}.}
Let $\iget[state]{\I} \in \States$ and $\solvar : [0, T) \to \States$ for some $0<T\leq\infty$, be the unique, right-maximal solution~\cite{Chicone2006} to ODE $\D{x}=\genDE{x}$ with initial value $\solvar(0)=\iget[state]{\I}$:
\begin{align*}
 \m{\imodels{\I}{\dbox{\pevolvein{\D{x}=\genDE{x}}{\ivr} }{\fvarA}}}~\text{iff}~&\text{for all}~0 \leq \tau < T~\text{where}~\solvar(\zeta)\,{\in}\,\imodel{\I}{\ivr}~\text{for all}~0 \leq \zeta \leq \tau\text{:} \solvar(\tau) \in \imodel{\I}{\fvarA} \\
 \m{\imodels{\I}{\ddiamond{\pevolvein{\D{x}=\genDE{x}}{\ivr}}{\fvarA}}}~\text{iff}~&\text{there exists}~0 \leq \tau < T~\text{such that:} \solvar(\tau) \in \imodel{\I}{\fvarA}~\text{and}~\solvar(\zeta) \in \imodel{\I}{\ivr}~\text{for all}~0 \leq \zeta \leq \tau
\end{align*}

For a formula $\rfvar$ the $\varepsilon$-neighborhood of $\rfvar$ with respect to $x$ is defined as $\neighborhood[\varepsilon]{\rfvar} \mdefequiv \lexists{y}{\big(\norm{x-y}^2 < \varepsilon^2 \land \rfvar(y)\big)}$, where the existentially quantified variables $y$ are fresh in $\rfvar$.
The neighborhood formula $\neighborhood[\varepsilon]{\rfvar}$ characterizes the set of states within distance $\varepsilon$ from $\rfvar$, with respect to the dynamically evolving variables $x$.
This is useful for syntactically expressing small $\varepsilon$ perturbations in the stability definitions of Sections~\ref{sec:asymstability} and~\ref{sec:genstability}.
For formulas $\rfvar$ of first-order real arithmetic, the $\varepsilon$-neighborhood, $\neighborhood[\varepsilon]{\rfvar}$, can be equivalently expressed in quantifier-free form by quantifier elimination~\cite{Bochnak1998}.
For example, $\neighborhood[\varepsilon]{x=0}$ is equivalent to the formula $\norm{x}^2 < \varepsilon^2$.
Formulas $\closure{\rfvar}$ and $\bdr{\rfvar}$ are the syntactically definable topological closure and boundary of the set characterized by $\rfvar$, respectively~\cite{Bochnak1998}.

\paragraph{Proof Calculus.}
\irlabel{qear|\usebox{\Rval}}
\irlabel{notr|$\lnot$\rightrule}
\irlabel{notl|$\lnot$\leftrule}
\irlabel{orr|$\lor$\rightrule}
\irlabel{orl|$\lor$\leftrule}
\irlabel{andr|$\land$\rightrule}
\irlabel{andl|$\land$\leftrule}
\irlabel{implyr|$\limply$\rightrule}
\irlabel{implyl|$\limply$\leftrule}
\irlabel{equivr|$\lbisubjunct$\rightrule}
\irlabel{equivl|$\lbisubjunct$\leftrule}
\irlabel{id|id}
\irlabel{cut|cut}
\irlabel{weakenr|W\rightrule}
\irlabel{weakenl|W\leftrule}
\irlabel{existsr|$\exists$\rightrule}
\irlabel{existsrinst|$\exists$\rightrule}
\irlabel{alll|$\forall$\leftrule}
\irlabel{alllinst|$\forall$\leftrule}
\irlabel{allr|$\forall$\rightrule}
\irlabel{existsl|$\exists$\leftrule}
\irlabel{iallr|i$\forall$}
\irlabel{iexistsr|i$\exists$}

All derivations and proof rules are presented in a classical sequent calculus.
The semantics of \emph{sequent} \(\lsequent{\Gamma}{\fvarA}\) is equivalent to the formula \((\landfold_{\fvarB \in\Gamma} \fvarB) \limply \fvarA\). A sequent is valid iff its corresponding formula is valid.
Completed branches in a sequent proof are marked with $\lclose$.
Assumptions $\fvarB \in \Gamma$ that have only ODE parameters as free variables remain true along ODE evolutions and are soundly kept across ODE deduction steps~\cite{DBLP:journals/jar/Platzer17,Platzer18}.
First-order real arithmetic is decidable~\cite{Bochnak1998} so we assume such a decision procedure and label proof steps with \irref{qear} when they follow from real arithmetic.
Axioms and proof rules are \emph{derivable} iff they can be deduced from sound \dL axioms and proof rules~\cite{DBLP:journals/jar/Platzer17,Platzer18}.

Formula $\invvar$ is an \emph{invariant} of the ODE $\pevolvein{\D{x}=\genDE{x}}{\ivr}$ iff the formula \(\invvar \limply \dbox{\pevolvein{\D{x}=\genDE{x}}{\ivr}}{\invvar}\) is valid.
The \dL proof calculus is \emph{complete} for ODE invariants~\cite{DBLP:journals/jacm/PlatzerT20}, i.e., any true ODE invariant expressible in first-order real arithmetic can be proved in the calculus.
The calculus also supports refinement reasoning~\cite{DBLP:journals/fac/TanP} for proving ODE liveness properties \(\rfvar \limply \ddiamond{\pevolvein{\D{x}=\genDE{x}}{\ivr}}{\rrfvar}\), which says that the goal $\rrfvar$ is reached along the ODE $\pevolvein{\D{x}=\genDE{x}}{\ivr}$ from precondition $\rfvar$.

An important syntactic tool for reasoning with ODE $\D{x}=\genDE{x}$ is the \emph{Lie derivative} of term $\ptermA$ defined as $\lied[]{\genDE{x}}{\ptermA} \mdefeq \sum_{x_i\in x} \Dp[x_i]{\ptermA} f_i(x)$, whose semantic value is equal to the time derivative of the value of $\ptermA$ along solutions $\solvar$ of the ODE~\cite{DBLP:journals/jar/Platzer17,DBLP:journals/jacm/PlatzerT20}.
They are provably definable in \dL using syntactic differentials~\cite{DBLP:journals/jar/Platzer17}.

\section{Asymptotic Stability of an Equilibrium Point}
\label{sec:asymstability}

This section presents Lyapunov's classical notion of asymptotic stability~\cite{Liapounoff1907} and its formal specification in \dL.
This formalization enables the derivation of \dL stability proof rules with \emph{Lyapunov functions}~\cite{10.2307/j.ctvcm4hws,MR1201326,Liapounoff1907,MR0450715}.
Several related stability concepts are formalized in \dL, along with their relationships and rules.

\subsection{Mathematical Preliminaries}
\label{subsec:asymstabilitymath}
An \emph{equilibrium point} of ODE $\D{x}=\genDE{x}$ is a point $x_0 \in \reals^n$ where $\genDE{x_0} = 0$, so a system that starts at $x_0$ stays at $x_0$ along its continuous evolution.
Such points are often interesting in real-world systems, e.g., the equilibrium point $\theta=0, \omega=0$ for $\expendulum$ from~\rref{eq:pendulum} is the resting state of a pendulum.
For a controlled system, equilibrium points often correspond to desired steady system states where no further continuous control input (modeled as part of $\genDE{x}$) is required~\cite{MR1201326}.

For brevity, assume the origin $0 \in \reals^n$ is an equilibrium point of interest.
Any other equilibrium point(s) of interest $x_0 \in \reals^n$ can be translated to the origin with the change of coordinates $x \mapsto x - x_0$ for the ODE%
\iflongversion%
, see~\rref{lem:translation}.
\fi

The following definition of asymptotic stability is standard~\cite{10.2307/j.ctvcm4hws,MR1201326,MR0450715}.\footnote{Some definitions require, or implicitly assume, right-maximal solutions $x(t)$ to be global, i.e., with $T = \infty$, see~\cite[Definition 4.1]{MR1201326} and associated discussion. The definitions presented here are better suited for subsequent generalizations.}

\begin{definition}[Asymptotic stability~\cite{10.2307/j.ctvcm4hws,MR1201326,MR0450715}]
The origin $0 \in \reals^n$ of ODE $\D{x}=\genDE{x}$ is
\begin{itemize}
\item \textbf{stable} if, for all $\varepsilon > 0$, there exists $\delta > 0$ such that for all initial states $x=x(0)$ with $\norm{x} < \delta$, the right-maximal ODE solution $x(t) : [0,T) \to \reals^n$ satisfies $\norm{x(t)} < \varepsilon$ for all times $0 \leq t < T$,
\item \textbf{attractive} if there exists $\delta > 0$ such that for all $x=x(0)$ with $\norm{x} < \delta$, the right-maximal ODE solution $x(t) : [0,T) \to \reals^n$ satisfies $\lim_{t \to T}{x(t) = 0}$,
\item \textbf{asymptotically stable} if it is stable and attractive.
\end{itemize}
\label{def:asymstabmath}
\end{definition}

These definitions can be understood using the resting state of the pendulum from~\rref{fig:pendulumintro} (bottom) which is asymptotically stable.
When the pendulum is given a light push from its bottom resting state (formally, $\norm{x} < \delta$), it gently oscillates near that resting state (formally, $\norm{x(t)} < \varepsilon$).
In the presence of friction, these oscillations eventually dissipate so the pendulum asymptotically returns to its resting state (formally, $\lim_{t \to T}{x(t) = 0}$).
This behavior is \emph{local}, i.e., for any given $\varepsilon > 0$, there \emph{exists} a sufficiently small $\delta > 0$ perturbation of the initial state that results in gentle oscillations with $\norm{x(t)} < \varepsilon$, see~\rref{fig:stability2} (left).
A strong push, e.g., with $\delta > \varepsilon$, could instead cause the pendulum to spin around on its pivot.
\begin{remark}
Stability and attractivity \emph{do not} imply each other~\cite[Chapter I.2.7]{MR0450715}.
However, if the origin is stable, attractivity can be defined in a simpler way.
This is proved in \dL, after characterizing stability and attractivity syntactically.
\end{remark}

\subsection{Formal Specification}
\label{subsec:asymstabilitydl}

\begin{figure}[t]
\centering%
\includegraphics[width=.85\textwidth,trim=0 380 0 380,clip]{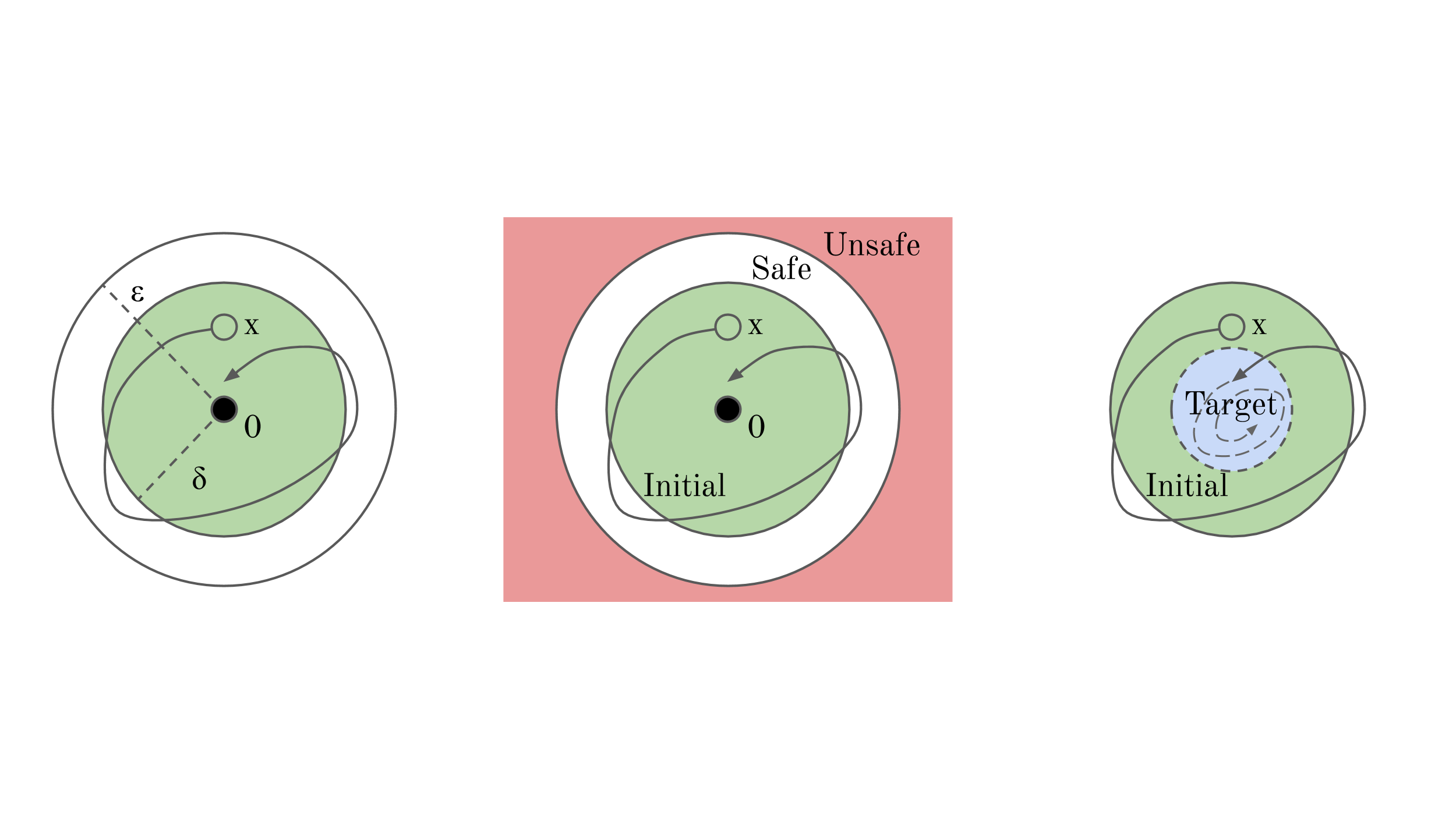}
\caption{Solutions from points in the $\delta$ ball around the origin, like the green initial point $x$, remain within the $\varepsilon$ ball around the origin $0 \in \reals^n$ (black dot) and asymptotically approach the origin. The latter two plots illustrate how asymptotic stability for an ODE can be broken down into a pair of (quantified) ODE safety and liveness properties.}
\label{fig:stability2}
\end{figure}

The formal specification of asymptotic stability in \dL combines \begin{inparaenum}[\it i)]
\item the dynamic modalities of \dL, which are used to quantify over the dynamics of the ODE, and
\item the first-order logic quantifiers, which are used to express combinations of (topologically) local and asymptotic properties of those dynamics.
\end{inparaenum}

\begin{lemma}[Asymptotic stability in \dL]
The origin of ODE $\D{x}=\genDE{x}$ is, respectively,
\begin{inparaenum}[\it i)]
\item \textbf{stable}, \item \textbf{attractive}, and \item \textbf{asymptotically stable}
\end{inparaenum} iff the \dL formulas
\begin{inparaenum}[\it i)]
\item $\stabode{\D{x}=\genDE{x}}$, \item $\attrode{\D{x}=\genDE{x}}$, and \item $\astabode{\D{x}=\genDE{x}}$
\end{inparaenum} respectively are valid.
Variables $\varepsilon,\delta$ are fresh, i.e., not in $x, f(x)$.
\begin{align*}
\stabode{\D{x}=\genDE{x}} &\mnodefequiv \lforall{\varepsilon {>} 0} { \lexists{\delta {>} 0}{ \lforall{x}{\big( \neighborhood[\delta]{x=0} \limply \dbox{\D{x}=\genDE{x}}{\,\neighborhood[\varepsilon]{x=0}}\big)}}} \\
\attrode{\D{x}=\genDE{x}} &\mnodefequiv \lexists{\delta {>} 0} { \lforall{x}{\big( \neighborhood[\delta]{x=0} \limply \asymode{\D{x}=\genDE{x}}{x=0}\big)}} \\
\astabode{\D{x}=\genDE{x}} &\mnodefequiv \stabode{\D{x}=\genDE{x}} \land \attrode{\D{x}=\genDE{x}}
\end{align*}

Formula $\asymode{\D{x}=\genDE{x}}{\rfvar} \mnodefequiv \lforall{\varepsilon {>} 0}{\ddiamond{\D{x}=\genDE{x}}{ \dbox{\D{x}=\genDE{x}}{\,\neighborhood[\varepsilon]{\rfvar}}}}$ characterizes the set of states that asymptotically approach $\rfvar$ along ODE solutions.

\label{lem:asymstabdl}
\end{lemma}

Formula $\stabode{\D{x}=\genDE{x}}$ is a syntactic \dL rendering of the corresponding quantifiers from~\rref{def:asymstabmath}.
The safety property $\neighborhood[\delta]{x=0} \limply \dbox{\D{x}=\genDE{x}}{\,\neighborhood[\varepsilon]{x=0}}$ expresses that solutions starting from the $\delta$-neighborhood of the origin always (for all times) stay safely in the $\varepsilon$-neighborhood, as visualized in~\rref{fig:stability2} (middle).

Formula $\attrode{\D{x}=\genDE{x}}$ uses the subformula $\asymode{\D{x}=\genDE{x}}{x=0}$ which characterizes the limit in~\rref{def:asymstabmath}.
Recall $\lim_{t \to T}{x(t) = 0}$ iff for all $\varepsilon > 0$ there exists a time $\tau$ with $0 \leq \tau < T$ such that for all times $t$ with $\tau \leq t < T$, the solution satisfies $\norm{x(t)} < \varepsilon$, i.e., the limit requires for all distances $\varepsilon > 0$, the ODE solution will \emph{eventually always} be within distance $\varepsilon$ of the origin, as visualized in~\rref{fig:stability2} (right).
This limit is characterized using nested $\didia{\cdot}\dibox{\cdot}$ modalities, together with first-order quantification according to~\rref{def:asymstabmath}.
More generally, formula $\asymode{\D{x}=\genDE{x}}{\rfvar}$ characterizes the set of initial states where the right-maximal ODE solution asymptotically approaches $\rfvar$; this set is known as the \emph{region of attraction} of $\rfvar$~\cite{MR1201326}.
Thus, attractivity requires that the region of attraction of the origin contains an open neighborhood $\neighborhood[\delta]{x=0}$ of the origin.

From~\rref{lem:asymstabdl}, proving validity of the formula $\astabode{\D{x}=\genDE{x}}$ yields a rigorous proof of asymptotic stability for $\D{x}=\genDE{x}$.
However, if the origin is stable, then attractivity can be provably simplified with the following corollary.

\begin{corollary}[Stable attractivity]
The following axiom is derivable in \dL.

\noindent
\begin{calculuscollection}
\begin{calculus}
\dinferenceRule[stabattr|SAttr]{}
{
\linferenceRule[impl]
  {\stabode{\D{x}=\genDE{x}}}
  {\big(\!\asymode{\D{x}=\genDE{x}}{x{=}0} \lbisubjunct \lforall{\varepsilon{>}0}{\ddiamond{\D{x}=\genDE{x}}{\,\neighborhood[\varepsilon]{x{=}0}}}\big)}
}{}
\end{calculus}
\end{calculuscollection}
\label{cor:asymstabsimp}
\end{corollary}

\rref{cor:asymstabsimp} simplifies the syntactic characterization of the region of attraction for stable equilibria from a nested $\didia{\cdot}\dibox{\cdot}$ formula to a $\didia{\cdot}$ formula, which is then directly amenable to ODE liveness reasoning~\cite{DBLP:journals/fac/TanP}.
This corollary is used to simplify proofs of asymptotic stability, as explained next.

\subsection{Lyapunov Functions}

\emph{Lyapunov functions} are the standard tool for showing stability of general, non-linear ODEs~\cite{10.2307/j.ctvcm4hws,MR1201326,MR0450715} and finding suitable Lyapunov functions is an important problem in its own right~\cite{DBLP:conf/tacas/AhmedPA20,261424,DBLP:conf/cav/GaoKDRSAK19,DBLP:conf/hybrid/KapinskiDSA14,DBLP:journals/mics/LiuZZ12,1184414,Parrilo00,DBLP:conf/nolcos/Sankaranarayanan0A13,DBLP:journals/automatica/TopcuPS08}.
This section shows how a candidate Lyapunov function, once found, can be used to rigorously prove stability.
The following proof rules derive Lyapunov stability arguments~\cite{10.2307/j.ctvcm4hws,MR1201326,MR0450715} syntactically in \dL.

\begin{lemma}[Lyapunov functions]
The following Lyapunov function proof rules are derivable in \dL.

\noindent
\begin{calculus}
\dinferenceRule[Lyap|Lyap${_\geq}$]{Lyapunov}
{\linferenceRule
  {\lsequent{} {\genDE{0}=0 \land \lterm(0) = 0} \quad
   \lsequent{}{\lexists{\gamma{>}0}{\lforall{x}{\big(0 {<} \norm{x}^2 {\leq} \gamma^2 \limply \lterm > 0 \land \boldsymbol{\lied[]{\genDE{x}}{\lterm} \leq 0}\big)}} }}
  {\lsequent{} {\stabode{\D{x}=\genDE{x}}}}
}{}
\dinferenceRule[StrictLyap|Lyap${_>}$]{Strict Lyapunov}
{\linferenceRule
  {\lsequent{} {\genDE{0}=0 \land \lterm(0) = 0} \quad
   \lsequent{}{\lexists{\gamma{>}0}{\lforall{x}{\big(0 {<} \norm{x}^2 {\leq} \gamma^2 \limply \lterm > 0 \land \boldsymbol{\lied[]{\genDE{x}}{\lterm} \boldsymbol{<} 0}\big)}} }  }
  {\lsequent{} {\astabode{\D{x}=\genDE{x}}}}
}{}
\end{calculus}
\label{lem:lyapunov}
\end{lemma}

Rules~\irref{Lyap+StrictLyap} use the Lyapunov function $\lterm$ as an auxiliary, energy-like function near the origin which is positive and has non-positive (resp. negative~\irref{StrictLyap}) derivative $\lied[]{\genDE{x}}{\lterm}$.
This guarantees that $\lterm$ is non-increasing (resp. decreasing) along ODE solutions near the origin, see~\rref{fig:pendulumlyap}.
The right premise of both rules use $\lexists{\gamma{>}0}{\lforall{x}{\big(0 {<} \norm{x}^2 {\leq} \gamma^2 \limply \cdots\big)}}$ to require that the Lyapunov function conditions are true in a $\gamma$-neighborhood of the origin.
The subtle difference in sign condition for $\lied[]{\genDE{x}}{\lterm}$ between rules~\irref{Lyap+StrictLyap} is illustrated for the pendulum.

\begin{example}[Pendulum asymptotic stability]
\label{ex:asympendulum}
For ODE $\expendulum$ from~\rref{eq:pendulum}, a suitable Lyapunov function for proving its stability~\cite{MR1201326} is $\lterm \mnodefeq a \frac{\theta^2}{2} + \frac{(b\theta + \omega)^2 + \omega^2}{4}$, where the Lie derivative of $\lterm$ along $\expendulum$ is $\lied[]{\expendulum}{\lterm} \mnodefeq -\frac{b}{2}(a\theta^2+\omega^2)$.
Stability\footnote{For the trigonometric pendulum ODE $\expendulumnonlin$ from~\rref{ex:pendulum}, the Lyapunov function $\lterm \mnodefeq a (1-\cos(\theta)) + \frac{(b\theta + \omega)^2 + \omega^2}{4}$ with Lie derivative $\lied[]{\expendulumnonlin}{\lterm} \mnodefeq -\frac{b}{2}(a\theta\sin(\theta)+\omega^2)$ proves its stability~\cite{MR1201326} but requires arithmetic reasoning over trigonometric functions.}
is formally proved in \dL for \emph{any} parameter values $a > 0, b \geq 0$ using rule~\irref{Lyap} because both of its resulting arithmetical premises are provable by~\irref{qear}.
The full \dL derivation, also used in \KeYmaeraX (\rref{sec:casestudies}), is shown in the proof of~\rref{lem:lyapunov}
\iflongversion
(\rref{app:proofs}).
\fi

When $b > 0$, i.e., friction is non-negligible, an identical derivation with~\irref{StrictLyap} instead of~\irref{Lyap} proves asymptotic stability because $-\frac{b}{2}(a\theta^2+\omega^2)$ is negative except at the origin.
Indeed, displacements to the pendulum's resting state can only be dissipated in the presence of friction, not when $b=0$.

\end{example}

\subsection{Asymptotic Stability Variations}
\label{subsec:asymstabvar}
Asymptotic stability is a strong guarantee about the local behavior of ODE solutions near equilibrium points of interest.
In certain applications, stronger stability guarantees may be needed for those equilibria~\cite{MR1201326}.
This section examines two standard stability variations, shows how they can be proved in \dL, and formally analyzes their logical relationship with asymptotic stability.

\paragraph{Exponential Stability.}
As the name suggests, the first stability variation, exponential stability, guarantees an exponential rate of convergence towards the equilibrium point from an initial displacement.
This is useful, e.g., for bounding the time spent by a perturbed system far away from its desired operating state.

\begin{definition}[Exponential stability~\cite{10.2307/j.ctvcm4hws,MR1201326,MR0450715}]
The origin $0 \in \reals^n$ of ODE $\D{x}=\genDE{x}$ is \textbf{exponentially stable} if there are positive constants $\alpha, \beta, \delta>0$ such that for all initial states $x=x(0)$ with $\norm{x} < \delta$, the right-maximal ODE solution $x(t) : [0,T) \to \reals^n$ satisfies $\norm{x(t)} \leq \alpha\norm{x(0)}\exp{(-\beta t)}$ for all times $0 \leq t < T$.
\label{def:expstab}
\end{definition}

Exponential stability bounds the norm of solutions to ODE $\D{x}=\genDE{x}$ near the origin by a decaying exponential.
It is specified in \dL as follows.

\begin{lemma}[Exponential stability in \dL]
The origin of ODE $\D{x}=\genDE{x}$ is \textbf{exponentially stable} iff the following \dL formula is valid. Variables $\alpha,\beta, \delta, y$ are fresh, i.e., not in $x, \genDE{x}$.
\begin{align*}
\expstabode{\D{x}=\genDE{x}} \mnodefequiv &\, \lexists{\alpha{>}0}{\lexists{\beta {>} 0}{\lexists{\delta {>} 0} \lforall{x}{\big(\neighborhood[\delta]{x=0} \limply \dbox{\pumod{y}{\alpha^2\norm{x}^2} ;\D{x}=\genDE{x},\D{y}=-2\beta y}{\,\norm{x}^2 \leq y} \big)}}}
\end{align*}
The discrete assignment $\pumod{y}{\alpha^2\norm{x}^2}$ sets the value of variable $y$ to that of $\alpha^2 \norm{x}^2$ and $;$ denotes sequential composition of hybrid programs~\cite{DBLP:journals/jar/Platzer17,Platzer18}.
\label{lem:expstabdl}
\end{lemma}

Formula $\expstabode{\D{x}=\genDE{x}}$ uses a fresh variable $y$ with ODE $\D{y}=-2\beta y$ and initialized to $\alpha^2 \norm{x}^2$ so that $y$ \emph{differentially axiomatizes}~\cite{DBLP:journals/jacm/PlatzerT20} the (squared) decaying exponential function $\alpha^2\norm{x(0)}^2\exp{(-2\beta t)}$ along ODE solutions.
Such an implicit (polynomial) characterization of exponential decay allows syntactic proof steps to use decidable real arithmetic reasoning.

\begin{lemma}[Lyapunov function for exponential stability] %
The following Lyapunov function proof rule for exponential stability is derivable in \dL, where $k_1, k_2, k_3 \in \rationals$ are positive constants. %

\noindent
\begin{calculus}
\dinferenceRule[ExpLyap|Lyap$_{\text{E}}$]{Exponential Lyapunov}
{\linferenceRule
  {\lsequent{} {\lexists{\gamma{>}0}{\lforall{x}{\big(\norm{x}^2 {\leq} \gamma^2 \limply  k_1^2 \norm{x}^2 \leq \lterm \leq k_2^2 \norm{x}^2 \land \lied[]{\genDE{x}}{\lterm} \leq -2 k_3 \lterm)}} }}
  {\lsequent{} {\expstabode{\D{x}=\genDE{x}}}}
}{}
\end{calculus}
\label{lem:expstablyap}
\end{lemma}

Rule~\irref{ExpLyap} enables proofs of exponential stability in \dL.
\iflongversion
In fact, the proof of~\rref{lem:expstablyap} (\rref{app:proofs}) yields \emph{quantitative} bounds, where $\expstabode{\D{x}=\genDE{x}}$ is explicitly witnessed with scaling constant $\alpha \mnodefeq \frac{k_2}{k_1}$ and decay rate $\beta \mnodefeq k_3$.
\fi%
These can be used to calculate time bounds when the system state will return sufficiently close to the origin.
Similarly, the disturbance $\delta$ in $\expstabode{\D{x}=\genDE{x}}$ is quantitatively witnessed by $\frac{k_1}{k_2}\gamma$ for any $\gamma$ witnessing validity of the premise of rule~\irref{ExpLyap}.
This yields a provable estimate of the region around the origin where exponential stability holds; this latter estimate is explored next.

\paragraph{Region of Attraction.}
Formulas $\attrode{\D{x}=\genDE{x}}$ and $\expstabode{\D{x}=\genDE{x}}$ both feature a subformula of the form $\lexists{\delta > 0}{\lforall{x}{(\neighborhood[\delta]{x=0} \limply \cdots)}}$ which expresses that attractivity (or exponential stability) is locally true in \emph{some} $\delta$ neighborhood of the origin.
In many applications, it is useful to find and rigorously prove that a given set is attractive or exponentially stable with respect to the origin~\cite[Chapter 8.2]{MR1201326}.
The second stability variation yields \emph{provable} subsets of the region of attraction, including the special case where it is the entire state space.
This is formalized using the following variants of $\attrode{\D{x}=\genDE{x}}$ and $\expstabode{\D{x}=\genDE{x}}$ within a region given by a formula $\rfvar$.
\begin{align*}
\attrodeP{\D{x}=\genDE{x}}{\rfvar} &\mnodefequiv \lforall{x}{\big(\rfvar \limply \asymode{\D{x}=\genDE{x}}{x=0}\big)} \\
\expstabodeP{\D{x}=\genDE{x}}{\rfvar} &\mnodefequiv \lexists{\alpha{>}0}{\lexists{\beta {>} 0}{\lforall{x}{\big(\rfvar \limply \dbox{\pumod{y}{\alpha^2\norm{x}^2} ; \D{x}=\genDE{x},\D{y}=-2\beta y}{\,\norm{x}^2 \leq y}\big)}}}
\end{align*}

The formula $\attrodeP{\D{x}=\genDE{x}}{\rfvar}$ is valid iff the set characterized by $\rfvar$ is a subset of the origin's region of attraction~\cite{MR1201326}.
For example, $\attrode{\D{x}=\genDE{x}}$ is $\lexists{\delta>0}{\attrodeP{\D{x}=\genDE{x}}{\neighborhood[\delta]{x=0} }}$.
This generalization is useful for formalizing stronger notions of stability in \dL, such as the following \emph{global} stability notions~\cite{10.2307/j.ctvcm4hws,MR1201326}.
\iflongversion
For brevity, \dL specifications of the stability properties (in \textbf{bold}) are given below with mathematical definitions deferred to~\rref{app:proofs}.
\fi

\begin{lemma}[Global stability in \dL]
The origin of ODE $\D{x}=\genDE{x}$ is \textbf{globally asymptotically stable} iff the \dL formula $\stabode{\D{x}=\genDE{x}} \land \attrodeP{\D{x}=\genDE{x}}{\ltrue}$ is valid.
The origin is \textbf{globally exponentially stable} iff the \dL formula $\expstabodeP{\D{x}=\genDE{x}}{\ltrue}$ is valid.
\label{lem:globstabdl}
\end{lemma}

Global stability ensures that \emph{all} perturbations to the system state are eventually dissipated.
Their proof rules are similar to~\irref{StrictLyap} and~\irref{ExpLyap} respectively.

\begin{lemma}[Lyapunov function for global stability] %
The following Lyapunov function proof rules for global asymptotic and exponential stability are derivable in \dL.
In rule~\irref{ExpLyapGlob}, $k_1, k_2, k_3 \in \rationals$ are positive constants.

\noindent
\begin{calculus}
\dinferenceRule[StrictLyapGlob|Lyap$_{>}^{\text{G}}$]{Strict Lyapunov Global}
{\linferenceRule
  {\lsequent{} {\genDE{0}{=}0 {\land} \lterm(0) {=} 0} \quad
   \lsequent{x {\neq} 0}{\lterm {>} 0 \land \lied[]{\genDE{x}}{\lterm} {<} 0} \quad
   \lsequent{}{\lforall{b}{\lexists{\gamma{>}0}{\lforall{x}{\big( \lterm {\leq} b {\limply} \neighborhood[\gamma]{x{=}0} \big)}}}}
   }
  {\lsequent{} {\stabode{\D{x}=\genDE{x}} \land \attrodeP{\D{x}=\genDE{x}}{\ltrue} }}
}{}

\dinferenceRule[ExpLyapGlob|Lyap$_{\text{E}}^{\text{G}}$]{Exponential Lyapunov Global}
{\linferenceRule
  {\lsequent{} {k_1^2 \norm{x}^2 \leq \lterm \leq k_2^2 \norm{x}^2 \land \lied[]{\genDE{x}}{\lterm} \leq -2 k_3 \lterm} }
  {\lsequent{} {\expstabodeP{\D{x}=\genDE{x}}{\ltrue} }}
}{}
\end{calculus}

\label{lem:globstablyap}
\end{lemma}

\begin{example}[Pendulum global exponential stability]
\label{ex:esympendulum}
For simplicity, instantiate \rref{ex:asympendulum} with parameters $a=1, b=1$.
The Lyapunov function then simplifies to $\lterm \mnodefeq \frac{\theta^2}{2} + \frac{(\theta + \omega)^2 + \omega^2}{4}$ with Lie derivative $\lied[]{\expendulum}{\lterm} \mnodefeq -\frac{(\theta^2+\omega^2)}{2}$, which satisfies the real arithmetic inequalities $\frac{\theta^2 + \omega^2}{4} \leq \lterm \leq \theta^2 + \omega^2$ and $\lied[]{\expendulum}{\lterm} \leq -\frac{1}{2} \lterm$.
Thus, rule~\irref{ExpLyapGlob} proves global exponential stability of $\expendulum$ with $k_1 =  \frac{1}{2}$, $k_2 = 1$, and $k_3 = \frac{1}{4}$.
An important caveat is that~\rref{ex:asympendulum} used a local small angle approximation, so this global phenomenon does \emph{not} hold for a real world pendulum (nor for $\expendulumnonlin$).
\end{example}

\paragraph{Logical Relationships.}
With the proliferation of stability variations just introduced, it is useful to take stock of their logical relationships.
An important example of such a relationship is shown in the following corollary.

\begin{corollary}[Exponential stability implies asymptotic stability]
The following axioms are derivable in \dL.

\noindent
\begin{calculuscollection}
\begin{calculus}
\dinferenceRule[EStabStab|EStabStab]{}
{
\linferenceRule[impl]
  {\expstabode{\D{x}=\genDE{x}}}
  {\stabode{\D{x}=\genDE{x}}}
}{}

\dinferenceRule[EStabAttr|EStabAttr]{}
{
\linferenceRule[impl]
  {\expstabodeP{\D{x}=\genDE{x}}{\rfvar} }
  {\attrodeP{\D{x}=\genDE{x}}{\rfvar} }
}{}
\end{calculus}
\end{calculuscollection}
\label{cor:expimpasym}
\end{corollary}

Derived axioms~\irref{EStabStab+EStabAttr} show that exponential stability implies asymptotic stability.
In proofs, axiom~\irref{EStabAttr} allows the region of attraction to be estimated by the region where solutions are exponentially bounded.

\section{General Stability}
\label{sec:genstability}
This section provides stability definitions and proof rules that generalize stability for an equilibrium point from~\rref{sec:asymstability} to the stability of sets.
These definitions are useful when the desired stable system state(s) is not modeled by a single equilibrium point, but may instead, e.g., lie on a periodic trajectory~\cite{MR1201326}, a hyperplane, or a continuum of equilibrium points within the state space~\cite{10.2307/j.ctvcm4hws}.
The generalized definition is used to formalize two stability notions from the literature~\cite{10.2307/j.ctvcm4hws,MR1201326}, and to justify their Lyapunov function proof rules.

\subsection{General Stability and General Attractivity}
The following \emph{general stability} formula defines stability in \dL with respect to an ODE $\D{x}=\genDE{x}$ and formulas $\rfvar, \rrfvar$.
The quantified variables $\varepsilon, \delta$ are assumed to be fresh by bound renaming, i.e., do not appear in $x, f(x), \rfvar$ or $\rrfvar$.
\[
  \stabodePR{\D{x}=\genDE{x}}{\rfvar}{\rrfvar} \mnodefequiv \lforall{\varepsilon {>} 0} {\lexists{\delta {>} 0}{ \lforall{x}{\big( \neighborhood[\delta]{\rfvar} \limply \dbox{\D{x}=\genDE{x}}{\,\neighborhood[\varepsilon]{\rrfvar}}\big)}}}
\]

This formula generalizes stability of the origin $\stabode{\D{x}=\genDE{x}}$ by adding two logical tuning knobs that can be intuitively understood as follows.
The \emph{precondition} $\rfvar$ characterizes the initial states from which the system state is expected to be disturbed by some disturbance $\delta$.
The \emph{postcondition} $\rrfvar$ characterizes the set of desired operating states that the system must remain close (within the $\varepsilon$ neighborhood of $\rrfvar$) after being disturbed from its initial states.

The \emph{general attractivity} formula similarly generalizes $\attrodeP{\D{x}=\genDE{x}}{\rfvar}$ with a postcondition $\rrfvar$ towards which the ODE solutions from initial states satisfying precondition $\rfvar$ are asymptotically attracted.
\[
  \attrodePR{\D{x}=\genDE{x}}{\rfvar}{\rrfvar} \mnodefequiv \lforall{x}{\big(\rfvar \limply \asymode{\D{x}=\genDE{x}}{\rrfvar}\big)}
\]

\begin{lemma}[General Lyapunov functions]
The following Lyapunov function proof rule for general stability with two stacked premises is derivable in \dL.

\noindent
\begin{calculus}
\dinferenceRule[LyapGen|GLyap]{Lyapunov General}
{\linferenceRule
  { \begin{array}{l}
    \lsequent{}{\rfvar \limply \rrfvar}\\
    \lsequent{}{\lforall{\varepsilon{>}0}{
    \lexists{0{<}\gamma{\leq}\varepsilon}{
    \lexists{k}{\left(
     \begin{array}{l}
     \lforall{x}{(\bdr{(\neighborhood[\gamma]{\rrfvar})} \limply \lterm \geq k)} \land \\
     \lexists{0{<}\delta{\leq}\gamma}{\lforall{x}{(\neighborhood[\delta]{\rfvar}  \limply \rrfvar \lor \lterm{<}k)} } \land \\
     \lforall{x}{\big(\rrfvar \lor \lterm {<} k \limply \dbox{\pevolvein{\D{x}=\genDE{x}}{\cneighborhood[\gamma]{\rrfvar}}}{(\rrfvar \lor \lterm {<} k)}\big)} \\
     \end{array}\right)}
    }}}
    \end{array}
  }
  {\lsequent{} {\stabodePR{\D{x}=\genDE{x}}{\rfvar}{\rrfvar}}}
}{}
\end{calculus}
\label{lem:genlyapunov}
\end{lemma}

Rule~\irref{LyapGen} proves general stability for precondition $\rfvar$ and postcondition $\rrfvar$.
It generalizes the Lyapunov function reasoning underlying rule~\irref{Lyap} to support arbitrary pre- and postconditions.
The conjunct \(\lforall{x}{(\bdr{(\neighborhood[\gamma]{\rrfvar})} \limply \lterm \geq k)}\) requires $v{\geq}k$ on the boundary of $\neighborhood[\gamma]{\rrfvar}$ while the middle conjunct requires $v {<} k$ for some small neighborhood of $\rfvar$ excluding $\rrfvar$.
The conjunct $\lforall{x}{\big(\rrfvar \lor \lterm {<} k \limply \cdots \big)}$ asserts that $\rrfvar \lor v < k$ is an invariant of the ODE \emph{within} closed domain $\cneighborhood[\gamma]{\rrfvar}$.
When $\rrfvar$ is a formula of real arithmetic, this invariance question is provably equivalent in \dL to a formula of real arithmetic~\cite{DBLP:journals/jacm/PlatzerT20}, so the premise of rule~\irref{LyapGen} is, \emph{in theory}, decidable by \irref{qear} for a candidate Lyapunov function $v$.
In practice, it is prudent to consider specialized stability notions, for which the premise of rule~\irref{LyapGen} can be arithmetically simplified.
Proof rules for generalized attractivity are also derivable for specialized instances.

\subsection{Specialization}
General stability specializes to several stability notions in the literature.
\iflongversion
For brevity, \dL specifications of the stability properties (in \textbf{bold}) are given below with mathematical definitions deferred to~\rref{app:proofs}.
\fi

\paragraph{Set Stability.}

An important special case is when the desired operating states are exactly the states from which disturbances are expected, i.e., $\rrfvar \mnodefequiv \rfvar$.
This leads to the notion of \textbf{set stability} of the set characterized by $\rfvar$~\cite{10.2307/j.ctvcm4hws,MR1201326}.

\begin{lemma}[Set Stability in \dL]
For the ODE $\D{x}=\genDE{x}$, the set characterized by formula $\rfvar$ is
\begin{inparaenum}[\it i)]
\item \textbf{stable}, \item \textbf{attractive}, \item \textbf{asymptotically stable}, and \item \textbf{globally asymptotically stable}
\end{inparaenum} iff the following \dL formulas are valid:
\begin{enumerate}[\it i)]
\item $\stabodePR{\D{x}=\genDE{x}}{\rfvar}{\rfvar}$,
\item $\lexists{\delta{>}0}{\attrodePR{\D{x}=\genDE{x}}{\neighborhood[\delta]{\rfvar}}{\rfvar}}$,
\item $\stabodePR{\D{x}=\genDE{x}}{\rfvar}{\rfvar} \land \lexists{\delta{>}0}{\attrodePR{\D{x}=\genDE{x}}{\neighborhood[\delta]{\rfvar}}{\rfvar}}$, and
\item $\stabodePR{\D{x}=\genDE{x}}{\rfvar}{\rfvar} \land \attrodePR{\D{x}=\genDE{x}}{\ltrue}{\rfvar}$
\end{enumerate}
\label{lem:setasymstabdl}
\end{lemma}

The intuition for~\rref{lem:setasymstabdl} is similar to Lemmas~\ref{lem:asymstabdl} and~\ref{lem:globstabdl}, except formula $\rfvar$ (instead of the origin) characterizes the set of desirable states.
An application of set stability is shown in the following example.

\begin{example}[Tennis racket theorem~\cite{Ashbaugh91}]
\label{ex:tennisracket}
The following system of ODEs models the rotation of a 3D rigid body~\cite{Chicone2006,10.2307/j.ctvcm4hws}, where $x_1,x_2,x_3$ are angular velocities and $I_1 > I_2 > I_3 > 0$ are the principal moments of inertia along the respective axes.
\[ \exrigid \mnodefequiv \D{x_1} = \frac{I_2-I_3}{I_1} x_2 x_3, \quad \D{x_2} = \frac{I_3-I_1}{I_2} x_3 x_1, \quad \D{x_3} = \frac{I_1-I_2}{I_3} x_1 x_2 \]

When such a rigid object is spun or rotated on each of its axes, a well-known physical curiosity~\cite{Ashbaugh91} is that the rotation is stable in the first and third axes, whilst additional (unstable) twisting motion is observed for the intermediate axis.
Mathematically, a perfect rotation, e.g., around $x_1$, corresponds to a (large) initial value for $x_1$ with no rotation in the other axes, i.e., $x_2 = 0$, $x_3 = 0$.
Accordingly the real world observation of stability for rotations about the first principal axis is explained by stability with respect to small perturbations in $x_2,x_3$, as formally specified by formula~\rref{eq:stab1} below.
Note that the set characterized by formula $x_2 = 0 \land x_3 = 0$ is the entire $x_1$ axis, not just a single point.
Similarly, rotations are stable around the third principal axis iff formula~\rref{eq:stab3} is valid.
\begin{align}
\stabodePR{\exrigid}{x_2=0 \land x_3=0}{x_2=0 \land x_3=0} \label{eq:stab1}\\
\stabodePR{\exrigid}{x_1=0 \land x_2=0}{x_1=0 \land x_2=0} \label{eq:stab3}
\end{align}

The validity of formulas~\rref{eq:stab1} and~\rref{eq:stab3} are proved in~\rref{ex:tennisracketproof}.
\end{example}

The formal specification of set stability yields three provable logical consequences which are important stepping stones for the set stability proof rules.

\begin{corollary}[Set stability properties]
The following axioms are derivable in \dL.
In axiom~\irref{stabclosure}, formula $\closure{\rfvar}$ characterizes the topological closure of formula $\rfvar$.
In axiom~\irref{stabclosed}, formula $\rfvar$ characterizes a closed set.

\noindent
\begin{calculuscollection}
\begin{calculus}
\dinferenceRule[stabattrgen|SetSAttr]{}
{
\linferenceRule[impl]
  {\stabodePR{\D{x}=\genDE{x}}{\rfvar}{\rfvar} }
  {\big(\axkey{\asymode{\D{x}=\genDE{x}}{\rfvar}} \lbisubjunct \lforall{\varepsilon{>}0}{\ddiamond{\D{x}=\genDE{x}}{\,\neighborhood[\varepsilon]{\rfvar}}}\big)}
}{}

\dinferenceRule[stabclosure|SClosure]{}
{
\linferenceRule[equiv]
  {\stabodePR{\D{x}=\genDE{x}}{\closure{\rfvar}}{\closure{\rfvar}}}
  {\stabodePR{\D{x}=\genDE{x}}{\rfvar}{\rfvar} }
}{}

\dinferenceRule[stabclosed|SClosed]{}
{
\linferenceRule[impl]
  {\stabodePR{\D{x}=\genDE{x}}{\rfvar}{\rfvar} }
  {\lforall{x}{\big( \rfvar \limply \dbox{\D{x}=\genDE{x}}{\rfvar}\big)}}
}{}
\end{calculus}
\end{calculuscollection}
\label{cor:setstabsimp}
\end{corollary}

Axiom~\irref{stabattrgen} generalizes~\irref{stabattr} and provides a syntactic simplification of the region of attraction for formula $\rfvar$ when $\rfvar$ is stable.
Axiom~\irref{stabclosure} says that stability of $\rfvar$ is equivalent to stability of its closure $\closure{\rfvar}$, because for any perturbation $\delta > 0$, the neighborhoods $\neighborhood[\delta]{\rfvar}$ and $\neighborhood[\delta]{\closure{\rfvar}}$ are provably equivalent in real arithmetic.
Axiom~\irref{stabclosed} says that for closed formulas $\rfvar$, invariance of $\rfvar$ is a necessary condition for stability of $\rfvar$.
Without loss of generality, it suffices to develop proof rules for stability of formulas characterizing closed (using \irref{stabclosure}) and invariant (using \irref{stabclosed}) sets.
Indeed, standard definitions of set stability~\cite{10.2307/j.ctvcm4hws,MR1201326} usually assume that the set of concern is closed and invariant.

\begin{lemma}[Set stability Lyapunov functions]
The following Lyapunov function proof rules for set stability are derivable in \dL.
In rules~\irref{SetLyap+SetStrictLyap}, formula $\rfvar$ characterizes a compact (i.e., closed and bounded) set.
In rule~\irref{SetLyapGen}, the two premises are stacked.

\noindent
\begin{calculus}
\dinferenceRule[SetLyap|SLyap${_\geq}$]{Set Lyapunov}
{\linferenceRule
  {\lsequent{\rfvar } {\dbox{\D{x}=\genDE{x}}{\rfvar}} \quad
   \lsequent{\lnot{\rfvar}}{ \lterm > 0 \land \lied[]{\genDE{x}}{\lterm} \leq 0} \qquad
   \lsequent{\bdr{\rfvar}}{ \lterm \leq 0}}
  {\lsequent{} {\stabodePR{\D{x}=\genDE{x}}{\rfvar}{\rfvar} }}
}{}
\dinferenceRule[SetStrictLyap|SLyap${_>}$]{Set Strict Lyapunov}
{\linferenceRule
  {\lsequent{\rfvar } {\dbox{\D{x}=\genDE{x}}{\rfvar}} \quad
   \lsequent{\lnot{\rfvar}}{ \lterm > 0 \land \lied[]{\genDE{x}}{\lterm} < 0} \qquad
   \lsequent{\bdr{\rfvar}}{ \lterm \leq 0}  }
  {\lsequent{} {\stabodePR{\D{x}=\genDE{x}}{\rfvar}{\rfvar} \land \lexists{\delta{>}0}{\attrodePR{\D{x}=\genDE{x}}{\neighborhood[\delta]{\rfvar}}{\rfvar}}}}
}{}
\end{calculus}

\noindent
\begin{calculus}
\dinferenceRule[SetLyapGen|SLyap${^*_\geq}$]{Set Lyapunov General}
{\linferenceRule
  { \begin{array}{l}
    \lsequent{\rfvar } {\dbox{\D{x}=\genDE{x}}{\rfvar}} \\
    \lsequent{}{\lforall{\varepsilon{>}0}{
    \lexists{0{<}\gamma{\leq}\varepsilon}{
   \left(\begin{array}{l}
   \lexists{k}{\left(
     \begin{array}{l}
     \lforall{x}{(\bdr{(\neighborhood[\gamma]{\rfvar})} \limply \lterm \geq k)} \land \\
     \lexists{0{<}\delta{\leq}\gamma}{\lforall{x}{(\neighborhood[\delta]{\rfvar} \land \lnot{\rfvar} \limply \lterm <k)} }
     \end{array}\right)} \land\\
   \lforall{x}{(\cneighborhood[\gamma]{\rfvar} \land \lnot{\rfvar} \limply \lied[]{\genDE{x}}{\lterm} \leq 0)}
    \end{array}\right)
    }}}
    \end{array}
  }
  {\lsequent{} {\stabodePR{\D{x}=\genDE{x}}{\rfvar}{\rfvar}}}
}{}
\end{calculus}
\label{lem:setstablyap}
\end{lemma}

All three proof rules have the necessary premise $\lsequent{\rfvar} {\dbox{\D{x}=\genDE{x}}{\rfvar}}$ which says that formula $\rfvar$ is an invariant of the ODE $\D{x}=\genDE{x}$.
Rules~\irref{SetLyap+SetStrictLyap} are slight generalizations of Lyapunov function proof rules for set stability~\cite{10.2307/j.ctvcm4hws} and they respectively generalize rules~\irref{Lyap+StrictLyap} to prove stability for an invariant $\rfvar$.
Importantly, both rules assume that $\rfvar$ characterizes a compact, i.e., closed and bounded set, which simplifies the arithmetical conditions on $\lterm$ in their premises.
The rule \emph{without} the boundedness requirement on $\rfvar$ suggested in the remark after~\cite[Definition 8.1]{MR1201326}, is unsound,
\iflongversion
see~\rref{cex:khalil}.
\fi

For asymptotic stability (in rule \irref{SetStrictLyap}), boundedness also guarantees that perturbed ODE solutions always exist for sufficient duration, which is a fundamental step in the ODE liveness proofs~\cite{DBLP:journals/fac/TanP}.
Rule~\irref{SetLyapGen} is derived from rule \irref{LyapGen} using invariance of $\rfvar$ by the first premise; it provides a means of formally proving the set stability properties~\rref{eq:stab1} and~\rref{eq:stab3} from~\rref{ex:tennisracket}.

\begin{example}[Stability of rigid body motion]
\label{ex:tennisracketproof}
The proof for~\rref{eq:stab1} uses the Lyapunov function $\lterm \mnodefeq \frac{1}{2}(\frac{I_1-I_2}{I_3}x_2^2 - \frac{I_3-I_1}{I_2} x_3^2)$, whose Lie derivative is $\lied[]{\genDE{x}}{\lterm} = 0$, and rule~\irref{SetLyapGen} with formula $\rfvar \mnodefequiv x_2=0 \land x_3=0$.
The proof for~\rref{eq:stab3} is symmetric.
For the top premise of rule~\irref{SetLyapGen}, formula $\rfvar$ is a provable invariant~\cite{DBLP:journals/jacm/PlatzerT20} of the ODE $\exrigid$.
The bottom premise, although arithmetically complicated, can be simplified by choosing $\gamma \mnodefeq \varepsilon$ and deciding the resulting formula by~\irref{qear}.

Recall that the $x_1$ axis is \emph{not} a compact set so neither of the standard proof rules for set stability~\irref{SetLyap+SetStrictLyap} would be sound for this proof.
\end{example}

\paragraph{Epsilon-Stability}
Motivated by numerical robustness of proofs of stability, Gao et al.~\cite{DBLP:conf/cav/GaoKDRSAK19} define $\varepsilon$-stability for ODEs.
The following \dL characterization shows how $\varepsilon$-stability can be understood as an instance of general stability.

\begin{lemma}[$\varepsilon$-Stability in \dL]
The origin of ODE $\D{x}=\genDE{x}$ is \textbf{$\bm{\varepsilon}$-stable} for constant $\varepsilon > 0$ iff the \dL formula $\stabodePR{\D{x}=\genDE{x}}{x=0}{\neighborhood[\varepsilon]{x=0}}$ is valid.
\label{lem:epsstab}
\end{lemma}

Unlike set stability, $\varepsilon$-stability is an instance of general stability where the pre- and postconditions differ.
In $\varepsilon$-stability, systems are perturbed from the precondition $x=0$ (the origin), but the postcondition enlarges the set of desired states to a $\varepsilon > 0$ neighborhood of the origin, which is considered indistinguishable from the origin itself~\cite{DBLP:conf/cav/GaoKDRSAK19}.
An immediate consequence of~\rref{lem:epsstab} is that rule~\irref{LyapGen} can be used to prove $\varepsilon$-stability, as shown in the next section.

\section{Stability in \KeYmaeraX}
\label{sec:casestudies}

This section puts the \dL stability specifications and derivations from the preceding sections into practice through proofs for several case studies in the \KeYmaeraX theorem prover~\cite{DBLP:conf/cade/FultonMQVP15}.\footnote{See \url{https://github.com/LS-Lab/KeYmaeraX-projects/blob/master/stability}}
Examples~\ref{ex:asympendulum},~\ref{ex:esympendulum},~\ref{ex:tennisracket},~\ref{ex:tennisracketproof} have also been formalized.
The insights from these proofs are discussed after an overview of the case studies.

\paragraph{Inverted Pendulum.}
The stability of the resting state of the pendulum is investigated in Examples~\ref{ex:asympendulum} and~\ref{ex:esympendulum}.
For the inverted pendulum $\expendulumforced$ from~\rref{eq:pendulumforced}, the controlled torque $u(\theta,\omega)$ must be designed and rigorously proved to ensure \emph{feedback stabilization}~\cite{MR1201326} of the inverted position.
A standard PD (Proportional-Derivative) controller can be used for stabilization, where the control input has the form $u(\theta,\omega) \mnodefeq k_1 \theta + k_2 \omega$ for tuning parameters $k_1, k_2$.
Asymptotic stability of the inverted position is achieved for any control parameter choice where $k_1 > a$ and $k_2 > -b$.
The sequent $\lsequent{a > 0, b \geq 0, k_1 > a, k_2 > -b}{\astabode{\expendulumforced}}$ is proved in \KeYmaeraX using the Lyapunov function $\frac{(k_1-a)\theta^2}{2} + \frac{(((b+k_2)\theta+\omega)^2+\omega^2)}{4}$.

\paragraph{Frictional Tennis Racket Theorem.}
The stability of a 3D rigid body is investigated for $\exrigid$ in Examples~\ref{ex:tennisracket} and~\ref{ex:tennisracketproof}.
The following ODEs model additional frictional forces that oppose the rotational motion in each axis of the rigid body, where $\alpha_1,\alpha_2,\alpha_3 > 0$ are positive coefficients of friction:
\[ \exrigidfriction \mnodefequiv \D{x_1} {=} \frac{I_2-I_3}{I_1} x_2 x_3 - \alpha_1 x_1,~\D{x_2} {=} \frac{I_3-I_1}{I_2} x_3 x_1- \alpha_2 x_2,~\D{x_3} {=} \frac{I_1-I_2}{I_3} x_1 x_2 - \alpha_3 x_3\]

In the presence of friction, rotations of the rigid body are globally asymptotically stable in the first and third principal axes, as proved in \KeYmaeraX.
\begin{align*}
&\Gamma \mnodefequiv I_1 > I_2, I_2 > I_3, I_3 > 0, \alpha_1 > 0, \alpha_2 > 0, \alpha_3 > 0 \\
&\lsequent{\Gamma}{\stabodePR{\exrigidfriction}{x_2{=}0 \land x_3{=}0}{x_2{=}0 \land x_3{=}0} \land \attrodePR{\exrigidfriction}{\ltrue}{x_2{=}0 \land x_3{=}0}} \\
&\lsequent{\Gamma}{\stabodePR{\exrigidfriction}{x_1{=}0 \land x_2{=}0}{x_1{=}0 \land x_2{=}0} \land \attrodePR{\exrigidfriction}{\ltrue}{x_1{=}0 \land x_2{=}0}}
\end{align*}

Both asymptotic stability properties are proved using~\irref{SetLyapGen} and the liveness property~\cite{DBLP:journals/fac/TanP} that the kinetic energy $I_1 x_1^2 + I_2 x_2^2 + I_3 x_3^2$ of the system tends to zero over time.
The latter property implies that solutions of $\exrigidfriction$ exist globally and that the values of $x_1,x_2,x_3$ asymptotically tend to zero, which proves global asymptotic stability with the aid of~\irref{stabattrgen}.
Even though a proof rule for (global) asymptotic stability of general nonlinear ODEs and unbounded sets is not available (\rref{sec:genstability}), this example shows that formalized stability properties can still be proved on a case-by-case basis using \dL's ODE reasoning principles.

\paragraph{Moore-Greitzer Jet Engine~\cite{DBLP:conf/cav/GaoKDRSAK19}.}
The origin of the ODE modeling a simplified jet engine $\exmoore \mnodefequiv \D{x_1}=-x_2-\frac{3}{2}x_1^2-\frac{1}{2}x_1^3,~\D{x_2} = 3x_1-x_2$ is $\varepsilon$-stable for $\varepsilon = 10^{-10}$~\cite{DBLP:conf/cav/GaoKDRSAK19}.
\iflongversion
The following sequent is proved in \KeYmaeraX: $\lsequent{\varepsilon = 10^{-10}}{\stabodePR{\exmoore}{x_1^2+x_2^2=0}{x_1^2+x_2^2 < \varepsilon^2}}$.
\fi
The key proof ingredients are an $\varepsilon$-Lyapunov function~\cite{DBLP:conf/cav/GaoKDRSAK19} and manual arithmetic steps, e.g., instantiating existential quantifiers appearing in the specification of $\varepsilon$-stability with appropriate values~\cite{DBLP:conf/cav/GaoKDRSAK19}.

\paragraph{Other Examples~\cite{DBLP:conf/tacas/AhmedPA20}.}
Stability for several ODEs with Lyapunov functions generated by an inductive synthesis technique~\cite[Examples 5--11]{DBLP:conf/tacas/AhmedPA20} were successfully verified in \KeYmaeraX.
The proof for the largest, 6-dim.~nonlinear ODE~\cite[Example 5]{DBLP:conf/tacas/AhmedPA20} required substantial manual arithmetic reasoning in \KeYmaeraX.\footnote{The Lyapunov function in~\cite[Example 5]{DBLP:conf/tacas/AhmedPA20} does \emph{not} work for its associated ODE. It works if the ODE is corrected with $\dot{x}_1 = -x_1^3+4x_2^3-6x_3x_4$, as in the literature~\cite{1184414}.}

The arithmetical conditions in~\cite[Equation 1]{DBLP:conf/tacas/AhmedPA20} are identical to the premises of rule~\irref{Lyap} except it unsoundly omits the condition $\lterm(0)=0$,
\iflongversion
see~\rref{cex:ahmed}.
\fi
The generated Lyapunov functions remain correct because the inductive synthesis technique~\cite{DBLP:conf/tacas/AhmedPA20} implicitly guarantees this omitted condition.

\paragraph{Summary.}
These case studies demonstrate the feasibility of carrying out proofs of various (advanced) stability properties within \KeYmaeraX using this paper's stability specifications.
The proofs share similar high-level proof structure, which suggests that proof automation could significantly reduce proof effort~\cite{DBLP:conf/itp/FultonMBP17}.
Such automation should also support user input of key insights for difficult reasoning steps, e.g., real arithmetic reasoning with nested, alternating quantifiers.

\section{Related Work}
\label{sec:related}

Stability is a fundamental property of interest across many different fields of mathematics~\cite{Chicone2006,hirsch1984,Liapounoff1907,Poincare92,MR0450715,MR3837141} and engineering~\cite{10.2307/j.ctvcm4hws,MR1201326,DBLP:books/sp/Liberzon03}.
This related work discussion focuses on formal approaches to stability of ODEs.

\paragraph{Logical Specification of Stability.}
Rouche, Habets, and Laloy~\cite{MR0450715} provide a pioneering example of using logical notation to specify and classify stability properties of ODEs.
Alternative logical frameworks have also been used to specify stability and related properties:
stability is expressed in HyperSTL~\cite{DBLP:conf/memocode/NguyenKJDJ17} as a hyperproperty relating the trace of an ODE against two constant traces;
$\epsilon$-stability is studied in the context of $\delta$-complete reasoning over the reals~\cite{DBLP:conf/cav/GaoKDRSAK19};
region stability for hybrid systems~\cite{DBLP:conf/hybrid/PodelskiW06} is discussed using CTL*;
the syntactic specification of $\asymode{\D{x}=\genDE{x}}{\rfvar}$ resembles the limit definition using filters~\cite{DBLP:conf/itp/HolzlIH13}.
This paper uses \dL as a \emph{sweet spot} logical framework, general enough to specify various stability properties of interest, e.g., asymptotic or exponential stability, and the stability of sets, while also enabling syntactic proofs of those properties.

\paragraph{Formal Verification of Stability.}
There is a vast literature on finding Lyapunov functions for stability, e.g., through numerical~\cite{Parrilo00,1184414,DBLP:journals/automatica/TopcuPS08} and algebraic methods~\cite{261424,DBLP:journals/mics/LiuZZ12}.
Formal approaches are often based on finding Lyapunov function candidates and \emph{certifying} the correctness of those generated candidates~\cite{DBLP:conf/tacas/AhmedPA20,DBLP:conf/cav/GaoKDRSAK19,DBLP:conf/hybrid/KapinskiDSA14,DBLP:conf/nolcos/Sankaranarayanan0A13}.
This paper's approach enables highly trustworthy certification of those candidates in \dL and \KeYmaeraX, with stability proof rules that are soundly \emph{derived} from \dL's parsimonious axiomatization~\cite{DBLP:conf/lics/Platzer12a,DBLP:journals/jar/Platzer17,Platzer18}, as implemented in \KeYmaeraX~\cite{DBLP:conf/cade/FultonMQVP15,DBLP:journals/jar/Platzer17}.
Sections~\ref{sec:genstability} and~\ref{sec:casestudies} further show that this paper's approach supports verification of advanced stability properties~\cite{DBLP:conf/cav/GaoKDRSAK19,10.2307/j.ctvcm4hws,MR1201326} within the same \dL framework.
New stability proof rules like~\irref{LyapGen} can also be soundly and \emph{syntactically} justified in \dL without the need for (low-level) semantic reasoning about the underlying ODE mathematics.
As an example of the latter, semantic approach, LaSalle's invariance principle is formalized in Coq~\cite{DBLP:conf/itp/CohenR17} and used to verify the correctness of an inverted pendulum controller~\cite{DBLP:conf/cpp/Rouhling18}.

\section{Conclusion}

This paper shows how ODE stability can be formalized in \dL using the key idea that stability properties are $\lforall{}{}/\lexists{}{}$-quantified dynamical formulas.
These specifications, their proof rules, and their logical relationships are all syntactically derived from \dL's sound proof calculus.
This further enables trustworthy \KeYmaeraX proofs that rigorously verify \emph{every step} in an ODE stability argument, from arithmetical premises down to dynamical reasoning for ODEs.
Directions for future work include \begin{inparaenum}[\it i)]
\item formalization of stability with respect to perturbations of the system dynamics, and
\item generalizations of stability to hybrid systems.
\end{inparaenum}

\paragraph{Acknowledgments.}
We thank Brandon Bohrer, Stefan Mitsch, and the anonymous reviewers for their helpful feedback on \KeYmaeraX and this paper.

The first author was supported by A*STAR, Singapore.
This research was sponsored by the AFOSR under grant number FA9550-16-1-0288.  The views and conclusions contained in this document are those of the author and should not be interpreted as representing the official policies, either expressed or implied, of any sponsoring institution, the U.S. government or any other entity.

\bibliographystyle{halpha}
\bibliography{paper}

\newcommand{\etalchar}[1]{$^{#1}$}
\begin{thebibliography}{FMQ{\etalchar{+}}15}
\expandafter\ifx\csname url\endcsname\relax
  \def\url#1{\texttt{#1}}\fi
\expandafter\ifx\csname doi\endcsname\relax
  \def\doi#1{\burlalt{doi:#1}{http://dx.doi.org/#1}}\fi
\expandafter\ifx\csname urlprefix\endcsname\relax\def\urlprefix{URL }\fi
\expandafter\ifx\csname href\endcsname\relax
  \def\href#1#2{#2}\fi
\expandafter\ifx\csname burlalt\endcsname\relax
  \def\burlalt#1#2{\href{#2}{#1}}\fi

\bibitem[ACC91]{Ashbaugh91}
Mark~S. Ashbaugh, Carmen~C. Chicone, and Richard~H. Cushman.
\newblock The twisting tennis racket.
\newblock {\em Journal of Dynamics and Differential Equations}, 3:67--85, 1991.
\newblock \doi{10.1007/BF01049489}.

\bibitem[Alu15]{10.2307/j.ctt17kkb0d}
Rajeev Alur.
\newblock {\em Principles of Cyber-Physical Systems}.
\newblock MIT Press, 2015.

\bibitem[APA20]{DBLP:conf/tacas/AhmedPA20}
Daniele Ahmed, Andrea Peruffo, and Alessandro Abate.
\newblock Automated and sound synthesis of {Lyapunov} functions with {SMT}
  solvers.
\newblock In Armin Biere and David Parker, editors, {\em TACAS}, volume 12078
  of {\em LNCS}, pages 97--114. Springer, 2020.
\newblock \doi{10.1007/978-3-030-45190-5\_6}.

\bibitem[BCR98]{Bochnak1998}
Jacek Bochnak, Michel Coste, and Marie-Fran{\c{c}}oise Roy.
\newblock {\em Real Algebraic Geometry}.
\newblock Springer, Heidelberg, 1998.
\newblock \doi{10.1007/978-3-662-03718-8}.

\bibitem[Boh17]{DBLP:journals/afp/Bohrer17}
Brandon Bohrer.
\newblock Differential dynamic logic.
\newblock {\em Arch. Formal Proofs}, 2017, 2017.
\newblock
  \urlprefix\url{https://www.isa-afp.org/entries/Differential\_Dynamic\_Logic.shtml}.

\bibitem[Bra05]{DBLP:books/sp/necs2005/Branicky05}
Michael~S. Branicky.
\newblock Introduction to hybrid systems.
\newblock In Dimitrios Hristu{-}Varsakelis and William~S. Levine, editors, {\em
  Handbook of Networked and Embedded Control Systems}, pages 91--116.
  Birkh{\"{a}}user, 2005.
\newblock \doi{10.1007/0-8176-4404-0\_5}.

\bibitem[Chi06]{Chicone2006}
Carmen Chicone.
\newblock {\em Ordinary Differential Equations with Applications}.
\newblock Springer, New York, second edition, 2006.
\newblock \doi{10.1007/0-387-35794-7}.

\bibitem[CR17]{DBLP:conf/itp/CohenR17}
Cyril Cohen and Damien Rouhling.
\newblock A formal proof in {Coq} of {LaSalle}'s invariance principle.
\newblock In Mauricio Ayala{-}Rinc{\'{o}}n and C{\'{e}}sar~A. Mu{\~{n}}oz,
  editors, {\em ITP}, volume 10499 of {\em LNCS}, pages 148--163. Springer,
  2017.
\newblock \doi{10.1007/978-3-319-66107-0\_10}.

\bibitem[DFPP18]{DoyenFPP18}
Laurent Doyen, Goran Frehse, George~J. Pappas, and Andr{\'e} Platzer.
\newblock Verification of hybrid systems.
\newblock In Edmund~M. Clarke, Thomas~A. Henzinger, Helmut Veith, and Roderick
  Bloem, editors, {\em Handbook of Model Checking}, pages 1047--1110. Springer,
  Cham, 2018.
\newblock \doi{10.1007/978-3-319-10575-8\_30}.

\bibitem[FMBP17]{DBLP:conf/itp/FultonMBP17}
Nathan Fulton, Stefan Mitsch, Brandon Bohrer, and Andr{\'{e}} Platzer.
\newblock Bellerophon: Tactical theorem proving for hybrid systems.
\newblock In Mauricio Ayala{-}Rinc{\'{o}}n and C{\'{e}}sar~A. Mu{\~{n}}oz,
  editors, {\em ITP}, volume 10499 of {\em LNCS}, pages 207--224. Springer,
  2017.
\newblock \doi{10.1007/978-3-319-66107-0\_14}.

\bibitem[FMQ{\etalchar{+}}15]{DBLP:conf/cade/FultonMQVP15}
Nathan Fulton, Stefan Mitsch, Jan{-}David Quesel, Marcus V{\"{o}}lp, and
  Andr{\'{e}} Platzer.
\newblock {{KeYmaera}} {X:} an axiomatic tactical theorem prover for hybrid
  systems.
\newblock In Amy~P. Felty and Aart Middeldorp, editors, {\em CADE}, volume 9195
  of {\em LNCS}, pages 527--538, Cham, 2015. Springer.
\newblock \doi{10.1007/978-3-319-21401-6\_36}.

\bibitem[{For}91]{261424}
K.~{Forsman}.
\newblock Construction of {Lyapunov} functions using {Gr\"obner} bases.
\newblock In {\em CDC}, volume~1, pages 798--799. IEEE, 1991.
\newblock \doi{10.1109/CDC.1991.261424}.

\bibitem[GKD{\etalchar{+}}19]{DBLP:conf/cav/GaoKDRSAK19}
Sicun Gao, James Kapinski, Jyotirmoy~V. Deshmukh, Nima Roohi, Armando
  Solar{-}Lezama, Nikos Ar{\'{e}}chiga, and Soonho Kong.
\newblock Numerically-robust inductive proof rules for continuous dynamical
  systems.
\newblock In Isil Dillig and Serdar Tasiran, editors, {\em CAV}, volume 11562
  of {\em LNCS}, pages 137--154. Springer, 2019.
\newblock \doi{10.1007/978-3-030-25543-5\_9}.

\bibitem[GST12]{10.2307/j.ctt7s02z}
Rafal Goebel, Ricardo~G. Sanfelice, and Andrew~R. Teel.
\newblock {\em Hybrid Dynamical Systems: Modeling, Stability, and Robustness}.
\newblock Princeton University Press, 2012.

\bibitem[HC08]{10.2307/j.ctvcm4hws}
Wassim~M. Haddad and VijaySekhar Chellaboina.
\newblock {\em Nonlinear Dynamical Systems and Control: A Lyapunov-Based
  Approach}.
\newblock Princeton University Press, 2008.

\bibitem[HIH13]{DBLP:conf/itp/HolzlIH13}
Johannes H{\"{o}}lzl, Fabian Immler, and Brian Huffman.
\newblock Type classes and filters for mathematical analysis in {Isabelle/HOL}.
\newblock In Sandrine Blazy, Christine Paulin{-}Mohring, and David Pichardie,
  editors, {\em ITP}, volume 7998 of {\em LNCS}, pages 279--294. Springer,
  2013.
\newblock \doi{10.1007/978-3-642-39634-2\_21}.

\bibitem[Hir84]{hirsch1984}
Morris~W. Hirsch.
\newblock The dynamical systems approach to differential equations.
\newblock {\em Bull. Amer. Math. Soc. (N.S.)}, 11(1):1--64, 07 1984.

\bibitem[KDH{\etalchar{+}}20]{DBLP:conf/tacas/KolcakDHKS020}
Juraj Kolc{\'{a}}k, J{\'{e}}r{\'{e}}my Dubut, Ichiro Hasuo, Shin{-}ya
  Katsumata, David Sprunger, and Akihisa Yamada.
\newblock Relational differential dynamic logic.
\newblock In Armin Biere and David Parker, editors, {\em TACAS}, volume 12078
  of {\em LNCS}, pages 191--208. Springer, 2020.
\newblock \doi{10.1007/978-3-030-45190-5\_11}.

\bibitem[KDSA14]{DBLP:conf/hybrid/KapinskiDSA14}
James Kapinski, Jyotirmoy~V. Deshmukh, Sriram Sankaranarayanan, and Nikos
  Ar{\'{e}}chiga.
\newblock Simulation-guided {Lyapunov} analysis for hybrid dynamical systems.
\newblock In Martin Fr{\"{a}}nzle and John Lygeros, editors, {\em HSCC}, pages
  133--142. {ACM}, 2014.
\newblock \doi{10.1145/2562059.2562139}.

\bibitem[Kha92]{MR1201326}
Hassan~K. Khalil.
\newblock {\em Nonlinear systems}.
\newblock Macmillan Publishing Company, New York, 1992.

\bibitem[Lia07]{Liapounoff1907}
A.~Liapounoff.
\newblock Probl\'eme g\'en\'eral de la stabilit\'e du mouvement.
\newblock {\em Annales de la Facult\'e des sciences de Toulouse :
  Math\'ematiques}, 9:203--474, 1907.

\bibitem[Lib03]{DBLP:books/sp/Liberzon03}
Daniel Liberzon.
\newblock {\em Switching in Systems and Control}.
\newblock Systems {\&} Control: Foundations {\&} Applications.
  Birkh{\"{a}}user, 2003.
\newblock \doi{10.1007/978-1-4612-0017-8}.

\bibitem[LZZ12]{DBLP:journals/mics/LiuZZ12}
Jiang Liu, Naijun Zhan, and Hengjun Zhao.
\newblock Automatically discovering relaxed {Lyapunov} functions for polynomial
  dynamical systems.
\newblock {\em Math. Comput. Sci.}, 6(4):395--408, 2012.
\newblock \doi{10.1007/s11786-012-0133-6}.

\bibitem[NKJ{\etalchar{+}}17]{DBLP:conf/memocode/NguyenKJDJ17}
Luan~Viet Nguyen, James Kapinski, Xiaoqing Jin, Jyotirmoy~V. Deshmukh, and
  Taylor~T. Johnson.
\newblock Hyperproperties of real-valued signals.
\newblock In Jean{-}Pierre Talpin, Patricia Derler, and Klaus Schneider,
  editors, {\em MEMOCODE}, pages 104--113. {ACM}, 2017.
\newblock \doi{10.1145/3127041.3127058}.

\bibitem[Par00]{Parrilo00}
Pablo~A. Parrilo.
\newblock {\em Structured semidefinite programs and semialgebraic geometry
  methods in robustness and optimization}.
\newblock PhD thesis, California Institute of Technology, 2000.

\bibitem[Pla12]{DBLP:conf/lics/Platzer12a}
Andr{\'{e}} Platzer.
\newblock The complete proof theory of hybrid systems.
\newblock In {\em LICS}, pages 541--550. {IEEE} Computer Society, 2012.
\newblock \doi{10.1109/LICS.2012.64}.

\bibitem[Pla17]{DBLP:journals/jar/Platzer17}
Andr{\'e} Platzer.
\newblock A complete uniform substitution calculus for differential dynamic
  logic.
\newblock {\em J. Autom. Reasoning}, 59(2):219--265, 2017.
\newblock \doi{10.1007/s10817-016-9385-1}.

\bibitem[Pla18]{Platzer18}
Andr{\'e} Platzer.
\newblock {\em Logical Foundations of Cyber-Physical Systems}.
\newblock Springer, Cham, 2018.
\newblock \doi{10.1007/978-3-319-63588-0}.

\bibitem[Poi99]{Poincare92}
Henri Poincar{\'{e}}.
\newblock {\em {Les m\'ethodes nouvelles de la m\'ecanique c\'eleste}}.
\newblock Gauthier-Villars, Paris, 1892--1899.

\bibitem[PP02]{1184414}
A.~{Papachristodoulou} and S.~{Prajna}.
\newblock On the construction of {Lyapunov} functions using the sum of squares
  decomposition.
\newblock In {\em CDC}, volume~3, pages 3482--3487. IEEE, 2002.
\newblock \doi{10.1109/CDC.2002.1184414}.

\bibitem[PT20]{DBLP:journals/jacm/PlatzerT20}
Andr{\'{e}} Platzer and Yong~Kiam Tan.
\newblock Differential equation invariance axiomatization.
\newblock {\em J. ACM}, 67(1), 2020.
\newblock \doi{10.1145/3380825}.

\bibitem[PW06]{DBLP:conf/hybrid/PodelskiW06}
Andreas Podelski and Silke Wagner.
\newblock Model checking of hybrid systems: From reachability towards
  stability.
\newblock In Jo{\~{a}}o~P. Hespanha and Ashish Tiwari, editors, {\em HSCC},
  volume 3927 of {\em LNCS}, pages 507--521. Springer, 2006.
\newblock \doi{10.1007/11730637\_38}.

\bibitem[RHL77]{MR0450715}
Nicolas Rouche, P.~Habets, and M.~Laloy.
\newblock {\em Stability Theory by {L}iapunov's Direct Method}.
\newblock Springer, New York, 1977.
\newblock \doi{10.1007/978-1-4684-9362-7}.

\bibitem[Rou18]{DBLP:conf/cpp/Rouhling18}
Damien Rouhling.
\newblock A formal proof in {Coq} of a control function for the inverted
  pendulum.
\newblock In June Andronick and Amy~P. Felty, editors, {\em CPP}, pages 28--41.
  {ACM}, 2018.
\newblock \doi{10.1145/3167101}.

\bibitem[Rud76]{MR0385023}
Walter Rudin.
\newblock {\em Principles of Mathematical Analysis}.
\newblock McGraw-Hill, third edition, 1976.

\bibitem[SC{\'{A}}13]{DBLP:conf/nolcos/Sankaranarayanan0A13}
Sriram Sankaranarayanan, Xin Chen, and Erika {\'{A}}brah{\'{a}}m.
\newblock {Lyapunov} function synthesis using {Handelman} representations.
\newblock In Sophie Tarbouriech and Miroslav Krstic, editors, {\em NOLCOS},
  pages 576--581. IFAC, 2013.
\newblock \doi{10.3182/20130904-3-FR-2041.00198}.

\bibitem[Str15]{MR3837141}
Steven~H. Strogatz.
\newblock {\em Nonlinear Dynamics and Chaos: With Applications to Physics,
  Biology, Chemistry, and Engineering}.
\newblock Westview Press, Boulder, CO, second edition, 2015.

\bibitem[TPar]{DBLP:journals/fac/TanP}
Yong~Kiam Tan and Andr{\'{e}} Platzer.
\newblock An axiomatic approach to existence and liveness for differential
  equations.
\newblock {\em Formal Aspects Comput.}, to appear.
\newblock \doi{10.1007/s00165-020-00525-0}.

\bibitem[TPS08]{DBLP:journals/automatica/TopcuPS08}
Ufuk Topcu, Andrew~K. Packard, and Peter~J. Seiler.
\newblock Local stability analysis using simulations and sum-of-squares
  programming.
\newblock {\em Autom.}, 44(10):2669--2675, 2008.
\newblock \doi{10.1016/j.automatica.2008.03.010}.

\end{thebibliography}

\iflongversion
\appendix
\section{Proof Calculus}
\label{app:proofcalc}

This appendix gives an extended introduction to the \dL proof calculus that is used for the proofs in~\rref{app:proofs}.
Propositional proof rules, e.g.,~\irref{cut+notr+implyr}, are omitted as they are standard from propositional sequent calculus and can be found in the literature~\cite{Platzer18,DBLP:journals/jacm/PlatzerT20}; the first-order quantifier rules~\irref{alll+existsr} instantiate quantified variables with a given term.
The following lemma summarizes the base axioms and proof rules of \dL.
\begin{lemma}[Axioms and proof rules of \dL~\cite{DBLP:journals/jar/Platzer17,Platzer18}]
\label{lem:dlaxioms}
The following are sound axioms and proof rules of \dL.

\noindent
\begin{calculuscollection}
\begin{calculus}
\cinferenceRule[diamond|$\didia{\cdot}$]{diamond axiom}
{\linferenceRule[equiv]
  {\lnot\dbox{\alpha}{\lnot{\rfvar}}}
  {\axkey{\ddiamond{\alpha}{\rfvar}}}
}
{}
\end{calculus}\qquad
\begin{calculus}
\cinferenceRule[K|K]{K axiom / modal modus ponens} %
{\linferenceRule[impl]
  {\dbox{\alpha}{(\rrfvar \limply \rfvar)}}
  {(\dbox{\alpha}{\rrfvar}\limply\axkey{\dbox{\alpha}{\rfvar}})}
}{}
\end{calculus}\\
\begin{calculus}
\cinferenceRule[V|V]{vacuous $\dbox{}{}$}
 {\linferenceRule[impl]
   {\fvarA}
   {\axkey{\dbox{\alpha}{\fvarA}}}
 }{\text{no free variable of $\fvarA$ is bound by $\alpha$}}
\end{calculus}\\
\begin{calculus}
\dinferenceRule[dIcmp|dI$_\cmp$]{}
{\linferenceRule
  {\lsequent{\ivr}{\lied[]{\genDE{x}}{\ptermA}\geq\lied[]{\genDE{x}}{\ptermB}}
  }
  {\lsequent{\Gamma,\ptermA \cmp \ptermB }{\dbox{\pevolvein{\D{x}=\genDE{x}}{\ivr}}{\ptermA \cmp \ptermB}} }
  \quad
}{where $\cmp$ is either $\geq$ or $>$}

\dinferenceRule[dC|dC]{}
{\linferenceRule
  {\lsequent{\Gamma}{\dbox{\pevolvein{\D{x}=\genDE{x}}{\ivr}}{\rcfvar}}
  &\lsequent{\Gamma}{\dbox{\pevolvein{\D{x}=\genDE{x}}{\ivr \land \rcfvar}}{\rfvar}}
  }
  {\lsequent{\Gamma}{\dbox{\pevolvein{\D{x}=\genDE{x}}{\ivr}}{\rfvar}}}
}{}

\cinferenceRule[DG|DG]{differential ghost}
{\linferenceRule[equiv]
  {\lexists{y}{\dbox{\pevolvein{\D{x}=\genDE{x},\D{y}=a(x)y+b(x)}{\ivr(x)}}{\rfvar(x)}}}
  {\axkey{\dbox{\pevolvein{\D{x}=\genDE{x}}{\ivr(x)}}{\rfvar(x)}}}
}{}

\dinferenceRule[dW|dW]{}
{\linferenceRule
  {\lsequent{\ivr}{\rfvar}}
  {\lsequent{\Gamma}{\dbox{\pevolvein{\D{x}=\genDE{x}}{\ivr}}{\rfvar}}}
}{}
\end{calculus}\\
\begin{calculus}
\dinferenceRule[MbW|M${\dibox{'}}$]{}
{\linferenceRule
  {\lsequent{\ivr,\rrfvar}{\rfvar} \quad \lsequent{\Gamma}{\dbox{\pevolvein{\D{x}=\genDE{x}}{\ivr}}{\rrfvar}}}
  {\lsequent{\Gamma}{\dbox{\pevolvein{\D{x}=\genDE{x}}{\ivr}}{\rfvar}}}
}{}
\end{calculus} \qquad
\begin{calculus}
\dinferenceRule[MdW|M${\didia{'}}$]{}
{\linferenceRule
  {\lsequent{\ivr,\rrfvar}{\rfvar} \quad \lsequent{\Gamma}{\ddiamond{\pevolvein{\D{x}=\genDE{x}}{\ivr}}{\rrfvar}}}
  {\lsequent{\Gamma}{\ddiamond{\pevolvein{\D{x}=\genDE{x}}{\ivr}}{\rfvar}}}
}{}
\end{calculus}
\end{calculuscollection}
\label{lem:axbase}
\end{lemma}
\begin{proof}
The soundness of all axioms and proof rules in~\rref{lem:axbase} are proved elsewhere~\cite{DBLP:journals/jar/Platzer17,Platzer18}.
\end{proof}

The first three hybrid program axioms~\irref{diamond+K+V}~\cite{DBLP:journals/jar/Platzer17,Platzer18} of \dL are used in this paper when the hybrid program $\alpha$ is an ODE.
Axiom~\irref{diamond} expresses the duality between the box and diamond modalities.
It is used to switch between the two in proofs to turn a liveness property ($\didia{\cdot}$) of an ODE to a (negated) safety property ($\dibox{\cdot}$) of the same ODE, or vice versa.
Axiom~\irref{K} is the modus ponens principle for the box modality.
Vacuous axiom~\irref{V} says if no free variable of $\fvarA$ is changed by hybrid program $\alpha$, then the truth value of $\fvarA$ is also unchanged.
This axiom formally justifies that assumptions involving only constant ODE parameters can be soundly kept across ODE deduction steps in proofs~\cite{Platzer18}, as explained in~\rref{sec:background}.

Differential invariants~\irref{dIcmp} say that if the Lie derivatives obey the inequality $\lied[]{\genDE{x}}{\ptermA} \geq \lied[]{\genDE{x}}{\ptermB}$, then $\ptermA \cmp \ptermB$ is an invariant of the ODE.
Differential cuts~\irref{dC} say that if one can separately prove that formula $\rcfvar$ is always satisfied along ODE solutions, then $\rcfvar$ may be assumed in the domain constraint when proving the same for formula $\rfvar$.
Differential ghosts~\irref{DG} say that, in order to prove safety postcondition $\rfvar(x)$ for the ODE $\D{x}=\genDE{x}$, it suffices to prove $\rfvar(x)$ for a larger system with an added ODE $\D{y}=a(x)y+b(x)$ that is linear in the fresh ghost variable $y$ (because $a(x),b(x)$ do not mention $y$).
This addition is sound because the ODE $\D{x}=\genDE{x}$ does not mention the added variables $y$, and so the evolution of $\D{x}=\genDE{x}$ is unaffected by the addition of a \emph{linear} ODE for $y$~\cite{DBLP:journals/jar/Platzer17}.
Since $y$ is fresh, its initial value can be either existentially (\irref{DG}) or universally (\irref{DGall}) quantified~\cite{DBLP:journals/jar/Platzer17}.
\irlabel{DGall|DG$_\forall$}%

In the box modality, solutions are restricted to stay in the domain constraint $\ivr$.
Thus, differential weakening~\irref{dW} says that postcondition $\rfvar$ is always satisfied along solutions if it is already implied by the domain constraint.
Using \irref{dW+K+diamond}, the final two monotonicity proof rules~\irref{MbW+MdW} for differential equations are derivable.
They strengthen the postcondition from $\rfvar$ to $\rrfvar$, assuming domain constraint $\ivr$, for the box and diamond modalities respectively.

The \dL proof calculus has a sound and complete axiomatization for ODE invariants~\cite{DBLP:journals/jacm/PlatzerT20}; completeness means that for ODE $\D{x}=\genDE{x}$, if formula $\rfvar$ is an invariant of the ODE, then its invariance is syntactically provable in \dL.
For added clarity in this paper's invariance proofs, the following lemma lists additional axioms and proof rules that are useful for step-by-step proofs (\rref{app:proofs}).
The additional detail in these proofs helps inform the implementation of stability proof rules in \KeYmaeraX~\cite{DBLP:conf/cade/FultonMQVP15,DBLP:journals/jacm/PlatzerT20}.
The lemma also lists axioms and proof rules from \dL's refinement-based proof approach for ODE liveness properties~\cite{DBLP:journals/fac/TanP}.

\begin{lemma}[Invariance and liveness \dL axioms and proof rules~\cite{DBLP:journals/jacm/PlatzerT20,DBLP:journals/fac/TanP}]
\label{lem:invliveaxioms}
The following are sound axioms and proof rules of \dL.

\noindent
{\begin{calculuscollection}
\begin{calculus}

\cinferenceRule[DX|DX]{}
{\linferenceRule[equiv]
  {(\ivr \limply \rfvar \land \dbox{\pevolvein{\D{x}=\genDE{x}}{\ivr}}{\rfvar})}
  {\axkey{\dbox{\pevolvein{\D{x}=\genDE{x}}{\ivr}}{\rfvar}}}
}{\text{$\D{x} \not\in \rfvar,\ivr$}}

\cinferenceRule[DSeq|D$\dibox{{;}}$]{}
{
\linferenceRule[equiv]
  {\dbox{\pevolvein{\D{x}=\genDE{x}}{\ivr}}{\dbox{\pevolvein{\D{x}=\genDE{x}}{\ivr}}{\rfvar}}}
  {\dbox{\pevolvein{\D{x}=\genDE{x}}{\ivr}}{\rfvar}}
}{}

\cinferenceRule[DMP|DMP]{differential modus ponens}
{\linferenceRule[impl]
  {\dbox{\pevolvein{\D{x}=\genDE{x}}{\ivr}}{(\ivr \limply \rrfvar)}}
  {(\dbox{\pevolvein{\D{x}=\genDE{x}}{\rrfvar}}{\rfvar} \limply \axkey{\dbox{\pevolvein{\D{x}=\genDE{x}}{\ivr}}{\rfvar}})}
}{}

\dinferenceRule[DCC|DCC]{}
{
\linferenceRule[impl]
  {\dbox{\pevolvein{\D{x}=\genDE{x}}{\ivr \land \rfvar}}{\rrfvar} \land \dbox{\pevolvein{\D{x}=\genDE{x}}{\ivr}}{( \lnot{\rfvar} \limply \dbox{\pevolvein{\D{x}=\genDE{x}}{\ivr}}{\lnot{\rfvar}})}}
  {\dbox{\pevolvein{\D{x}=\genDE{x}}{\ivr}}{(\rfvar \limply \rrfvar)}}
}{}

\cinferenceRule[dBarcan|B$'$]{}
{\linferenceRule[equiv]
  {\lexists{y}{\ddiamond{\pevolvein{\D{x}=\genDE{x}}{\ivr\argx}}{\rfvar(x,y)}}}
  {\axkey{\ddiamond{\pevolvein{\D{x}=\genDE{x}}{\ivr\argx}}{\exists{y}\rfvar(x,y)}}}
}{\text{$y \not\in x$}}

\cinferenceRule[BDG|BDG]{}
{\linferenceRule[impll]
  {\dbox{\pevolvein{\D{x}=\genDE{x},\D{y}=g(x,y)}{\ivr(x)}}{\,\norm{y}^2 \leq \ptermA(x)}}
  {\big( \dbox{\pevolvein{\D{x}=\genDE{x}}{\ivr(x)}}{\rfvar(x)} \lbisubjunct \axkey{\dbox{\pevolvein{\D{x}=\genDE{x}, \D{y}=g(x,y)}{\ivr(x)}}{\rfvar(x)}}\big)}
}{}

\dinferenceRule[Enc|Enc]{Enclosure}
{\linferenceRule
  {
  \lsequent{\Gamma}{\closure{\rfvar}} &
  \lsequent{\Gamma}{\dbox{\pevolvein{\D{x}=\genDE{x}}{\ivr \land \closure{\rfvar}}}{\rfvar}}
  }
  {\lsequent{\Gamma}{\dbox{\pevolvein{\D{x}=\genDE{x}}{\ivr}}{\rfvar}}}
}{\text{formula $\rfvar$ characterizes an open set}}

\dinferenceRule[dbxineq|dbx${_\cmp}$]{Darboux inequality}
{\linferenceRule
  {\lsequent{\ivr} {\lied[]{\genDE{x}}{\ptermA}\geq \cofterm\ptermA}}
  {\lsequent{\ptermA\cmp0} {\dbox{\pevolvein{\D{x}=\genDE{x}}{\ivr}}{\ptermA\cmp0}}}
}{\text{$\cmp$ is either $\geq$ or $>$}}

\dinferenceRule[Prog|K${\didia{\&}}$]{}
{
\linferenceRule[impl]
  {\dbox{\pevolvein{\D{x}=\genDE{x}}{\ivr \land \lnot{\rfvar}}}{\lnot{\rgvar}}}
  {\big(\ddiamond{\pevolvein{\D{x}=\genDE{x}}{\ivr}}{\rgvar} \limply \axkey{\ddiamond{\pevolvein{\D{x}=\genDE{x}}{\ivr}}{\rfvar}}\big)}
}{}

\dinferenceRule[SPc|SP$_c$]{}
{
\linferenceRule
  { \lsequent{\Gamma}{\dbox{\pevolvein{\D{x}=\genDE{x}}{\lnot{\rfvar}}}{\rsfvar}}
   &\lsequent{\rsfvar}{\lied[]{\genDE{x}}{\ptermA} > 0}
  }
  {\lsequent{\Gamma}{\ddiamond{\pevolve{\D{x}=\genDE{x}}}{\rfvar}} }
}{\text{formula $\rsfvar$ characterizes a compact set}}
\end{calculus}
\end{calculuscollection}
}%
\label{lem:axinvlive}
\end{lemma}
\begin{proof}
The soundness of all axioms and proof rules in~\rref{lem:axinvlive} are proved elsewhere~\cite{DBLP:journals/jacm/PlatzerT20,DBLP:journals/fac/TanP} except axioms~\irref{DSeq} and~\irref{DCC}.
Axiom~\irref{DCC} is stated as a proof rule elsewhere~\cite{DBLP:conf/tacas/KolcakDHKS020} and its axiomatic version is formally verified~\cite{DBLP:journals/afp/Bohrer17}.
Axiom~\irref{DSeq} is proved sound next.

Consider the initial state $\iget[state]{\I} \in \States$ and let $\solvar : [0, T) \to \States, 0<T\leq\infty$ be the unique, right-maximal solution~\cite{Chicone2006} to the ODE $\D{x}=\genDE{x}$ with initial value $\solvar(0)=\iget[state]{\I}$.
Unfolding the semantics, the RHS of axiom~\irref{DSeq} is true in state $\iget[state]{\I}$ iff for all times $0 \leq \tau < T$ such that $\solvar(\zeta) \in \imodel{\I}{\ivr}~\text{for all}~0 \leq \zeta \leq \tau$, the solution at time $\tau$ satisfies $\solvar(\tau) \in \imodel{\I}{\dbox{\pevolvein{\D{x}=\genDE{x}}{\ivr}}{\rfvar}}$.
Further unfolding the semantics, by uniqueness of ODE solutions~\cite{Chicone2006}, this means that for all times $\tau \leq t < T$, such that $\solvar(\zeta) \in \imodel{\I}{\ivr}~\text{for all}~\tau \leq \zeta \leq t$, the solution at time $t$ satisfies $\solvar(t) \in \imodel{\I}{\rfvar}$.
Thus, the RHS is true in state $\iget[state]{\I}$ iff for all times $0 \leq \tau < T$ such that $\solvar(\zeta) \in \imodel{\I}{\ivr}~\text{for all}~0 \leq \zeta \leq \tau$, the solution at time $\tau$ satisfies $\solvar(\tau) \in \imodel{\I}{\rfvar}$, which is the unfolded semantics of the LHS of axiom~\irref{DSeq}. \qedhere
\end{proof}

Differential skip~\irref{DX} is a reflexivity property of differential equation solutions, if the domain constraint $\ivr$ is true initially, then $\rfvar$ must also be true initially because of the ODE solution at time $0$.
Differential compose~\irref{DSeq} is a transitivity property of differential equation solutions, because any state reachable from two sequential runs of an ODE is reachable in a single run of the ODE.

Axiom~\irref{DMP} is the modus ponens principle for domain constraints.
Axiom~\irref{DCC} says in order to prove that an implication $\rfvar \limply \rrfvar$ is always true along an ODE, it suffices to prove it assuming $\rfvar$ in the domain if $\lnot{\rfvar}$ is invariant along the ODE.
The ODE Barcan axiom~\irref{dBarcan} specializes the Barcan axiom~\cite{Platzer18,DBLP:journals/jacm/PlatzerT20} to ODEs in the diamond modality, allowing an existential quantifier $\lexists{y}{}$ to be commuted with the diamond modality.
The bounded differential ghosts axiom~\irref{BDG} generalizes~\irref{DG} by allowing a vectorial ghost ODE with arbitrary RHS $\D{y}=g(x,y)$ to be added.
For soundness, this RHS must be bounded in norm with respect to the existing variables $x$ of the ODE so that the added ODEs for $y$ do not blow up before the original solution for $\D{x}=\genDE{x}$~\cite{DBLP:journals/fac/TanP}.
Rule~\irref{Enc} is a derived proof rule which says that to prove that a solution always stays in the open set characterized by formula $\rfvar$, it suffices to prove it assuming the closure $\closure{\rfvar}$ in the domain constraint.
Rule~\irref{dbxineq} is a derived proof rule~\cite{DBLP:journals/jacm/PlatzerT20} which proves the invariance of formula $\ptermA \cmp 0$ if its Lie derivative satisfies inequality  $\lied[]{\genDE{x}}{\ptermA}\geq \cofterm\ptermA$ for an arbitrarily chosen cofactor term $\cofterm$.

Axiom~\irref{Prog} is a derived ODE liveness refinement axiom~\cite{DBLP:journals/fac/TanP}.
The formula $\dbox{\pevolvein{\D{x}=\genDE{x}}{\ivr \land \lnot{\rfvar}}}{\lnot{\rgvar}}$ says $\rgvar$ never happens along the solution while $\lnot{\rfvar}$ holds.
Thus, the solution cannot get to $\rgvar$ unless it gets to $\rfvar$ first.
Rule~\irref{SPc} is a derived ODE liveness proof rule~\cite{DBLP:journals/fac/TanP}.
Its left premise says that as long as the ODE solution has not reached the target $\rfvar$, it stays in the compact staging set $\rsfvar$.
Its right premise says that the value of $\ptermA$ increases as long as solutions stay in $\rsfvar$.
Due to the compactness assumption, $\ptermA$ \emph{cannot} increase indefinitely in $\rsfvar$ and so the ODE solution must eventually leave the staging set by entering the target $\rfvar$.
This rule implicitly proves that the ODE solution exists for sufficient duration to reach $\rfvar$, which is a fundamental requirement for soundness in ODE liveness arguments~\cite{DBLP:journals/fac/TanP}.

\section{Proofs}
\label{app:proofs}

This appendix provides proofs for all lemmas and corollaries from Sections~\ref{sec:asymstability} and~\ref{sec:genstability} using the \dL proof calculus in~\rref{app:proofcalc}.
For ease of reference, this appendix is organized into two sections, corresponding to proofs for~\rref{sec:asymstability} and~\rref{sec:genstability} respectively.
Additional definitions and lemmas omitted in the main paper are provided as required.

\subsection{Proofs for Asymptotic Stability of an Equilibrium Point}
This section concerns stability for the origin, whose $\varepsilon$ neighborhoods $\neighborhood[\varepsilon]{x=0}$ are equivalently unfolded as the formula $\norm{x}^2 < \varepsilon^2$.
The following lemma formalizes the claim in~\rref{subsec:asymstabilitymath} that a point $x_0$ of interest for the ODE $\D{x}=\genDE{x}$ can be rigorously translated \emph{with proof} to the origin so that, without loss of generality, only stability of the origin needs to be considered for the stability proof rules of~\rref{sec:asymstability}.

\begin{lemma}[Translation to origin]
The following axioms are derivable in \dL, where ODE $\D{y}=\genDE{y+x_0}$ has point $x_0$ translated to the origin and variables $y$ are fresh, i.e., not in ODE $\D{x}=\genDE{x}$ or formula $\rfvar(x)$.

\noindent
\begin{calculuscollection}
\begin{calculus}
\dinferenceRule[Trans|Trans]{}
{
\linferenceRule[impl]
  {y=x-x_0}
  {\big(\dbox{\D{x}=\genDE{x}}{\rfvar(x-x_0)} \lbisubjunct \dbox{\D{y}=\genDE{y+x_0}}{\rfvar(y)}\big)}
}{}

\dinferenceRule[TransStab|TransStab]{}
{
\linferenceRule[impl]
  {\stabode{\D{y}=\genDE{y+x_0}} }
  {\stabodePR{\D{x}=\genDE{x}}{x=x_0}{x=x_0} }
}{}
\end{calculus}
\end{calculuscollection}
\label{lem:translation}
\end{lemma}

\begin{proof}
Axiom~\irref{Trans} is derived first before axiom~\irref{TransStab} is derived from it as a corollary.
Only the ``$\limply$'' direction of the inner equivalence for axiom~\irref{Trans} is derived since the ``$\lylpmi$'' direction follows from the ``$\limply$'' direction by renaming and translation with respect to ${-}x_0$.

Let $\progxy \mnodefequiv \D{x} = \genDE{x},\D{y} = \genDE{y+x_0}$ abbreviate the combined ODE for variables $x$ and $y$.
The derivation starts with a~\irref{cut} of formula $\dbox{\progxy}{\,y = x-x_0}$, which says the value of $y$ is always equal to $x-x_0$ along solutions of the combined ODE $\progxy$.
This cut is provable in \dL using its complete, derived proof rule for algebraic invariants~\cite{DBLP:journals/jacm/PlatzerT20}.
Subsequently, axiom~\irref{BDG} adds the ghost ODE $\D{y}=\genDE{y+x_0}$ to the antecedent box modality and ODE $\D{x}=\genDE{x}$ to the succedent box modality.
The resulting boundedness premises from the use of~\irref{BDG} are respectively abbreviated \textcircled{1} and \textcircled{2}, and are both proved using the cut antecedent as shown further below.
{\begin{sequentdeduction}[array]
  \linfer[cut]{
  \linfer[BDG]{
  \linfer[BDG]{
    \textcircled{2} \qquad \lsequent{\dbox{\progxy}{y = x - x_0},\dbox{\progxy}{\rfvar(x-x_0)}}{\dbox{\progxy}{\rfvar(y)}}
  }
    {\textcircled{1} \qquad \lsequent{\dbox{\progxy}{y = x - x_0},\dbox{\progxy}{\rfvar(x-x_0)}}{\dbox{\D{y}=\genDE{y+x_0}}{\rfvar(y)}}}
  }
  {\lsequent{\dbox{\progxy}{y = x - x_0},\dbox{\D{x}=\genDE{x}}{\rfvar(x-x_0)}}{\dbox{\D{y}=\genDE{y+x_0}}{\rfvar(y)}}}
  }
  {\lsequent{y=x-x_0,\dbox{\D{x}=\genDE{x}}{\rfvar(x-x_0)}}{\dbox{\D{y}=\genDE{y+x_0}}{\rfvar(y)}}}
\end{sequentdeduction}
}%

From the (unabbreviated) open right premise, a~\irref{dC} step adds formulas $y = x - x_0$ and $\rfvar(x - x_0)$ to the domain constraint of the succedent.
A subsequent~\irref{dW} step completes the proof by substituting $y=x-x_0$.
{\begin{sequentdeduction}[array]
  \linfer[dC]{
  \linfer[dC]{
  \linfer[dW]{
  \linfer[qear]{
    \lclose
  }
    {\lsequent{y=x-x_0 \land \rfvar(x-x_0)}{\rfvar(y)}}
  }
  {\lsequent{}{\dbox{\pevolvein{\progxy}{y=x-x_0 \land \rfvar(x-x_0)}}{\rfvar(y)}}}
  }
  {\lsequent{\dbox{\progxy}{\rfvar(x-x_0)}}{\dbox{\pevolvein{\progxy}{y=x-x_0}}{\rfvar(y)}}}
  }
  {\lsequent{\dbox{\progxy}{y = x - x_0},\dbox{\progxy}{\rfvar(x-x_0)}}{\dbox{\progxy}{\rfvar(y)}}}
\end{sequentdeduction}
}%

For premise \textcircled{1}, the ghost variables $y$ are provably bounded in (squared) norm by the term $\norm{x-x_0}^2$ for ODE $\progxy$.
The~\irref{dC} step adds $y = x - x_0$ from the antecedent box modality to the domain constraint, and the subsequent~\irref{dW} step completes the proof by substituting $y=x-x_0$ and~\irref{qear}.

{\begin{sequentdeduction}[array]
  \linfer[dC]{
  \linfer[dW]{
  \linfer[qear]{
      \lclose
  }
    {\lsequent{y = x - x_0}{\norm{y}^2 \leq \norm{x-x_0}^2}}
  }
    {\lsequent{}{\dbox{\pevolvein{\progxy}{y = x - x_0}}{\norm{y}^2 \leq \norm{x-x_0}^2}}}
  }
  {\lsequent{\dbox{\progxy}{y = x - x_0}}{\dbox{\progxy}{\norm{y}^2 \leq \norm{x-x_0}^2}}}
\end{sequentdeduction}
}%

The derivation for premise \textcircled{2} is similar, where the ghost variables $x$ are provably bounded in (squared) norm by the term $\norm{y+x_0}^2$ for ODE $\progxy$.
{\begin{sequentdeduction}[array]
  \linfer[dC]{
  \linfer[dW]{
  \linfer[qear]{
      \lclose
  }
    {\lsequent{y = x - x_0}{\norm{x}^2 \leq \norm{y+x_0}^2}}
  }
    {\lsequent{}{\dbox{\pevolvein{\progxy}{y = x - x_0}}{\norm{x}^2 \leq \norm{y+x_0}^2}}}
  }
  {\lsequent{\dbox{\progxy}{y = x - x_0}}{\dbox{\progxy}{\norm{x}^2 \leq \norm{y+x_0}^2}}}
\end{sequentdeduction}
}%

The derivation of axiom~\irref{TransStab} starts by Skolemizing $\varepsilon$ in the succedent with~\irref{allr} and then instantiating the antecedent with the resulting fresh Skolem variable $\varepsilon$ using~\irref{alll}.
This is followed by~\irref{existsl+existsr} which Skolemizes $\delta$ in the antecedent and then witnesses the succedent with $\delta$.
Next,~\irref{allr+implyr} Skolemizes the succedent before~\irref{alll} instantiates $y$ in the quantified antecedent, with the translated coordinate $y=x-x_0$.
Formula $y=x-x_0 \land \norm{x-x_0}^2 < \delta^2 \limply \norm{y}^2 < \delta^2$ is provable in real arithmetic, so~\irref{implyl+qear} proves the LHS of the implication ($\norm{y}^2 < \delta^2$) in the antecedent before~\irref{Trans} completes the proof.
The formulas are abbreviated $\rrfvar_y \mnodefequiv \dbox{\D{y}=\genDE{y+x_0}}{\,\norm{y}^2 < \varepsilon^2}$ and $\rrfvar \mnodefequiv \dbox{\D{x}=\genDE{x}}{\,\norm{x-x_0}^2 < \varepsilon^2}$, respectively.
{\begin{sequentdeduction}[array]
  \linfer[allr+alll]{
  \linfer[existsl+existsr]{
  \linfer[allr+implyr]{
  \linfer[alll]{
  \linfer[implyl+qear]{
  \linfer[Trans]{
    \lclose
  }  {\lsequent{y=x-x_0,\rrfvar_y} {\rrfvar}}
  }
    {\lsequent{y=x-x_0, \norm{y}^2 < \delta^2 \limply \rrfvar_y,\norm{x-x_0}^2 < \delta^2 }
      {\rrfvar}}
  }
    {\lsequent{\lforall{y}{\big( \norm{y}^2 < \delta^2 \limply \rrfvar_y \big)},\norm{x-x_0}^2 < \delta^2 } {\rrfvar}}
  }
    {\lsequent{\lforall{y}{\big( \norm{y}^2 < \delta^2 \limply \rrfvar_y \big)}}
      {\lforall{x}{\big( \norm{x-x_0}^2 < \delta^2 \limply \rrfvar \big)}}}
  }
    {\lsequent{\lexists{\delta {>} 0}{\lforall{y}{\big( \norm{y}^2 < \delta^2 \limply \rrfvar_y \big)}}}
      {\lexists{\delta {>} 0}{ \lforall{x}{\big( \norm{x-x_0}^2 < \delta^2 \limply \rrfvar \big)}}}}
  }
  {\lsequent{\stabode{\D{y}=\genDE{y+x_0}} }{\stabodePR{\D{x}=\genDE{x}}{x=x_0}{x=x_0}}}
\\[-\normalbaselineskip]\tag*{\qedhere}
\end{sequentdeduction}
}%

\end{proof}

\begin{proof}[\textbf{Proof of~\rref{lem:asymstabdl}}]
The correctness of these specifications follows directly from the semantics of \dL formulas~\cite{DBLP:journals/jar/Platzer17,Platzer18} because they syntactically express the logical connectives and quantifiers from~\rref{def:asymstabmath} in \dL.
The open neighborhood formulas $\neighborhood[\delta]{x=0}$ and $\neighborhood[\varepsilon]{x=0}$ are true in states where $\norm{x} < \delta$ and $\norm{x} < \varepsilon$ respectively.
The main subtlety is formula $\attrode{\D{x}=\genDE{x}}$ which characterizes the limit $\lim_{t \to T}{x(t) = 0}$ using its subformula $\asymode{\D{x}=\genDE{x}}{x=0}$ as follows.

Unfolding the semantics, formula $\asymode{\D{x}=\genDE{x}}{\rfvar}$ is true in an initial state iff for any $\varepsilon > 0$, the right-maximal ODE solution to $\D{x}=\genDE{x}$ (restricted to variables $x$) denoted $x(t) : [0,T) \to \reals^n$ has a time $\tau \in [0,T)$ where, because of uniqueness of ODE solutions~\cite[Theorem 1.2]{Chicone2006}, for all future times $t$ with $\tau \leq t < T$, the solution at $x(t)$ satisfies formula $\neighborhood[\varepsilon]{\rfvar}$; for $\rfvar \mnodefequiv x=0$, this implies the bound $\norm{x(t)} < \varepsilon$ at those future times, which is the real analytic definition of the limit $\lim_{t \to T}{x(t) = 0}$~\cite[Definition 4.1]{MR0385023}. \qedhere
\end{proof}

\begin{proof}[\textbf{Proof of~\rref{cor:asymstabsimp}}]
A full proof is omitted as axiom~\irref{stabattr} is an instance of the more general axiom~\irref{stabattrgen} derived in~\rref{cor:setstabsimp} (where $\rfvar \mnodefequiv x=0$).
Briefly, the ``$\limply$'' direction of the inner equivalence is valid even without assuming stability because postcondition $\dbox{\D{x}=\genDE{x}}{\,\neighborhood[\varepsilon]{x{=}0}}$ monotonically implies postcondition $\neighborhood[\varepsilon]{x{=}0}$.
The (interesting) ``$\lylpmi$'' direction of the inner equivalence uses the stability assumption by choosing $\delta > 0$ sufficiently small so that solutions reaching $\neighborhood[\delta]{x{=}0}$ must stay in $\neighborhood[\varepsilon]{x{=}0}$ thereafter because of stability.\qedhere
\end{proof}

\begin{proof}[\textbf{Proof of~\rref{lem:lyapunov}}]
Rule~\irref{Lyap} is derived first as it is an important stepping stone in the derivation of rule~\irref{StrictLyap}.

The derivation of rule~\irref{Lyap} begins with a series of arithmetic cuts which are justified stepwise.
For any $\varepsilon > 0$ and an equilibrium point at the origin with $\genDE{0}=0$, the second (right) premise of~\irref{Lyap} can be equivalently strengthened to choose $\gamma \leq \varepsilon$, i.e., the second premise provably implies the following formula in real arithmetic; the universal quantifier on $x$ is also distributed across the inner conjunction as shown in conjuncts \textcircled{a} and \textcircled{b} below:
\[ \lexists{0{<}\gamma{\leq}\varepsilon}{\big(\underbrace{\lforall{x}{ (0 < \norm{x}^2 \leq \gamma^2 \limply \lterm > 0)}}_{\textcircled{a}} \land \underbrace{\lforall{x}{( \norm{x}^2 \leq \gamma^2 \limply \lied[]{\genDE{x}}{\lterm} \leq 0)}}_{\textcircled{b}}\big)} \]

The derivation begins with a~\irref{cut} of this formula and Skolemizing with~\irref{existsl}, yielding antecedents \textcircled{a} and \textcircled{b} as indicated above.
The postcondition is then monotonically strengthened to $\norm{x}^2 < \gamma^2$ using antecedent $\gamma \leq \varepsilon$.
{\begin{sequentdeduction}[array]
  \linfer[allr]{
  \linfer[cut+existsl]{
  \linfer[MbW]{
    \lsequent{\gamma > 0, \textcircled{a}, \textcircled{b}}{\lexists{\delta {>} 0} {\lforall{x}{\big( \norm{x}^2 < \delta^2 \limply \dbox{\D{x}=\genDE{x}}{ \,\norm{x}^2 < \gamma^2}\big)}}}
  }
    {\lsequent{\gamma>0, \gamma \leq \varepsilon, \textcircled{a}, \textcircled{b}}{\lexists{\delta {>} 0} {\lforall{x}{\big( \norm{x}^2 < \delta^2 \limply \dbox{\D{x}=\genDE{x}}{ \,\norm{x}^2 < \varepsilon^2}\big)}}}}
  }
  {\lsequent{\varepsilon > 0}{\lexists{\delta {>} 0} {\lforall{x}{\big( \norm{x}^2 < \delta^2 \limply \dbox{\D{x}=\genDE{x}}{ \,\norm{x}^2 < \varepsilon^2}\big)}}}}}
  {\lsequent{}{\stabode{\D{x}=\genDE{x}}}}
\end{sequentdeduction}
}%

From \textcircled{a}, the continuous Lyapunov function $\lterm$ is positive on the compact set characterized by $\norm{x}^2 = \gamma^2$ and therefore is bounded below by its minimum $k > 0$ on that set.
Furthermore, from premise $\lterm(0)=0$, by continuity, $\lterm$ must take values smaller than $k$ in a ball with sufficiently small radius $0 < \delta < \gamma$ around the origin.
Thus, the following formula proves in real arithmetic from \textcircled{a}.
\[ \lexists{k}{\Big( \underbrace{\lforall{x}{\big( \norm{x}^2 = \gamma^2 \limply \lterm \geq k \big)}}_{\textcircled{c}} \land~\lexists{0 {<} \delta {<} \gamma}{\underbrace{\lforall{x}{\big( \norm{x}^2 < \delta^2 \limply \lterm < k\big)}}_{\textcircled{d}}} \Big)} \]

The proof continues with a~\irref{cut} of the formula and Skolemizing the resulting antecedent with~\irref{existsl}.
{\begin{sequentdeduction}[array]
  \linfer[cut+qear+existsl]{
    \lsequent{0 {<} \delta {<} \gamma, \textcircled{b}, \textcircled{c}, \textcircled{d}}{\lexists{\delta {>} 0} {\lforall{x}{\big( \norm{x}^2 < \delta^2 \limply \dbox{\D{x}=\genDE{x}}{ \,\norm{x}^2 < \gamma^2}\big)}}}
  }
  {\lsequent{\gamma> 0, \textcircled{a}, \textcircled{b}}{\lexists{\delta {>} 0} {\lforall{x}{\big( \norm{x}^2 < \delta^2 \limply \dbox{\D{x}=\genDE{x}}{ \,\norm{x}^2 < \gamma^2}\big)}}}
}
\end{sequentdeduction}
}%

The succedent existential quantifier $\exists{\delta{>}0}$ is instantiated with the antecedent's $\delta$ using~\irref{existsr}, followed by simplification steps, Skolemizing and unfolding the succedent with~\irref{allr+implyr} then instantiating \textcircled{d} by \irref{alll+implyl} with the resulting $\norm{x}^2 < \delta^2$ assumption.
{\begin{sequentdeduction}[array]
  \linfer[existsr]{
  \linfer[allr+implyr]{
  \linfer[alll+implyl]{
    \lsequent{\delta {<} \gamma, \textcircled{b}, \textcircled{c}, \norm{x}^2 < \delta^2, v < k}{\dbox{\D{x}=\genDE{x}}{ \,\norm{x}^2 < \gamma^2}}
  }
    {\lsequent{\delta {<} \gamma, \textcircled{b}, \textcircled{c}, \textcircled{d}, \norm{x}^2 < \delta^2}{\dbox{\D{x}=\genDE{x}}{ \,\norm{x}^2 < \gamma^2}}}
  }
    {\lsequent{\delta {<} \gamma, \textcircled{b}, \textcircled{c}, \textcircled{d}}{\lforall{x}{\big( \norm{x}^2 < \delta^2 \limply \dbox{\D{x}=\genDE{x}}{ \,\norm{x}^2 < \gamma^2}\big)}}}
  }
  {\lsequent{0 {<} \delta {<} \gamma, \textcircled{b}, \textcircled{c}, \textcircled{d}}{\lexists{\delta {>} 0} {\lforall{x}{\big( \norm{x}^2 < \delta^2 \limply \dbox{\D{x}=\genDE{x}}{ \,\norm{x}^2 < \gamma^2}\big)}}}}
\end{sequentdeduction}
}%

Since formula $\norm{x}^2 < \gamma^2$ characterizes an open ball and $\delta < \gamma$ so antecedent $\norm{x}^2 < \delta^2$ implies $\norm{x}^2 < \gamma^2$ arithmetically, rule~\irref{Enc} is used to assume its closure $\norm{x}^2 \leq \gamma^2$ in the domain constraint of the succedent ODE.
With the strengthened domain,~\irref{dC} adds $\lied[]{\genDE{x}}{\lterm} \leq 0$ to the domain constraint using \textcircled{b}, which is universally quantified over $x$.
{\begin{sequentdeduction}[array]
  \linfer[Enc]{
  \linfer[dC]{
    \lsequent{\textcircled{c}, v < k}{\dbox{\pevolvein{\D{x}=\genDE{x}}{\norm{x}^2 \leq \gamma^2 \land \lied[]{\genDE{x}}{\lterm} \leq 0}}{\norm{x}^2 < \gamma^2}}
  }
    {\lsequent{\textcircled{b}, \textcircled{c}, v < k}{\dbox{\pevolvein{\D{x}=\genDE{x}}{\norm{x}^2 \leq \gamma^2}}{\norm{x}^2 < \gamma^2}}}
  }
  {\lsequent{\delta {<} \gamma, \textcircled{b}, \textcircled{c}, \norm{x}^2 < \delta^2, v < k}{\dbox{\D{x}=\genDE{x}}{ \,\norm{x}^2 < \gamma^2}}
}
\end{sequentdeduction}
}%

The proof continues using~\irref{dC} to add invariant $v < k$ to the domain constraint; this differential cut proves by~\irref{dIcmp} using conjunct $\lied[]{\genDE{x}}{\lterm} \leq 0$ in the domain constraint.
The subsequent~\irref{dC} step adds $\norm{x}^2 \neq \gamma^2$ to the domain constraint using the contrapositive direction of the universally quantified antecedent \textcircled{c}.
Finally, a~\irref{dW} step completes the proof since conjuncts $\norm{x}^2 \leq \gamma^2$ and $\norm{x}^2 \neq \gamma^2$ in the resulting domain constraint imply the postcondition $\norm{x}^2 < \gamma^2$ by~\irref{qear}.

{\begin{sequentdeduction}[array]
  \linfer[dIcmp+dC]{
  \linfer[dC]{
  \linfer[dW]{
  \linfer[qear]{
    \lclose
  }
    {\lsequent{\norm{x}^2 \leq \gamma^2, \norm{x}^2 \neq \gamma^2 }{\norm{x}^2 < \gamma^2}}
  }
    {\lsequent{}{\dbox{\pevolvein{\D{x}=\genDE{x}}{\norm{x}^2 \leq \gamma^2  \land \cdots \land \norm{x}^2 \neq \gamma^2 }}{\norm{x}^2 < \gamma^2}}}
  }
    {\lsequent{\textcircled{c}}{\dbox{\pevolvein{\D{x}=\genDE{x}}{\norm{x}^2 \leq \gamma^2 \land \lied[]{\genDE{x}}{\lterm} \leq 0 \land v < k}}{\norm{x}^2 < \gamma^2}}}
  }
    {\lsequent{\textcircled{c}, v < k}{\dbox{\pevolvein{\D{x}=\genDE{x}}{\norm{x}^2 \leq \gamma^2 \land \lied[]{\genDE{x}}{\lterm} \leq 0}}{\norm{x}^2 < \gamma^2}}}
\end{sequentdeduction}
}%

The derivation of rule~\irref{StrictLyap} starts with a~\irref{cut} that proves stability using rule~\irref{Lyap} because the premises of~\irref{StrictLyap} are identical to those of~\irref{Lyap} except for a strict inequality requirement on the Lie derivative of $\lterm$.
The right conjunct of the right premise of~\irref{StrictLyap} is~\irref{cut} and Skolemized with~\irref{existsl}; the resulting antecedent is abbreviated with $\textcircled{e} \mnodefequiv \lforall{x}{\big(0 < \norm{x}^2 \leq \gamma^2 \limply \lied[]{\genDE{x}}{\lterm} < 0\big)}$ below.
Next, instantiating $\varepsilon \mnodefeq \gamma$ in the stability antecedent with~\irref{alll} and Skolemizing yields an initial disturbance $\delta > 0$ so that the ODE solution from states satisfying $\norm{x}^2 < \delta^2$ always stay within the $\gamma$ ball $\norm{x}^2 < \gamma^2$.
The resulting antecedent is abbreviated with $\textcircled{f} \mnodefequiv \lforall{x}{\big( \norm{x}^2 < \delta^2 \limply \dbox{\D{x}=\genDE{x}}{ \,\norm{x}^2 < \gamma^2}\big)}$.
{\begin{sequentdeduction}[array]
  \linfer[cut+Lyap]{
  \linfer[cut+existsl]{
  \linfer[alll+existsl]{
    \lsequent{\stabode{\D{x}=\genDE{x}}, \textcircled{e}, \delta > 0, \textcircled{f} }{\attrode{\D{x}=\genDE{x}}}
  }
    {\lsequent{\stabode{\D{x}=\genDE{x}}, \gamma > 0, \textcircled{e} }{\attrode{\D{x}=\genDE{x}}}}
  }
    {\lsequent{\stabode{\D{x}=\genDE{x}}}{\attrode{\D{x}=\genDE{x}}}}
  }
  {\lsequent{}{\astabode{\D{x}=\genDE{x}}}}
\end{sequentdeduction}
}%

The existential quantifier in the succedent is witnessed by $\delta$ with~\irref{existsr} and the resulting sequent is simplified by Skolemization and instantiation of \textcircled{f} (similar to~\irref{Lyap}) before axiom~\irref{stabattr} is used to further simplify the succedent using the stability assumption.
{\begin{sequentdeduction}[array]
  \linfer[existsr]{
  \linfer[allr+implyr]{
  \linfer[alll+implyl]{
  \linfer[stabattr]{
  \linfer[allr]{
    \lsequent{\textcircled{e},\dbox{\D{x}=\genDE{x}}{ \,\norm{x}^2 < \gamma^2},\varepsilon {>} 0}{ \ddiamond{\D{x}=\genDE{x}}{\norm{x}^2 < \varepsilon^2}}
  }
    {\lsequent{\textcircled{e},\dbox{\D{x}=\genDE{x}}{ \,\norm{x}^2 < \gamma^2} }{ \lforall{\varepsilon {>} 0}{\ddiamond{\D{x}=\genDE{x}}{\norm{x}^2 < \varepsilon^2}}}}
  }
    {\lsequent{\stabode{\D{x}=\genDE{x}}, \textcircled{e},\dbox{\D{x}=\genDE{x}}{ \,\norm{x}^2 < \gamma^2} }{\asymode{\D{x}=\genDE{x}}{x=0}}}
  }
    {\lsequent{\stabode{\D{x}=\genDE{x}}, \textcircled{e}, \textcircled{f}, \norm{x}^2 < \delta^2 }{ \asymode{\D{x}=\genDE{x}}{x=0}}}
  }
    {\lsequent{\stabode{\D{x}=\genDE{x}}, \textcircled{e}, \textcircled{f} }{ \lforall{x}{\big( \norm{x}^2 < \delta^2 \limply \asymode{\D{x}=\genDE{x}}{x=0}\big)}}}
  }
    {\lsequent{\stabode{\D{x}=\genDE{x}}, \textcircled{e}, \delta > 0, \textcircled{f} }{\attrode{\D{x}=\genDE{x}}}}
\end{sequentdeduction}
}%

The remaining open premise is a liveness property which is proved using rule~\irref{SPc} with the choice of compact staging set $\rsfvar \mnodefequiv \varepsilon^2 \leq \norm{x}^2 \leq \gamma^2$ and $\ptermA \mnodefequiv \lterm$.
{\begin{sequentdeduction}[array]
  \linfer[SPc]{
    \linfer[dC]{
    \linfer[dW+qear]{
      \lclose
    }
    {\lsequent{}{\dbox{\pevolvein{\D{x}=\genDE{x}}{\norm{x}^2 \geq \varepsilon^2 \land \norm{x}^2 {<} \gamma^2}}{\rsfvar}}}
    }
    {\lsequent{\dbox{\D{x}=\genDE{x}}{ \,\norm{x}^2 {<} \gamma^2}}{\dbox{\pevolvein{\D{x}=\genDE{x}}{\norm{x}^2 \geq \varepsilon^2}}{\rsfvar}}} !
    \linfer[]{
      \lclose
    }
    {\lsequent{\textcircled{e}, \varepsilon{>} 0, \rsfvar}{ \lied[]{\genDE{x}}{\lterm} < 0}}
  }
    {\lsequent{\textcircled{e},\dbox{\D{x}=\genDE{x}}{ \,\norm{x}^2 {<} \gamma^2},\varepsilon {>} 0}{ \ddiamond{\D{x}=\genDE{x}}{\norm{x}^2 {<} \varepsilon^2}}}
\end{sequentdeduction}
}%

The left premise proves with a differential cut~\irref{dC} of the antecedent and~\irref{dW+qear} from the resulting strengthened domain constraint $\norm{x}^2 \geq \varepsilon^2 \land \norm{x}^2 {<} \gamma^2$.
The right premise proves using \textcircled{e} with antecedents $\rsfvar$ and $\varepsilon {>} 0$ to prove its implication LHS.
\qedhere
\end{proof}

\begin{proof}[\textbf{Proof of~\rref{lem:expstabdl}}]
The quantifiers $\lexists{\alpha{>}0}{\lexists{\beta {>} 0}{\lexists{\delta {>} 0} \lforall{x}{\big(\neighborhood[\delta]{x=0} \limply \cdots\big)}}}$ syntactically express the respective quantifiers in the definition of exponential stability from~\rref{def:expstab} in \dL.
For an initial state satisfying $\neighborhood[\delta]{x=0}$, i.e., with $\norm{x(0)} < \delta$, the assignment $\pumod{y}{\alpha^2\norm{x}^2}$ sets the initial value of fresh variable $y$ (before the ODE) to $\alpha^2\norm{x(0)}^2$.
Let $x(t) : [0,T) \to \reals^n$ and $y(t) : [0,T) \to \reals$ respectively be the $x$ and $y$ projections of the unique, right-maximal solution of the ODE $\D{x}=\genDE{x},\D{y}=-2\beta y$.
By construction, the unique ODE solution for the $y$-coordinates is $y(t) = \alpha^2\norm{x(0)}^2\exp{(-2\beta t)}$, so the postcondition $\norm{x}^2 \leq y$ of the box modality expresses that for all times $0 \leq t < T$, $\norm{x(t)}^2 \leq  \alpha^2\norm{x(0)}^2\exp{(-2\beta t)}$ or, equivalently, $\norm{x(t)} \leq \alpha\norm{x(0)}\exp{(-\beta t)}$, as required. \qedhere
\end{proof}

\begin{proof}[\textbf{Proof of~\rref{lem:expstablyap}}]
The proof starts by instantiating the existentially quantified variables $\alpha,\beta$ in $\expstabode{\D{x}=\genDE{x}}$ with $\alpha \mnodefeq \frac{k_2}{k_1}$ and decay rate $\beta \mnodefeq k_3$.
Since $k_1,k_2,k_3$ are all positive constants, these choices satisfy $\alpha > 0, \beta > 0$.
{\begin{sequentdeduction}[array]
  \linfer[existsr]{
    \lsequent{}{\lexists{\delta {>} 0}{ \lforall{x}{\big(\norm{x}^2 < \delta^2 \limply \dbox{\pumod{y}{(\frac{k_2}{k_1})^2\norm{x}^2} ; \D{x}=\genDE{x},\D{y}=-2 k_3 y}{\,\norm{x}^2 \leq y} \big)}}}
  }
  {\lsequent{}{\expstabode{\D{x}=\genDE{x}}}}
\end{sequentdeduction}
}%

The subsequent~\irref{cut} step adds the premise of rule~\irref{ExpLyap} to the antecedents and Skolemizes it with~\irref{existsl}.
The resulting antecedent is abbreviated with $\textcircled{a} \mnodefequiv \lforall{x}{\big(\norm{x}^2 \leq \gamma^2 \limply  k_1^2 \norm{x}^2 \leq \lterm \leq k_2^2 \norm{x}^2 \land \lied[]{\genDE{x}}{\lterm} \leq -2 k_3 \lterm)}$.
Then~\irref{existsr} instantiates the succedent with $\delta \mnodefeq \frac{k_1}{k_2}\gamma$ and the sequent is propositionally simplified.
Note \textcircled{a} also implies $k_1 \leq k_2$ as $k_1, k_2 > 0$ so $\delta \leq \gamma$ and $\norm{x}^2 < \delta^2$ implies $\norm{x}^2 < \gamma$ in real arithmetic.
The hybrid program $\pumod{y}{(\frac{k_2}{k_1})^2\norm{x}^2} ; \D{x}=\genDE{x},\D{y}=-2 k_3 y$ is abbreviated with $\cdots$ in the first three steps.
{\begin{sequentdeduction}[array]
  \linfer[cut+existsl]{
  \linfer[existsr]{
  \linfer[implyr+implyl]{
    \lsequent{\gamma > 0, \textcircled{a}, \norm{x}^2 < (\frac{k_1}{k_2}\gamma)^2 }{\dbox{\pumod{y}{(\frac{k_2}{k_1})^2\norm{x}^2} ; \D{x}=\genDE{x},\D{y}=-2 k_3 y}{\,\norm{x}^2 \leq y} }
  }
    {\lsequent{\gamma > 0, \textcircled{a}}{ \lforall{x}{\big(\norm{x}^2 < (\frac{k_1}{k_2}\gamma)^2 \limply \dbox{\cdots}{\,\norm{x}^2 \leq y} \big)}}}
  }
    {\lsequent{\gamma > 0, \textcircled{a}}{\lexists{\delta {>} 0}{\lforall{x}{\big(\norm{x^2} < \delta^2 \limply \dbox{\cdots}{\,\norm{x}^2 \leq y} \big)}}}}
  }
    {\lsequent{}{\lexists{\delta {>} 0}{ \lforall{x}{\big(\norm{x^2} < \delta^2 \limply \dbox{\cdots}{\,\norm{x}^2 \leq y} \big)}}}}
\end{sequentdeduction}
}%

\irlabel{assignb|$\dibox{:=}$}

The discrete assignment $\pumod{y}{(\frac{k_2}{k_1})^2\norm{x}^2}$ sets the value of variable $y$ to $(\frac{k_2}{k_1})^2\norm{x}^2$ initially.
It is turned into an equational assumption with the assignment axiom~\irref{assignb} of \dL as follows (the axiom is omitted but can be found in the literature~\cite{DBLP:journals/jar/Platzer17,Platzer18}).
{\begin{sequentdeduction}
  \linfer[assignb]{
    \lsequent{\gamma > 0, \textcircled{a}, \norm{x}^2 < (\frac{k_1}{k_2}\gamma)^2, y=(\frac{k_2}{k_1})^2\norm{x}^2}{\dbox{\D{x}=\genDE{x},\D{y}=-2 k_3 y}{\,\norm{x}^2 \leq y} }
  }
    {\lsequent{\gamma > 0, \textcircled{a}, \norm{x}^2 < (\frac{k_1}{k_2}\gamma)^2 }{\dbox{\pumod{y}{(\frac{k_2}{k_1})^2\norm{x}^2} ; \D{x}=\genDE{x},\D{y}=-2 k_3 y}{\,\norm{x}^2 \leq y} }}
\end{sequentdeduction}
}%

The antecedents are abbreviated $\Gamma \mnodefequiv \gamma > 0, \textcircled{a},\norm{x}^2 < (\frac{k_1}{k_2}\gamma)^2, y = (\frac{k_2}{k_1})^2\norm{x}^2$.
The derivation continues with a differential cut~\irref{dC} adding formula $\norm{x}^2 < \gamma^2$ to the domain constraint.
This cut is abbreviated \textcircled{1} and proved further below.
The next differential cut~\irref{dC} adds formula $v \leq k_1^2 y$ to the domain constraint.
This cut is abbreviated \textcircled{2} and also proved further below.
The derivation is completed with a~\irref{dW+qear} step with the quantified antecedent \textcircled{a} and the domain constraint, since they imply the chain of inequalities $k_1^2 \norm{x}^2 \leq \lterm \leq k_1^2 y$, which implies the succedent (after~\irref{dW}) $\norm{x}^2 \leq y$ by~\irref{qear}.
{\begin{sequentdeduction}[array]
\linfer[dC]{
\linfer[dC]{
\linfer[dW]{
\linfer[qear]{
  \lclose
}
  {\lsequent{\textcircled{a},\norm{x}^2 < \gamma^2, v \leq k_1^2 y} {\norm{x}^2 \leq y}}
}
  {\textcircled{2} \qquad \lsequent{\Gamma} {\dbox{\pevolvein{\D{x}=\genDE{x},\D{y}=-2 k_3 y}{\norm{x}^2 < \gamma^2 \land v \leq k_1^2 y}}{\norm{x}^2 \leq y}}}
}
   {\textcircled{1} \qquad \lsequent{\Gamma} {\dbox{\pevolvein{\D{x}=\genDE{x},\D{y}=-2 k_3 y}{\norm{x}^2 < \gamma^2}}{\norm{x}^2 \leq y}}}
  }
    {\lsequent{\Gamma} {\dbox{\D{x}=\genDE{x},\D{y}=-2 k_3 y}{\norm{x}^2 \leq y}}}
\end{sequentdeduction}
}%

Returning to premise~\textcircled{1}, the derivation uses~\irref{Enc} to assume $\norm{x}^2 \leq \gamma^2$ in the domain constraint, since $\norm{x}^2 < \gamma^2$ is true initially.
Then, a~\irref{dC+dIcmp} step adds formula $\lterm < k_1^2 \gamma^2$ to the domain constraint.
This formula is proved true initially by~\irref{qear} with antecedents $\Gamma$ using the chain of inequalities from \textcircled{a}, $\lterm \leq k_2^2 \norm{x}^2 < k_2^2 \big(\frac{k_1}{k_2} \gamma\big)^2 = k_1^2 \gamma^2$.
It is proved invariant by~\irref{dIcmp} because domain constraint $\norm{x}^2 \leq \gamma^2$ and quantified antecedent \textcircled{a} proves the chain of inequalities $\lied[]{\genDE{x}}{\lterm} \leq {-}2 k_3 \lterm \leq -2 k_3 (k_1^2 \norm{x}^2) \leq 0$.
A~\irref{dW+qear} step completes the proof because the domain constraint $\norm{x}^2 \leq \gamma^2 \land \lterm < k_1^2 \gamma^2$ and quantified antecedent \textcircled{a} prove the chain of inequalities $k_1^2\norm{x}^2 \leq \lterm < k_1^2 \gamma^2$, which implies the succedent (after~\irref{dW}) $\norm{x}^2 < \gamma^2$ by~\irref{qear}.
{\begin{sequentdeduction}[array]
\linfer[Enc]{
\linfer[dC+dIcmp]{
\linfer[dW+qear]{
  \lclose
}
  {\lsequent{\Gamma} {\dbox{\pevolvein{\D{x}=\genDE{x},\D{y}=-2 k_3 y}{\norm{x}^2 \leq \gamma^2\land \lterm < k_1^2 \gamma^2}}{\norm{x}^2 < \gamma^2 }}}
}
  {\lsequent{\Gamma} {\dbox{\pevolvein{\D{x}=\genDE{x},\D{y}=-2 k_3 y}{\norm{x}^2 \leq \gamma^2 }}{\norm{x}^2 < \gamma^2 }}}
  }
    {\lsequent{\Gamma} {\dbox{\D{x}=\genDE{x},\D{y}=-2 k_3 y}{\norm{x}^2 < \gamma^2 }}}
\end{sequentdeduction}
}%

For premise~\textcircled{2}, the inequality $k_1^2 y - \lterm \geq 0$ is proved invariant using rule~\irref{dbxineq} with cofactor $\cofterm \mnodefeq -2k_3$ as follows.
{\begin{sequentdeduction}[array]
\linfer[cut+MbW]{
\linfer[dbxineq]{
  \linfer[qear]{
    \lclose
  }
  {\lsequent{\Gamma} {k_1^2 y - \lterm \geq 0}} !
  \linfer[qear]{
    \lclose
  }
  {\lsequent{\textcircled{a}, \norm{x}^2 < \gamma^2}{ -2k_1^2k_3y - \lied[]{\genDE{x}}{\lterm} \geq -2k_3 (k_1^2 y - v)}}
}
  {\lsequent{\Gamma} {\dbox{\pevolvein{\D{x}=\genDE{x},\D{y}=-2 k_3 y}{\norm{x}^2 < \gamma^2}}{k_1^2 y - \lterm \geq 0}}}
}
  {\lsequent{\Gamma} {\dbox{\pevolvein{\D{x}=\genDE{x},\D{y}=-2 k_3 y}{\norm{x}^2 < \gamma^2}}{v \leq k_1^2 y}}}
\end{sequentdeduction}
}%

The resulting left premise proves by~\irref{qear} because the antecedents \textcircled{a} and $y = (\frac{k_2}{k_1})^2\norm{x}^2$ in $\Gamma$ prove the chain of inequalities $\lterm \leq k_2^2 \norm{x}^2 = k_1^2 y$.
For the resulting right premise, the Lie derivative of $k_1^2y-v$ from the LHS of the postcondition is $-2k_1^2k_3y - \lied[]{\genDE{x}}{\lterm}$.
With domain constraint $\norm{x}^2 < \gamma^2$ and quantified antecedent \textcircled{a}, this derivative provably satisfies the chain of inequalities $-2k_1^2k_3y - \lied[]{\genDE{x}}{\lterm} \geq -2k_1^2k_3y + 2k_3\lterm = -2k_3(k_1^2y-\lterm)$ by~\irref{qear}.
\qedhere
\end{proof}

The following global stability definitions for an equilibrium point are motivated in~\rref{subsec:asymstabvar}.
They are used in the proof of~\rref{lem:globstabdl}.

\begin{definition}[Global stability~\cite{10.2307/j.ctvcm4hws,MR1201326,MR0450715}]
The origin $0 \in \reals^n$ of ODE $\D{x}=\genDE{x}$ is \textbf{globally asymptotically stable} if it is stable and its region of attraction is the entire state space, i.e., for all $x=x(0) \in \reals^n$, the right-maximal ODE solution $x(t) : [0,T) \to \reals^n$ satisfies $\lim_{t \to T}{x(t) = 0}$.
The origin is \textbf{globally exponentially stable} if there are positive constants $\alpha, \beta>0$ such that for all initial states $x=x(0) \in \reals^n$, the right-maximal ODE solution $x(t) : [0,T) \to \reals^n$ satisfies $\norm{x(t)} \leq \alpha\norm{x(0)}\exp{(-\beta t)}$ for all times $0 \leq t < T$.
\end{definition}

\begin{proof}[\textbf{Proof of~\rref{lem:globstabdl}}]
The proof is identical to Lemmas~\ref{lem:asymstabdl} and~\ref{lem:expstabdl}, respectively, except the existential quantification over a local neighborhood of the origin $\lexists{\delta > 0}{\lforall{x}{(\neighborhood[\delta]{x=0} \limply \cdots)}}$ is replaced by universal quantification over all initial states (where $\rfvar \mnodefequiv \ltrue$), i.e., $\lforall{x}{(\ltrue \limply \cdots)}$, as required by the global stability definitions. \qedhere
\end{proof}

\begin{proof}[\textbf{Proof of~\rref{lem:globstablyap}}]
The rules are derived in order starting with rule~\irref{StrictLyapGlob}.
First, observe that the first two premises of rule~\irref{StrictLyapGlob} imply the premises of~\irref{Lyap} because if the sign conditions on $\lterm$ and $\lied[]{\genDE{x}}{\lterm}$ are true globally, then they also hold for any choice of neighborhood of the origin.
Thus, the derivation starts with a~\irref{cut+Lyap} step which proves stability of the origin.
Next, the definition of $\attrodeP{\D{x}=\genDE{x}}{\ltrue}$ is logically unfolded and axiom~\irref{stabattr} is used to simplify the succedent, together with unfolding steps~\irref{allr+implyr}.
{\begin{sequentdeduction}[array]
\linfer[cut+Lyap]{
  \linfer[allr+implyr]{
  \linfer[stabattr]{
  \linfer[allr+implyr]{
    \lsequent{\varepsilon{>}0} {\ddiamond{\D{x}=\genDE{x}}{\norm{x}^2 < \varepsilon^2}}
  }
    {\lsequent{} {\lforall{\varepsilon{>}0}{\ddiamond{\D{x}=\genDE{x}}{\norm{x}^2 < \varepsilon^2}}}}
  }
    {\lsequent{\stabode{\D{x}=\genDE{x}}} {\asymode{\D{x}=\genDE{x}}{\ltrue} }}
  }
    {\lsequent{\stabode{\D{x}=\genDE{x}}} {\attrodeP{\D{x}=\genDE{x}}{\ltrue} }}
  }
    {\lsequent{} {\stabode{\D{x}=\genDE{x}} \land \attrodeP{\D{x}=\genDE{x}}{\ltrue} }}
\end{sequentdeduction}
}%

The derivation continues with a~\irref{cut+qear+existsl} step, introducing a fresh Skolem variable $b$ which stores the initial value of the Lyapunov function $\lterm$.
Next, a~\irref{cut} adds the box modality formula $\dbox{\D{x}=\genDE{x}}{\,\lterm \leq b}$ to the antecedents.
This cut proves by~\irref{dIcmp} because formula $\lterm \leq b$ is true initially and the premises of rule~\irref{StrictLyapGlob} prove the formula $\lied[]{\genDE{x}}{\lterm} \leq 0$.\footnote{When $x=0$, the premise $f(0)=0$ implies $\lied[]{\genDE{x}}{\lterm} = 0$.}
The subsequent~\irref{MbW} step strengthens the postcondition $\lterm \leq b$ of the antecedent box modality to $\norm{x}^2 < \gamma^2$ using the (Skolemized) rightmost premise of rule~\irref{StrictLyapGlob} (i.e., $\lterm$ is radially unbounded~\cite{10.2307/j.ctvcm4hws,MR1201326}).
{\begin{sequentdeduction}[array]
\linfer[cut+qear+existsl]{
\linfer[cut+dIcmp]{
\linfer[MbW]{
  \lsequent{\varepsilon{>}0, \dbox{\D{x}=\genDE{x}}{\norm{x}^2 < \gamma^2}}{\ddiamond{\D{x}=\genDE{x}}{\norm{x}^2 < \varepsilon^2}}
}
  {\lsequent{\varepsilon{>}0, \dbox{\D{x}=\genDE{x}}{\lterm \leq b}}{\ddiamond{\D{x}=\genDE{x}}{\norm{x}^2 < \varepsilon^2}}}
  }
  {\lsequent{\varepsilon{>}0, v=b}{\ddiamond{\D{x}=\genDE{x}}{\norm{x}^2 < \varepsilon^2}}}
  }
  {\lsequent{\varepsilon{>}0}{\ddiamond{\D{x}=\genDE{x}}{\norm{x}^2 < \varepsilon^2}}}
\end{sequentdeduction}
}%

Like the derivation of rule~\irref{StrictLyap} in~\rref{lem:lyapunov}, the remaining open premise is an ODE liveness property which is proved using rule~\irref{SPc} with the choice of compact staging set $\rsfvar \mnodefequiv \varepsilon^2 \leq \norm{x}^2 \leq \gamma^2$ and $\ptermA \mnodefequiv \lterm$.
The resulting left premise proves with a differential cut~\irref{dC} of the antecedent and~\irref{dW}.
The resulting right premise proves using~\irref{qear} from the middle premise of rule~\irref{StrictLyapGlob} and antecedents $\varepsilon{>}0$, $\rsfvar$.
{\begin{sequentdeduction}[array]
  \linfer[SPc]{
    \linfer[dC+dW]{
      \lclose
    }
    {\lsequent{\dbox{\D{x}=\genDE{x}}{ \,\norm{x}^2 {<} \gamma^2}}{\dbox{\pevolvein{\D{x}=\genDE{x}}{\norm{x}^2 \geq \varepsilon^2}}{\rsfvar}}} !
    \linfer[qear]{
      \lclose
    }
    {\lsequent{\varepsilon{>} 0, \rsfvar}{ \lied[]{\genDE{x}}{\lterm} < 0}}
  }
    {\lsequent{\varepsilon{>}0, \dbox{\D{x}=\genDE{x}}{\norm{x}^2 {<} \gamma^2}}{\ddiamond{\D{x}=\genDE{x}}{\norm{x}^2 {<} \varepsilon^2}}
}
\end{sequentdeduction}
}%

The derivation of rule~\irref{ExpLyapGlob} is similar to the derivation of rule~\irref{ExpLyap} from~\rref{lem:expstablyap}. The proof steps are repeated briefly.
The derivation starts by instantiating (using \irref{existsr}) the existentially quantified variables in succedent $\expstabodeP{\D{x}=\genDE{x}}{\ltrue}$ with $\alpha \mnodefeq \frac{k_2}{k_1}$ and decay rate $\beta \mnodefeq k_3$, followed by logical unfolding of the sequent.
{\begin{sequentdeduction}[array]
\linfer[existsr]{
\linfer[allr+implyr]{
\linfer[assignb]{
  \lsequent{y{=}(\frac{k_2}{k_1})^2\norm{x}^2}{\dbox{\D{x}{=}\genDE{x},\D{y}{=}-2k_3 y}{\,\norm{x}^2 \leq y}}
}
  {\lsequent{}{\dbox{\pumod{y}{(\frac{k_2}{k_1})^2\norm{x}^2} ; \D{x}{=}\genDE{x},\D{y}{=}-2k_3 y}{\,\norm{x}^2 \leq y}}}
}
  {\lsequent{} {\lforall{x}{\big(\ltrue \limply \dbox{\pumod{y}{(\frac{k_2}{k_1})^2\norm{x}^2} ; \D{x}{=}\genDE{x},\D{y}{=}-2k_3 y}{\,\norm{x}^2 \leq y}\big)}}}
  }
    {\lsequent{} {\expstabodeP{\D{x}{=}\genDE{x}}{\ltrue} }}
\end{sequentdeduction}
}%

The proof continues with a differential cut of the formula $\lterm \leq k_1^2 y$, which is proved invariant with its equivalent rephrasing $k_1^2 y - \lterm \geq 0$ and rule~\irref{dbxineq} using cofactor $\cofterm \mnodefeq -2k_3$.
Similar to~\rref{lem:expstablyap}, the cut formula is true initially because the antecedent $y=(\frac{k_2}{k_1})^2\norm{x}^2$ and the premise of~\irref{ExpLyapGlob} imply the chain of inequalities $\lterm \leq k_2^2 \norm{x}^2 = k_1^2 y$.
The premise of~\irref{ExpLyapGlob} are also used to show that the Lie derivative of $k_1^2 y - \lterm$ provably satisfies the chain of inequalities $-2k_1^2k_3y - \lied[]{\genDE{x}}{\lterm} \geq -2k_1^2k_3y + 2k_3\lterm = -2k_3(k_1^2y-\lterm)$.
The proof is completed by~\irref{dW+qear} using the premise of rule~\irref{ExpLyapGlob}.
{\begin{sequentdeduction}[array]
\linfer[dC]{
\linfer[dW]{
\linfer[qear]{
\linfer[qear]{
  \lclose
}
  {\lsequent{}{k_1^2 \norm{x}^2 \leq v}}
}
  {\lsequent{\lterm \leq k_1^2 y}{\,\norm{x}^2 \leq y}}
}
  {\lsequent{}{\dbox{\pevolvein{\D{x}=\genDE{x},\D{y}=-2k_3 y}{\lterm \leq k_1^2 y}}{\,\norm{x}^2 \leq y}}}
}
  {\lsequent{y=(\frac{k_2}{k_1})^2\norm{x}^2}{\dbox{\D{x}=\genDE{x},\D{y}=-2k_3 y}{\,\norm{x}^2 \leq y}}}
\\[-\normalbaselineskip]\tag*{\qedhere}
\end{sequentdeduction}
}%
\end{proof}

\begin{proof}[\textbf{Proof of~\rref{cor:expimpasym}}]
The two axioms are derived in order, starting with axiom~\irref{EStabStab}.
The derivation of axiom~\irref{EStabStab} starts by Skolemizing the existential quantifiers in $\expstabode{\D{x}=\genDE{x}}$ with~\irref{existsl}, then Skolemizing the succedent (\irref{allr}) and instantiating (\irref{existsr}) the existentially quantified $\delta {>} 0$ in the succedent with $\gamma \mnodefeq \min(\frac{\varepsilon}{\alpha}, \delta)$ (note $\gamma > 0$).
The sequent is then simplified, noting that $\gamma \leq \delta$ so that assumption $\norm{x}^2 < \gamma^2$ proves the implication LHS $\lforall{x}{\big(\norm{x}^2 < \delta^2 \limply \cdots \big)}$ in the antecedents.
The subformula with discrete assignment to $y$ in the antecedent is abbreviated with $\rrfvar_y \mnodefequiv \dbox{\pumod{y}{\alpha^2\norm{x}^2} ; \D{x}=\genDE{x},\D{y}=-2\beta y}{\,\norm{x}^2 \leq y}$.

{\begin{sequentdeduction}[array]
\linfer[existsl]{
\linfer[allr]{
\linfer[existsr]{
\linfer[implyr+implyl]{
  \lsequent{\alpha{>}0,\beta{>}0,\delta{>}0,\varepsilon {>} 0,\rrfvar_y, \norm{x}^2 < \gamma^2}{\dbox{\D{x}=\genDE{x}}{\,\norm{x^2} < \varepsilon^2}}
}
  {\lsequent{\alpha{>}0,\beta{>}0,\delta{>}0,\varepsilon {>} 0,\lforall{x}{\big(\norm{x}^2 < \delta^2 \limply \rrfvar_y \big)}}{\lforall{x}{\big( \norm{x}^2 < \gamma^2 \limply \dbox{\D{x}=\genDE{x}}{\,\norm{x^2} < \varepsilon^2} \big)}}}
}
  {\lsequent{\alpha{>}0,\beta{>}0,\delta{>}0,\varepsilon {>} 0,\lforall{x}{\big(\norm{x}^2 < \delta^2 \limply \rrfvar_y \big)}}{\lexists{\delta {>} 0}{ \lforall{x}{\big( \norm{x}^2 < \delta^2 \limply \dbox{\D{x}=\genDE{x}}{\,\norm{x}^2 < \varepsilon^2}\big)}}}}
}
  {\lsequent{\alpha{>}0,\beta{>}0,\delta{>}0,\lforall{x}{\big(\norm{x}^2 < \delta^2 \limply \rrfvar_y \big)}}{\stabode{\D{x}=\genDE{x}}}}
  }
  {\lsequent{\expstabode{\D{x}=\genDE{x}}}{\stabode{\D{x}=\genDE{x}}}}
\end{sequentdeduction}
}%

The discrete assignment in antecedent $\rrfvar_y$ is unfolded with~\irref{assignb}, similar to the proof of~\rref{lem:expstablyap}, yielding the antecedent $y{=}\alpha^2\norm{x}^2$ and abbreviated antecedent $\rfvar_y \mnodefequiv \dbox{\D{x}=\genDE{x},\D{y}=-2\beta y}{\,\norm{x}^2 \leq y}$.
Axiom~\irref{DG} adds differential ghost $\D{y}=-2\beta y$ to the succedent ODE and the postcondition is strengthened to $y < \varepsilon^2$ by~\irref{K} using assumption $\rfvar_y$.
The proof is completed with a~\irref{dbxineq} step with cofactor term $-2\beta$ because the (rephrased) postcondition $\varepsilon^2 - y > 0$ is proved by~\irref{qear} from the antecedents with the chain of inequalities $y = \alpha^2 \norm{x}^2 < \alpha^2\gamma^2 \leq \alpha^2 (\frac{\varepsilon}{\alpha})^2 = \varepsilon^2$.
The Lie derivative of $\varepsilon^2 - y$ is $2\beta y$ which provably satisfies the inequality $2\beta y \geq 2\beta y - 2\beta \varepsilon^2 = -2\beta(\varepsilon^2 - y)$ for $\beta > 0$.
\begin{sequentdeduction}[array]
  \linfer[assignb]{
  \linfer[DG]{
  \linfer[K]{
  \linfer[MbW]{
  \linfer[dbxineq]{
    \lclose
  }
    {\lsequent{\alpha{>}0,\beta{>}0,\delta{>}0,\varepsilon {>} 0,y{=}\alpha^2\norm{x}^2, \norm{x}^2 < \gamma^2}{\dbox{\D{x}{=}\genDE{x},\D{y}{=}{-}2\beta y}{\,\varepsilon^2 - y > 0}}}
  }
    {\lsequent{\alpha{>}0,\beta{>}0,\delta{>}0,\varepsilon {>} 0,y{=}\alpha^2\norm{x}^2, \norm{x}^2 < \gamma^2}{\dbox{\D{x}{=}\genDE{x},\D{y}{=}{-}2\beta y}{\,y < \varepsilon^2}}}
  }
    {\lsequent{\alpha{>}0,\beta{>}0,\delta{>}0,\varepsilon {>} 0,y{=}\alpha^2\norm{x}^2, \rfvar_y, \norm{x}^2 < \gamma^2}{\dbox{\D{x}{=}\genDE{x},\D{y}{=}{-}2\beta y}{\,\norm{x^2} < \varepsilon^2}}}
  }
  {\lsequent{\alpha{>}0,\beta{>}0,\delta{>}0,\varepsilon {>} 0,y{=}\alpha^2\norm{x}^2, \rfvar_y, \norm{x}^2 < \gamma^2}{\dbox{\D{x}{=}\genDE{x}}{\,\norm{x^2} < \varepsilon^2}}}
  }
  {\lsequent{\alpha{>}0,\beta{>}0,\delta{>}0,\varepsilon {>} 0,\rrfvar_y, \norm{x}^2 < \gamma^2}{\dbox{\D{x}{=}\genDE{x}}{\,\norm{x^2} < \varepsilon^2}}}
\end{sequentdeduction}

The derivation of axiom~\irref{EStabAttr} starts by unfolding and Skolemizing the existential quantifiers for the antecedent, with abbreviated ODE $\progxy \mnodefequiv \D{x}=\genDE{x},\D{y}=-2\beta y$ and subformula $\rrfvar_y \mnodefequiv \dbox{\pumod{y}{\alpha^2\norm{x}^2} ; \progxy}{\,\norm{x}^2 \leq y}$ (identically to the preceding derivation for axiom~\irref{EStabStab}).
The succedent is logically unfolded and the resulting antecedent $\rfvar$ proves the implication LHS in the antecedent $\lforall{x}{(\rfvar \limply \rrfvar_y)}$.
Succedent $\asymode{\D{x}=\genDE{x}}{x=0}$ is then Skolemized with~\irref{allr}.\footnote{Unlike earlier proofs , the formula is \emph{not} simplified using a stability assumption because the generalized formula $\expstabodeP{\D{x}=\genDE{x}}{\rfvar}$ does not directly imply stability of the origin unless formula $\rfvar$ provably contains a neighborhood of the origin.}
The subsequent step uses axioms~\irref{DG+DGall} to add the linear differential ghost $\D{y}=-2\beta y$ to both ODEs in the succedent.
The postcondition of the succedent diamond modality is monotonically strengthened with~\irref{MdW} and postcondition $(y < \varepsilon^2 \land \dbox{\progxy}{\,\norm{x}^2 \leq y})$.
The two resulting premises are abbreviated \textcircled{1} and \textcircled{2}; they are shown and proved further below.
{\begin{sequentdeduction}[array]
\linfer[existsl]{
\linfer[allr+implyr]{
\linfer[alll+implyl]{
\linfer[allr]{
\linfer[DG+DGall]{
\linfer[MdW]{
  \textcircled{1} ! \textcircled{2}
}
  {\lsequent{\alpha{>}0, \beta{>}0,\rrfvar_y}{\ddiamond{\progxy}{ \dbox{\progxy}{\, \norm{x}^2 < \varepsilon^2}}}}
}
  {\lsequent{\alpha{>}0, \beta{>}0,\rrfvar_y}{\ddiamond{\D{x}=\genDE{x}}{ \dbox{\D{x}=\genDE{x}}{\, \norm{x}^2 < \varepsilon^2}}}}
}
  {\lsequent{\alpha{>}0, \beta{>}0,\rrfvar_y, \rfvar}{\asymode{\D{x}=\genDE{x}}{x=0}}}
}
  {\lsequent{\alpha{>}0, \beta{>}0,\lforall{x}{\big(\rfvar \limply \rrfvar_y\big)},\rfvar}{\asymode{\D{x}=\genDE{x}}{x=0}}}
}
  {\lsequent{\alpha{>}0, \beta{>}0,\lforall{x}{\big(\rfvar \limply \rrfvar_y\big)}}{\attrodeP{\D{x}=\genDE{x}}{\rfvar} }}
  }
  {\lsequent {\expstabodeP{\D{x}=\genDE{x}}{\rfvar} }{\attrodeP{\D{x}=\genDE{x}}{\rfvar} }}
\end{sequentdeduction}
}%

From premise \textcircled{1}, a differential cut~\irref{dC} proves formula $y < \varepsilon^2$ invariant for the succedent ODE $\progxy$ using rule~\irref{dbxineq}.
The subsequent~\irref{dC} adds the postcondition of the antecedent to the domain constraint before~\irref{dW+qear} finish the derivation.
{\begin{sequentdeduction}[array]
\linfer[dC+dbxineq]{
\linfer[dC]{
\linfer[dW]{
\linfer[qear]{
  \lclose
}
  {\lsequent{y < \varepsilon^2 \land \norm{x}^2 \leq y}{\norm{x}^2 < \varepsilon^2}}
}
   {\lsequent{}{ \dbox{\pevolvein{\progxy}{y < \varepsilon^2 \land \norm{x}^2 \leq y}}{\, \norm{x}^2 < \varepsilon^2}}}
}
   {\lsequent{\dbox{\progxy}{\,\norm{x}^2 \leq y}}{ \dbox{\pevolvein{\progxy}{y < \varepsilon^2}}{\, \norm{x}^2 < \varepsilon^2}}}
  }
  {\lsequent{\beta{>}0, y < \varepsilon^2, \dbox{\progxy}{\,\norm{x}^2 \leq y}}{ \dbox{\progxy}{\, \norm{x}^2 < \varepsilon^2}}}
\end{sequentdeduction}
}%

From premise \textcircled{2}, the antecedent $\rrfvar_y$ is unfolded with~\irref{assignb}, then~\irref{Prog+DSeq} remove the right conjunct of the postcondition because postcondition $\dbox{\progxy}{\,\norm{x}^2 \leq y}$ is true after all runs of $\progxy$.
Axiom~\irref{BDG} removes the ODEs for $x$ in the succedent because $\norm{x}^2$ is bounded using antecedent $\dbox{\progxy}{\,\norm{x}^2 \leq y}$.
{\begin{sequentdeduction}[array]
\linfer[assignb]{
\linfer[Prog+DSeq]{
\linfer[BDG]{
  \lsequent{\beta{>}0}{\ddiamond{\D{y}=-2\beta y}{y < \varepsilon^2}}
}
   {\lsequent{\beta{>}0, \dbox{\progxy}{\,\norm{x}^2 \leq y}}{\ddiamond{\progxy}{y < \varepsilon^2}}}
}
   {\lsequent{\beta{>}0, \dbox{\progxy}{\,\norm{x}^2 \leq y}}{\ddiamond{\progxy}{\big(y < \varepsilon^2 \land \dbox{\progxy}{\,\norm{x}^2 \leq y}\big)}}}
  }
  {\lsequent{\beta{>}0, \rrfvar_y}{\ddiamond{\progxy}{\big(y < \varepsilon^2 \land \dbox{\progxy}{\,\norm{x}^2 \leq y}\big)}}}
\end{sequentdeduction}
}%

The remaining open premise is an ODE liveness property for variable $y$.
Its proof starts by introducing a fresh variable $y_0$ storing the initial value of $y$.
Then, rule~\irref{SPc} is used with $\rsfvar \mnodefequiv \big(\varepsilon^2 \leq y \leq y_0\big)$.
The resulting left premise is an invariance property of the ODE which proves using~\irref{MbW+dIcmp}, while the right premise proves by~\irref{qear}.
{\begin{sequentdeduction}[array]%
\linfer[cut+existsl]{
\linfer[SPc]{
  \linfer[MbW]{
  \linfer[dIcmp]{
    \lclose
  }
    {\lsequent{\beta{>}0, y=y_0}{\dbox{\pevolvein{\D{y}=-2\beta y}{y \geq \varepsilon^2}}{y \leq y_0}}}
  }
  {\lsequent{\beta{>}0, y=y_0}{\dbox{\pevolvein{\D{y}=-2\beta y}{y \geq \varepsilon^2}}{\rsfvar}}} !
  \linfer[qear]{
    \lclose
  }
  {\lsequent{\beta{>}0, \rsfvar}{-2\beta y < 0}}
}
  {\lsequent{\beta{>}0, y=y_0}{\ddiamond{\D{y}=-2\beta y}{y < \varepsilon^2}}}
}
  {\lsequent{\beta{>}0}{\ddiamond{\D{y}=-2\beta y}{y < \varepsilon^2}}}
\\[-\normalbaselineskip]\tag*{\qedhere}
\end{sequentdeduction}
}%
\end{proof}

\subsection{Proofs for General Stability}
This section derives proof rules for general stability and its specialized instances which are introduced and motivated in~\rref{sec:genstability}.

\begin{proof}[\textbf{Proof of~\rref{lem:genlyapunov}}]
The derivation of rule~\irref{LyapGen} generalizes the ideas behind the derivation of rule~\irref{Lyap}.
The derivation starts with an~\irref{allr} step, followed by a~\irref{cut} and Skolemization \irref{existsl} of the second (bottom) premise of rule~\irref{LyapGen}.
The resulting assumptions (for Skolem variables $\gamma, \delta, k$) are abbreviated with: $\textcircled{a} \mnodefequiv \lforall{x}{(\bdr{(\neighborhood[\gamma]{\rrfvar})} \limply \lterm \geq k)}$, $\textcircled{b} \mnodefequiv \lforall{x}{(\neighborhood[\delta]{\rfvar}  \limply \rrfvar \lor \lterm{<}k)}$, and $\textcircled{c} \mnodefequiv \lforall{x}{\big(\rrfvar \lor \lterm {<} k \limply \dbox{\pevolvein{\D{x}=\genDE{x}}{\cneighborhood[\gamma]{\rrfvar}}}{(\rrfvar \lor \lterm {<} k)}\big)}$.
A subsequent~\irref{MbW} step strengthens the postcondition monotonically since the formula $\neighborhood[\gamma]{\rrfvar} \limply \neighborhood[\varepsilon]{\rrfvar}$ is provable by~\irref{qear} for $\gamma \leq \varepsilon$.
{%
\begin{sequentdeduction}[array]
\linfer[allr]{
\linfer[cut+existsl]{
\linfer[MbW+qear]{
  \lsequent{0{<}\gamma, 0{<}\delta{\leq}\gamma, \textcircled{a} , \textcircled{b} , \textcircled{c} }{\lexists{\delta {>} 0}{ \lforall{x}{\big( \neighborhood[\delta]{\rfvar} \limply \dbox{\D{x}=\genDE{x}}{\,\neighborhood[\gamma]{\rrfvar}}\big)}}}
}
  {\lsequent{0{<}\gamma{\leq}\varepsilon, 0{<}\delta{\leq}\gamma, \textcircled{a} , \textcircled{b} , \textcircled{c} }{\lexists{\delta {>} 0}{ \lforall{x}{\big( \neighborhood[\delta]{\rfvar} \limply \dbox{\D{x}=\genDE{x}}{\,\neighborhood[\varepsilon]{\rrfvar}}\big)}}}}
}
  {\lsequent{\varepsilon > 0}{\lexists{\delta {>} 0}{ \lforall{x}{\big( \neighborhood[\delta]{\rfvar} \limply \dbox{\D{x}=\genDE{x}}{\,\neighborhood[\varepsilon]{\rrfvar}}\big)}}}}
  }
  {\lsequent{}{\stabodePR{\D{x}=\genDE{x}}{\rfvar}{\rrfvar}}}
\end{sequentdeduction}
}%

The existentially quantified $\delta$ in the succedent is witnessed with the Skolem variable $\delta$ in the antecedents and the sequent is simplified with~\irref{allr+implyr+implyl}, where the implication LHS in \textcircled{b} is proved.
{%
\begin{sequentdeduction}[array]
\linfer[existsr]{
\linfer[allr+implyr+implyl]{
  \lsequent{0{<}\gamma, 0{<}\delta{\leq}\gamma, \textcircled{a} , \textcircled{c}, \neighborhood[\delta]{\rfvar}, \rrfvar \lor v < k }{ \dbox{\D{x}=\genDE{x}}{\,\neighborhood[\gamma]{\rrfvar}}}
}
  {\lsequent{0{<}\gamma, 0{<}\delta{\leq}\gamma, \textcircled{a} , \textcircled{b} , \textcircled{c} }{ \lforall{x}{\big( \neighborhood[\delta]{\rfvar} \limply \dbox{\D{x}=\genDE{x}}{\,\neighborhood[\gamma]{\rrfvar}}\big)}}}
}
  {\lsequent{0{<}\gamma, 0{<}\delta{\leq}\gamma, \textcircled{a} , \textcircled{b} , \textcircled{c} }{\lexists{\delta {>} 0}{ \lforall{x}{\big( \neighborhood[\delta]{\rfvar} \limply \dbox{\D{x}=\genDE{x}}{\,\neighborhood[\gamma]{\rrfvar}}\big)}}}}
\end{sequentdeduction}
}%

The derivation continues with rule~\irref{Enc} to assume the closure formula $\cneighborhood[\gamma]{\rrfvar}$ in the domain constraint.
This step uses the first premise of rule~\irref{LyapGen}, i.e., precondition $\rfvar$ implies postcondition $\rrfvar$, so that the neighborhood formula $\neighborhood[\delta]{\rfvar}$ provably implies neighborhood formula $\neighborhood[\gamma]{\rrfvar}$ initially by~\irref{qear} for $\delta \leq \gamma$.
The subsequent~\irref{dC} step uses the antecedent~\textcircled{c} to prove the invariance of formula $\rrfvar \lor \lterm {<} k$ for the ODE $\pevolvein{\D{x}=\genDE{x}}{\cneighborhood[\gamma]{\rrfvar} \land (\rrfvar \lor \lterm {<} k )}$ and adds it to the domain constraint.
The derivation is completed using~\irref{dW+qear}, where the arithmetic premise is justified below.
{\renewcommand{\arraystretch}{1.4}%
\begin{sequentdeduction}[array]
\linfer[Enc]{
\linfer[dC]{
\linfer[dW]{
\linfer[qear]{
    \lclose
}
  {\lsequent{0{<}\gamma, \textcircled{a}, \cneighborhood[\gamma]{\rrfvar}, (\rrfvar \lor \lterm {<} k ) }{\neighborhood[\gamma]{\rrfvar}}}
}
  {\lsequent{0{<}\gamma, \textcircled{a} , \rrfvar \lor \lterm {<} k }{ \dbox{\pevolvein{\D{x}=\genDE{x}}{\cneighborhood[\gamma]{\rrfvar} \land (\rrfvar \lor \lterm {<} k )}}{\,\neighborhood[\gamma]{\rrfvar}}}}
}
  {\lsequent{0{<}\gamma, \textcircled{a} , \textcircled{c}, \rrfvar \lor \lterm {<} k }{ \dbox{\pevolvein{\D{x}=\genDE{x}}{\cneighborhood[\gamma]{\rrfvar}}}{\,\neighborhood[\gamma]{\rrfvar}}}}
}
  {\lsequent{0{<}\gamma, 0{<}\delta{\leq}\gamma, \textcircled{a} , \textcircled{c}, \neighborhood[\delta]{\rfvar}, \rrfvar \lor \lterm {<} k }{ \dbox{\D{x}=\genDE{x}}{\,\neighborhood[\gamma]{\rrfvar}}}}
\end{sequentdeduction}
}%

To prove arithmetical premise $\lsequent{0{<}\gamma, \textcircled{a}, \cneighborhood[\gamma]{\rrfvar}, \rrfvar \lor \lterm {<} k }{\neighborhood[\gamma]{\rrfvar}}$, note that the left disjunct $\rrfvar$ in the antecedent implies its neighborhood formula $\neighborhood[\gamma]{\rrfvar}$ in real arithmetic for $\gamma > 0$.
Thus, it suffices to justify the premise for the right disjunct, i.e., $\lsequent{\textcircled{a}, \cneighborhood[\gamma]{\rrfvar}, \lterm {<} k}{\neighborhood[\gamma]{\rrfvar}}$.
Antecedent \textcircled{a} is instantiated to obtain assumption $\bdr{(\neighborhood[\gamma]{\rrfvar})} \limply \lterm \geq k$.
This assumption says that the right disjunct $\lterm < k$ does \emph{not} occur on the boundary of $\bdr{(\neighborhood[\gamma]{\rrfvar})}$ which, together with antecedent $ \cneighborhood[\gamma]{\rrfvar}$ implies succedent $\neighborhood[\gamma]{\rrfvar}$ in real arithmetic.
\qedhere
\end{proof}

The following set stability definitions are standard~\cite{10.2307/j.ctvcm4hws,MR1201326}, except (compared to the literature) the following definitions do \emph{not} assume any topological properties of the set characterized by formula $\rfvar$.
The motivation for additional topological restrictions is explained in~\rref{cor:setstabsimp}.

\begin{definition}[Set stability~\cite{10.2307/j.ctvcm4hws,MR1201326}]
Let $\dist{x}{\rfvar}$ denote the distance of a point $x \in \reals^n$ to the set characterized by formula $\rfvar$.
The set characterized by formula $\rfvar$ is
\begin{itemize}
\item \textbf{stable} if, for all $\varepsilon {>} 0$, there exists $\delta {>} 0$ such that for all initial states $x=x(0)$ with $\dist{x}{\rfvar} < \delta$, the right-maximal ODE solution $x(t) : [0,T) \to \reals^n$ satisfies $\dist{x(t)}{\rfvar} < \varepsilon$ for all times $0 \leq t < T$,
\item \textbf{attractive} if there exists $\delta {>} 0$ such that for all initial states $x=x(0)$ with $\dist{x}{\rfvar} < \delta$, the right-maximal ODE solution $x(t) : [0,T) \to \reals^n$ satisfies the limit $\lim_{t \to T}{\dist{x(t)}{\rfvar} = 0}$,
\item \textbf{asymptotically stable} if it is stable and attractive,
\item \textbf{globally asymptotically stable} if it is stable and for all initial states $x=x(0)\in\reals^n$, the right-maximal ODE solution $x(t) : [0,T) \to \reals^n$ satisfies the limit $\lim_{t \to T}{\dist{x(t)}{\rfvar} = 0}$.
\end{itemize}
\end{definition}

The set stability definitions are formalized in \dL in~\rref{lem:setasymstabdl}.

\begin{proof}[\textbf{Proof of~\rref{lem:setasymstabdl}}]
Like~\rref{lem:asymstabdl}, the correctness of these definitions is immediate from the semantics of \dL formulas because these definitions directly syntactically express the definitions in \dL.
For $\varepsilon > 0$, the neighborhood formula $\neighborhood[\varepsilon]{\rfvar}$ characterizes the set of points $x \in \reals^n$ within distance $\varepsilon$ from $\rfvar$, i.e., $\dist{x}{\rfvar} < \varepsilon$.
Formula $\asymode{\D{x}=\genDE{x}}{\rfvar}$ syntactically expresses the limit $\lim_{t \to T}{\dist{x(t)}{\rfvar} = 0}$ for the right-maximal ODE solution $x(t) : [0,T) \to \reals^n$, as shown in the proof of~\rref{lem:asymstabdl}.
\end{proof}

\begin{proof}[\textbf{Proof of~\rref{cor:setstabsimp}}]
The three axioms are derived in order.

\begin{itemize}
\item[\irref{stabattrgen}]
The two directions of the inner equivalence of~\irref{stabattrgen} are proved separately.
The easier ``$\limply$'' direction follows by Skolemizing $\varepsilon$ in the succedent with~\irref{allr}, choosing the same $\varepsilon$ in the antecedent with~\irref{alll}, and then by~\irref{MdW} because the resulting postcondition $\dbox{\D{x}=\genDE{x}}{\,\neighborhood[\varepsilon]{\rfvar}}$ monotonically implies postcondition $\neighborhood[\varepsilon]{\rfvar}$ by differential skip~\irref{DX}.
{\begin{sequentdeduction}[array]
  \linfer[allr+alll]{
  \linfer[MdW]{
  \linfer[DX]{
   \lclose
  }
   {\lsequent{\dbox{\D{x}=\genDE{x}}{\neighborhood[\varepsilon]{\rfvar}} }{\neighborhood[\varepsilon]{\rfvar}}}
  }
   {\lsequent{\ddiamond{\D{x}=\genDE{x}}{ \dbox{\D{x}=\genDE{x}}{\,\neighborhood[\varepsilon]{\rfvar}}} }{\ddiamond{\D{x}=\genDE{x}}{\,\neighborhood[\varepsilon]{\rfvar}}}}
  }{\lsequent{\asymode{\D{x}=\genDE{x}}{\rfvar}}{\lforall{\varepsilon {>} 0}{\ddiamond{\D{x}=\genDE{x}}{\,\neighborhood[\varepsilon]{\rfvar}}}}}
\end{sequentdeduction}
}%

The more interesting ``$\lylpmi$'' direction uses the stability assumption.
The first step Skolemizes the succedent with~\irref{allr}, then the stability antecedent is instantiated with~\irref{alll} and Skolemized with~\irref{existsl} (yielding fresh Skolem variable $\delta$).
Using the resulting quantified assumption $\lforall{x}{\big( \neighborhood[\delta]{\rfvar} \limply  \dbox{\D{x}=\genDE{x}}{\,\neighborhood[\varepsilon]{\rfvar}} \big)} $, the postcondition of the succedent is monotonically strengthened to $\neighborhood[\delta]{\rfvar}$ with~\irref{MdW}.
The derivation is completed by~\irref{alll} instantiating the remaining quantified antecedent with $\delta$.

{\begin{sequentdeduction}[array]
  \linfer[allr]{
  \linfer[alll+existsl]{
  \linfer[MdW]{
  \linfer[alll]{
      \lclose
  }
    {\lsequent{\delta > 0,\lforall{\varepsilon {>} 0}{\ddiamond{\D{x}=\genDE{x}}{\,\neighborhood[\varepsilon]{\rfvar}}}}{\ddiamond{\D{x}=\genDE{x}}{\,\neighborhood[\delta]{\rfvar}}}}
  }
    {\lsequent{\delta > 0,\lforall{x}{\big( \neighborhood[\delta]{\rfvar} \limply \dbox{\D{x}=\genDE{x}}{\,\neighborhood[\varepsilon]{\rfvar}}\big)} ,\lforall{\varepsilon {>} 0}{\ddiamond{\D{x}=\genDE{x}}{\,\neighborhood[\varepsilon]{\rfvar}}}}{\ddiamond{\D{x}=\genDE{x}}{ \dbox{\D{x}=\genDE{x}}{\,\neighborhood[\varepsilon]{\rfvar}}} }}
 }
    {\lsequent{\stabodePR{\D{x}=\genDE{x}}{\rfvar}{\rfvar},\lforall{\varepsilon {>} 0}{\ddiamond{\D{x}=\genDE{x}}{\,\neighborhood[\varepsilon]{\rfvar}}},\varepsilon {>} 0}{\ddiamond{\D{x}=\genDE{x}}{ \dbox{\D{x}=\genDE{x}}{\,\neighborhood[\varepsilon]{\rfvar}}} }}
 }{\lsequent{\stabodePR{\D{x}=\genDE{x}}{\rfvar}{\rfvar},\lforall{\varepsilon {>} 0}{\ddiamond{\D{x}=\genDE{x}}{\,\neighborhood[\varepsilon]{\rfvar}}}}{\asymode{\D{x}=\genDE{x}}{\rfvar}}}
\end{sequentdeduction}
}%

\item[\irref{stabclosure}] This axiom is derived immediately by equivalently rewriting with arithmetic equivalences because for $\delta > 0$ the open neighborhood formulas $\neighborhood[\delta]{\rfvar}$ and $\neighborhood[\delta]{\closure{\rfvar}}$ are provably equivalent in real arithmetic by~\irref{qear}.

\item[\irref{stabclosed}] This axiom is proved by contradiction so its derivation starts by negating the invariance succedent using~\irref{diamond}.
{\begin{sequentdeduction}[array]
  \linfer[allr+implyr]{
  \linfer[diamond+notr]{
    \lsequent{\stabodePR{\D{x}=\genDE{x}}{\rfvar}{\rfvar},\rfvar, \ddiamond{\D{x}=\genDE{x}}{\lnot{\rfvar}}}{\lfalse}
  }
    {\lsequent{\stabodePR{\D{x}=\genDE{x}}{\rfvar}{\rfvar},\rfvar}{\dbox{\D{x}=\genDE{x}}{\rfvar}}}
  }{\lsequent{\stabodePR{\D{x}=\genDE{x}}{\rfvar}{\rfvar}}{\lforall{x}{\big( \rfvar \limply \dbox{\D{x}=\genDE{x}}{\rfvar}\big)}}}
\end{sequentdeduction}}%

Since formula $\rfvar$ characterizes a closed set, every point satisfying $\lnot{\rfvar}$ must be contained in some $\varepsilon > 0$ ball in the interior of the open set characterized by $\lnot{\rfvar}$.
Accordingly, these points are $\varepsilon > 0$ distance away from the set characterized by $\rfvar$ and therefore, the formula $\lnot{\rfvar} \lbisubjunct \lexists{\varepsilon{>}0}{\lnot{\neighborhood[\varepsilon]{\rfvar}}} $ is provable in real arithmetic.
This equivalence is used to rewrite the postcondition of the diamond modality antecedent before the resulting existentially quantified variable $\varepsilon$ is commuted with the diamond modality and Skolemized with~\irref{dBarcan+existsl} and~\irref{V} to extract the constant assumption $\varepsilon > 0$.
{\begin{sequentdeduction}[array]
 \linfer[MdW+qear]{
 \linfer[dBarcan+existsl]{
 \linfer[V]{
 \linfer[alll+existsl]{
 \linfer[alll+implyl]{
 \linfer[diamond]{
  \lclose
  }
  {\lsequent{\dbox{\D{x}=\genDE{x}}{\,\neighborhood[\varepsilon]{\rfvar}}, \ddiamond{\D{x}=\genDE{x}}{\lnot{\neighborhood[\varepsilon]{\rfvar}}}}{\lfalse}}
 }
 {\lsequent{\delta> 0, \lforall{x}{\big( \neighborhood[\delta]{\rfvar} \limply \dbox{\D{x}=\genDE{x}}{\,\neighborhood[\varepsilon]{\rfvar}}\big)}, \rfvar, \ddiamond{\D{x}=\genDE{x}}{\lnot{\neighborhood[\varepsilon]{\rfvar}}}}{\lfalse}}
 }
  {\lsequent{\stabodePR{\D{x}=\genDE{x}}{\rfvar}{\rfvar}, \rfvar, \varepsilon > 0, \ddiamond{\D{x}=\genDE{x}}{\lnot{\neighborhood[\varepsilon]{\rfvar}}}}{\lfalse}}
 }
   {\lsequent{\stabodePR{\D{x}=\genDE{x}}{\rfvar}{\rfvar}, \rfvar, \ddiamond{\D{x}=\genDE{x}}{\big(\varepsilon{>}0 \land \lnot{\neighborhood[\varepsilon]{\rfvar}}\big)}}{\lfalse}}
 }
   {\lsequent{\stabodePR{\D{x}=\genDE{x}}{\rfvar}{\rfvar}, \rfvar, \ddiamond{\D{x}=\genDE{x}}{\lexists{\varepsilon{>}0}{\lnot{\neighborhood[\varepsilon]{\rfvar}}}}}{\lfalse}}
 }{\lsequent{\stabodePR{\D{x}=\genDE{x}}{\rfvar}{\rfvar}, \rfvar, \ddiamond{\D{x}=\genDE{x}}{\lnot{\rfvar}}}{\lfalse}}
\end{sequentdeduction}}%

The stability assumption is instantiated with $\varepsilon$ using~\irref{alll} and Skolemized with~\irref{existsl}.
Since the formula $\rfvar \limply \neighborhood[\delta]{\rfvar}$ is provable in real arithmetic for $\delta > 0$, the implication LHS in the antecedents is proved with~\irref{alll+implyl}.
The proof is completed using~\irref{diamond} since the resulting box and diamond modality antecedents are contradictory. \qedhere
\end{itemize}
\end{proof}

\begin{proof}[\textbf{Proof of~\rref{lem:setstablyap}}]
Rule~\irref{SetLyapGen} is derived first before~\irref{SetLyap} and~\irref{SetStrictLyap} are derived from~\irref{SetLyapGen} as corollaries further below.
The derivation of rule~\irref{SetLyapGen} starts with a~\irref{LyapGen} step.
The (left) resulting premise $\rfvar \limply \rfvar$ proves trivially and is not shown below.
For the (right) resulting premise, the first two conjuncts under the nested quantifiers prove trivially from a~\irref{cut} of the second (bottom) premise of rule~\irref{SetLyapGen}.
It remains to prove the final conjunct for fresh Skolem variable $k$ with antecedent abbreviated $\textcircled{a} \mnodefequiv \lforall{x}{(\cneighborhood[\gamma]{\rfvar} \land \lnot{\rfvar} \limply \lied[]{\genDE{x}}{\lterm} \leq 0)}$ from the premise of~\irref{SetLyapGen}.
{%
\begin{sequentdeduction}[array]
 \linfer[LyapGen]{
 \linfer[cut]{
  \lsequent{\textcircled{a}, \rfvar \lor \lterm {<} k} {\dbox{\pevolvein{\D{x}=\genDE{x}}{\cneighborhood[\gamma]{\rfvar}}}{(\rfvar \lor \lterm {<} k)}}
 }
    {\lsequent{}{\lforall{\varepsilon{>}0}{
    \lexists{0{<}\gamma{\leq}\varepsilon}{
    \lexists{k}{\left(
     \begin{array}{l}
     \lforall{x}{(\bdr{(\neighborhood[\gamma]{\rfvar})} \limply \lterm \geq k)} \land \\
     \lexists{0{<}\delta{\leq}\gamma}{\lforall{x}{(\neighborhood[\delta]{\rfvar}  \limply \rfvar \lor \lterm{<}k)} } \land \\
     \lforall{x}{\big(\rfvar \lor \lterm {<} k \limply \dbox{\pevolvein{\D{x}=\genDE{x}}{\cneighborhood[\gamma]{\rfvar}}}{(\rfvar \lor \lterm {<} k)}\big)} \\
     \end{array}\right)}
    }}}}
 }{\lsequent{}{\stabodePR{\D{x}=\genDE{x}}{\rfvar}{\rfvar}}}
\end{sequentdeduction}
}%

The derivation continues using rule~\irref{DCC} to prove that $v < k$ is true along ODE solutions until the invariant $\rfvar$ is entered; the first step uses an equivalent propositional rephrasing of $\rfvar \lor v < k$ as $\lnot{\rfvar} \limply v < k$.
The two resulting premises are abbreviated \textcircled{1} and \textcircled{2} and continued below.
{\def\arraystretch{1.4}%
\begin{sequentdeduction}[array]
  \linfer[cut+MbW]{
  \linfer[DCC]{
    \textcircled{1} ! \textcircled{2}
  }
  {\lsequent{\textcircled{a}, \lnot{\rfvar} \limply v < k} {\dbox{\pevolvein{\D{x}=\genDE{x}}{\cneighborhood[\gamma]{\rfvar}}}{(\lnot{\rfvar} \limply v < k)}}}
  }
  {\lsequent{\textcircled{a}, \rfvar \lor \lterm {<} k} {\dbox{\pevolvein{\D{x}=\genDE{x}}{\cneighborhood[\gamma]{\rfvar}}}{(\rfvar \lor \lterm {<} k)}}}
\end{sequentdeduction}
}%

From premise \textcircled{1}, a~\irref{DX} step strengthens the antecedent to $\lterm < k$ using the domain constraint $\lnot{\rfvar}$.
Rule~\irref{dIcmp} completes the proof because formula $ \lied[]{\genDE{x}}{\lterm} \leq 0$ proves from antecedent \textcircled{a} with domain $\cneighborhood[\gamma]{\rfvar} \land \lnot{\rfvar}$.
{\def\arraystretch{1.4}%
\begin{sequentdeduction}[array]
  \linfer[DX]{
  \linfer[dIcmp]{
  \linfer[]{
      \lclose
  }
  {\lsequent{\textcircled{a}, \cneighborhood[\gamma]{\rfvar} \land \lnot{\rfvar}}{ \lied[]{\genDE{x}}{\lterm} \leq 0}}
  }
  {\lsequent{\textcircled{a}, v < k} {\dbox{\pevolvein{\D{x}=\genDE{x}}{\cneighborhood[\gamma]{\rfvar} \land \lnot{\rfvar}}}{v < k}}}
  }
  {\lsequent{\textcircled{a}, \lnot{\rfvar} \limply v < k} {\dbox{\pevolvein{\D{x}=\genDE{x}}{\cneighborhood[\gamma]{\rfvar} \land \lnot{\rfvar}}}{v < k}}}
\end{sequentdeduction}
}%

From premise \textcircled{2}, a~\irref{dW} step reduces the premise to an invariance question for formula $\rfvar$ (since $\lnot{\lnot{\rfvar}}$ is equivalent to $\rfvar$).
The~\irref{DMP} step weakens the domain constraint, which proves using the first premise of rule~\irref{SetLyapGen}.
{\def\arraystretch{1.4}%
\begin{sequentdeduction}[array]
  \linfer[dW]{
  \linfer[DMP]{
  \linfer[]{
    \lclose
  }
    {\lsequent{\rfvar} {\dbox{\D{x}=\genDE{x}}{\rfvar}}}
  }
    {\lsequent{\rfvar} {\dbox{\pevolvein{\D{x}=\genDE{x}}{\cneighborhood[\gamma]{\rfvar}}}{\rfvar}}}
  }
  {\lsequent{} {\dbox{\pevolvein{\D{x}=\genDE{x}}{\cneighborhood[\gamma]{\rfvar}}}{( \lnot{\lnot{\rfvar}} \limply \dbox{\pevolvein{\D{x}=\genDE{x}}{\cneighborhood[\gamma]{\rfvar}}}{\lnot{\lnot{\rfvar}}})}}}
\end{sequentdeduction}
}%

Rule~\irref{SetLyap} derives from~\irref{SetLyapGen} because the two rules share the invariance premise on $\rfvar$ and the latter two premises of rule~\irref{SetLyap} imply the latter premise of~\irref{SetLyapGen} in real arithmetic when $\rfvar$ characterizes a compact set.
The variable $\gamma$ is witnessed by $\varepsilon$ in the premise after~\irref{SetLyapGen}.
The proof of the arithmetic premise is explained below.
{\begin{sequentdeduction}[array]
  \linfer[SetLyapGen]{
  \linfer[allr+existsr]{
  \linfer[qear]{
    \lclose
  }
    {\lsequent{\varepsilon{>}0}{
   \left(\begin{array}{l}
   \lexists{k}{\left(
     \begin{array}{l}
     \lforall{x}{(\bdr{(\neighborhood[\varepsilon]{\rfvar})} \limply \lterm \geq k)} \land \\
     \lexists{0{<}\delta{\leq}\varepsilon}{\lforall{x}{(\neighborhood[\delta]{\rfvar} \land \lnot{\rfvar} \limply \lterm <k)} }
     \end{array}\right)} \land\\
   \lforall{x}{(\cneighborhood[\varepsilon]{\rfvar} \land \lnot{\rfvar} \limply \lied[]{\genDE{x}}{\lterm} \leq 0)}
    \end{array}\right)
    }}
  }
    {\lsequent{}{\lforall{\varepsilon{>}0}{
    \lexists{0{<}\gamma{\leq}\varepsilon}{
   \left(\begin{array}{l}
   \lexists{k}{\left(
     \begin{array}{l}
     \lforall{x}{(\bdr{(\neighborhood[\gamma]{\rfvar})} \limply \lterm \geq k)} \land \\
     \lexists{0{<}\delta{\leq}\gamma}{\lforall{x}{(\neighborhood[\delta]{\rfvar} \land \lnot{\rfvar} \limply \lterm <k)} }
     \end{array}\right)} \land\\
   \lforall{x}{(\cneighborhood[\gamma]{\rfvar} \land \lnot{\rfvar} \limply \lied[]{\genDE{x}}{\lterm} \leq 0)}
    \end{array}\right)
    }}}}
  }
  {\lsequent{} {\stabodePR{\D{x}=\genDE{x}}{\rfvar}{\rfvar}}}
\end{sequentdeduction}
}%

The conjunct $\lforall{x}{(\cneighborhood[\varepsilon]{\rfvar} \land \lnot{\rfvar} \limply \lied[]{\genDE{x}}{\lterm} \leq 0)}$ proves logically from the right conjunct of the middle premise of rule~\irref{SetLyap}.
For the existentially quantified conjunct, $\lexists{k}{\big( \cdots \big)}$, since formula $\rfvar$ characterizes a compact set, the boundary $\bdr{(\neighborhood[\varepsilon]{\rfvar})}$ is also compact and therefore the continuous Lyapunov function $\lterm$ must attain its minimum $k$ on that set.
This $k$ witnesses the quantifier $\lexists{k}{}$, and note that $k > 0$ from the left conjunct of the middle premise of rule~\irref{SetLyap}.
From the rightmost premise of rule~\irref{SetLyap}, the Lyapunov function satisfies $v \leq 0$ for all points $x \in \reals^n$ on the boundary characterized by formula $\bdr{\rfvar}$.
For each such point on the boundary, by continuity, there is a radius $\delta > 0$ where points in the open ball $\norm{x}<\delta$ satisfy $v < k$ because $k > 0$.
The union of all such balls over all points on the boundary is an open cover of the compact boundary which therefore has a finite subcover.
The minimum radius $\delta > 0$ of balls in this finite subcover witnesses the formula $\lexists{0{<}\delta{\leq}\varepsilon}{\lforall{x}{(\neighborhood[\delta]{\rfvar} \land \lnot{\rfvar} \limply \lterm{<}k)} }$, justifying the use of~\irref{qear}.

Rule~\irref{SetStrictLyap} is derived from rule~\irref{SetLyap} similar to the derivation of~\irref{StrictLyap} from~\irref{Lyap}.
The derivation starts with a~\irref{cut} of the set stability formula $\stabodePR{\D{x}=\genDE{x}}{\rfvar}{\rfvar}$ which proves by~\irref{SetLyap} because rules~\irref{SetStrictLyap} and~\irref{SetLyap} have identical premises except for a strict inequality on $\lied[]{\genDE{x}}{\lterm}$.
{\begin{sequentdeduction}[array]
  \linfer[cut+SetLyap]{
    \lsequent{\stabodePR{\D{x}=\genDE{x}}{\rfvar}{\rfvar}}{\lexists{\delta{>}0}{\attrodePR{\D{x}=\genDE{x}}{\neighborhood[\delta]{\rfvar}}{\rfvar}}}
  }
  {\lsequent{}{\stabodePR{\D{x}=\genDE{x}}{\rfvar}{\rfvar} \land \lexists{\delta{>}0}{\attrodePR{\D{x}=\genDE{x}}{\neighborhood[\delta]{\rfvar}}{\rfvar}}}}
\end{sequentdeduction}
}%

The stability antecedent is instantiated with $\varepsilon \mnodefeq 1$; the positive constant $1$ is chosen arbitrarily to obtain a neighborhood in which solutions are trapped.
After Skolemization, this yields an initial disturbance $\delta > 0$ and the antecedent $\lforall{x}{\big( \neighborhood[\delta]{\rfvar} \limply \dbox{\D{x}=\genDE{x}}{\,\neighborhood[1]{\rfvar}}\big)}$.
The succedent is witnessed with $\delta$, and the resulting sequent is simplified with~\irref{allr+implyr+implyl}.
Then, axiom~\irref{stabattrgen} simplifies the succedent using the stability antecedent.
{\begin{sequentdeduction}[array]
  \linfer[cut+existsl]{
  \linfer[existsr]{
  \linfer[allr+implyr+implyl]{
  \linfer[stabattrgen]{
    \lsequent{\delta{>}0, \dbox{\D{x}=\genDE{x}}{\,\neighborhood[1]{\rfvar}}, \neighborhood[\delta]{\rfvar}}{\lforall{\varepsilon{>}0}{\ddiamond{\D{x}=\genDE{x}}{\,\neighborhood[\varepsilon]{\rfvar}}}}
  }
  {\lsequent{\stabodePR{\D{x}=\genDE{x}}{\rfvar}{\rfvar}, \delta{>}0, \dbox{\D{x}=\genDE{x}}{\,\neighborhood[1]{\rfvar}}, \neighborhood[\delta]{\rfvar}}{\asymode{\D{x}=\genDE{x}}{\rfvar}}}
  }
    {\lsequent{\stabodePR{\D{x}=\genDE{x}}{\rfvar}{\rfvar}, \delta{>}0, \lforall{x}{\big( \neighborhood[\delta]{\rfvar} \limply \dbox{\D{x}=\genDE{x}}{\,\neighborhood[1]{\rfvar}}\big)}}{\attrodePR{\D{x}=\genDE{x}}{\neighborhood[\delta]{\rfvar}}{\rfvar}}}
  }
  {\lsequent{\stabodePR{\D{x}=\genDE{x}}{\rfvar}{\rfvar}, \delta{>}0, \lforall{x}{\big( \neighborhood[\delta]{\rfvar} \limply \dbox{\D{x}=\genDE{x}}{\,\neighborhood[1]{\rfvar}}\big)}}{\lexists{\delta{>}0}{\attrodePR{\D{x}=\genDE{x}}{\neighborhood[\delta]{\rfvar}}{\rfvar}}}}
  }
  {\lsequent{\stabodePR{\D{x}=\genDE{x}}{\rfvar}{\rfvar}}{\lexists{\delta{>}0}{\attrodePR{\D{x}=\genDE{x}}{\neighborhood[\delta]{\rfvar}}{\rfvar}}}}
\end{sequentdeduction}
}%

The proof of the liveness property in the open premise uses rule~\irref{SPc} with the choice of compact staging set $\rsfvar \mnodefequiv \cneighborhood[1]{\rfvar} \land \lnot{\neighborhood[\varepsilon]{\rfvar}}$ and $\ptermA \mnodefequiv \lterm$.
Note that formula $\rsfvar$ characterizes a compact set because $\rfvar$ is compact so conjuncts $\cneighborhood[1]{\rfvar}$ and $\lnot{\neighborhood[\varepsilon]{\rfvar}}$ are both closed and $\cneighborhood[1]{\rfvar}$ is bounded.

{\begin{sequentdeduction}[array]
  \linfer[allr]{
  \linfer[SPc]{
    \linfer[dC+dW]{
      \lclose
    }
    {\lsequent{\dbox{\D{x}=\genDE{x}}{\,\neighborhood[1]{\rfvar}}}{\dbox{\pevolvein{\D{x}=\genDE{x}}{\lnot{\neighborhood[\varepsilon]{\rfvar}}}}{\rsfvar}}} !
    \linfer[qear]{
      \lclose
    }
    {\lsequent{\varepsilon{>} 0, \rsfvar}{ \lied[]{\genDE{x}}{\lterm} < 0}}
  }
    {\lsequent{\delta{>}0, \dbox{\D{x}=\genDE{x}}{\,\neighborhood[1]{\rfvar}}, \neighborhood[\delta]{\rfvar}, \varepsilon{>}0}{\ddiamond{\D{x}=\genDE{x}}{\,\neighborhood[\varepsilon]{\rfvar}}}}
  }
    {\lsequent{\delta{>}0, \dbox{\D{x}=\genDE{x}}{\,\neighborhood[1]{\rfvar}}, \neighborhood[\delta]{\rfvar}}{\lforall{\varepsilon{>}0}{\ddiamond{\D{x}=\genDE{x}}{\,\neighborhood[\varepsilon]{\rfvar}}}}}
\end{sequentdeduction}
}%

The left premise proves with a cut~\irref{dC} of the antecedent and~\irref{dW}.
The right premise proves by real arithmetic~\irref{qear} using the middle premise of rule~\irref{SetStrictLyap} because the antecedents imply $\lnot{\rfvar}$.
 \qedhere
\end{proof}

The following definition of $\varepsilon$-stability is standard~\cite{DBLP:conf/cav/GaoKDRSAK19}, except the first quantification over $\gamma > \varepsilon$ is strict whereas in the original definition~\cite{DBLP:conf/cav/GaoKDRSAK19} it is $\gamma \geq \varepsilon$.
This difference is immaterial for the purpose of $\varepsilon$-stability as $\varepsilon$ is a numerical parameter for the radius of a ball around which disturbances to the origin are to be ignored.
In particular, an ODE $\varepsilon$-stable by the following definition is $\alpha \varepsilon$-stable for any $\alpha \in (0,1)$ by its original definition~\cite{DBLP:conf/cav/GaoKDRSAK19}.

\begin{definition}[Epsilon-Stability~\cite{DBLP:conf/cav/GaoKDRSAK19}]
The origin $0 \in \reals^n$ of ODE $\D{x}=\genDE{x}$ is \textbf{$\varepsilon$-stable} for a positive constant $\varepsilon > 0$ if for all $\gamma > \varepsilon$, there exists $\delta > 0$ such that for points $x=x(0)$ with $\norm{x} < \delta$, the right-maximal ODE solution $x(t) : [0,T) \to \reals^n$ satisfies $\norm{x(t)} < \gamma$ for all times $0 \leq t < T$.
\label{def:epsstabmath}
\end{definition}

\begin{proof}[\textbf{Proof of~\rref{lem:epsstab}}]
The formula $\stabodePR{\D{x}=\genDE{x}}{x=0}{\neighborhood[\varepsilon]{x=0}}$ is valid iff for all $\gamma > 0$, there exists $\delta > 0$ such that for points $x=x(0)$ with $\norm{x} < \delta$, the right-maximal ODE solution $x(t) : [0,T) \to \reals^n$ satisfies $\norm{x(t)} < \varepsilon + \gamma$ for all times $0 \leq t < T$, where the neighborhood $\neighborhood[\gamma]{\neighborhood[\varepsilon]{x=0}}$ is equivalently characterized by $\norm{x}^2 < \gamma+\varepsilon$.
This unfolded semantics is equivalent to the mathematical definition of $\varepsilon$-stability in~\rref{def:epsstabmath} by reindexing the universal quantifier with $\gamma \mapsto \gamma + \varepsilon$ instead. \qedhere
\end{proof}

\section{Counterexamples}
\label{app:cex}

This appendix provides counterexamples for the soundness issues highlighted in Sections~\ref{sec:genstability} and~\ref{sec:casestudies}.
The first counterexample illustrates the need to assume compactness, i.e., formula $\rfvar$ is closed \emph{and bounded} in rule~\irref{SetLyap}.
The remark after~\cite[Definition 8.1]{MR1201326} suggests that the following variant of~\irref{SetLyap} is sound for formulas $\rfvar$ that characterize a closed, invariant set:

\dinferenceRule[SetLyapbad|SLyap${_\geq}$\usebox{\Lightningval}]{Set Lyapunov}
{\linferenceRule
  {\lsequent{\rfvar}{ \lterm = 0} \qquad
   \lsequent{\lnot{\rfvar}}{ \lterm > 0 \land \lied[]{\genDE{x}}{\lterm} \leq 0}
   }
  {\lsequent{} {\stabodePR{\D{x}=\genDE{x}}{\rfvar}{\rfvar} }}
}{}

The rule~\irref{SetLyapbad} is unsound (indicated by $\mbox{\lightning}$); indeed, the rule~\irref{SetLyap} from~\rref{lem:setstablyap} is also unsound if the assumption that formula $\rfvar$ characterizes a bounded set is omitted.

\begin{counterexample}
\label{cex:khalil}
\newcommand{\excounter}{\ensuremath{\alpha_c}}

Consider the ODE $\excounter \mnodefequiv \D{y}=y, \D{t}=1$ and the formula $\rfvar \mnodefequiv y = 0$ which characterizes a closed invariant set of $\excounter$ that is \emph{not bounded}.
The Lyapunov function $\lterm \mnodefeq y^2\exp(-2t)$, satisfies all of the premises of rule~\irref{SetLyapbad} because $\lterm = 0$ when $y = 0$, $\lterm > 0$ for $y \not= 0$, and $\lied[]{\excounter}{\lterm} = 0$.
However, $\rfvar$ is not stable for ODE $\excounter$, as can be seem from~\rref{fig:cex}. The norm of the right-maximal solution from all initial states that satisfy $y \neq 0$ approach $\infty$.

\begin{figure}[t]
\centering%
\includegraphics[width=.45\textwidth]{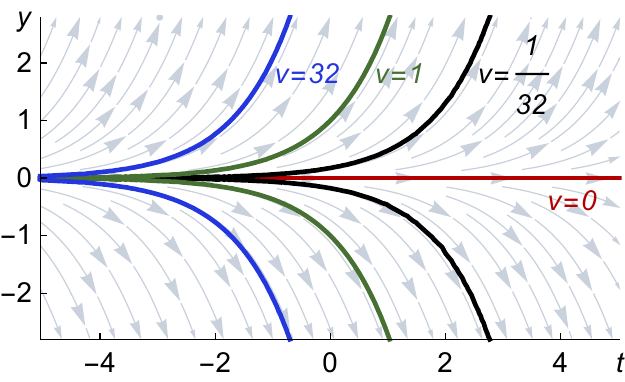}
\caption{An illustration of $\excounter$ and the Lyapunov function $v$ from~\rref{cex:khalil}, with level curves (where $\lterm = k$ for various $k$) shown in color.}
\label{fig:cex}
\end{figure}

This counterexample also illustrates the importance of the boundedness assumption for formula $\rfvar$ in~\rref{lem:setstablyap} for rule~\irref{SetLyap} since all other premises of the rule are satisfied by the above example.
\end{counterexample}

The second counterexample below shows that rule~\irref{Lyap} crucially needs the premise $\lterm(0) = 0$.
This premise is unsoundly omitted from the arithmetical conditions in~\cite[Equation 1]{DBLP:conf/tacas/AhmedPA20}.

\begin{counterexample}
\label{cex:ahmed}
Consider the ODE $y'=y$ with solution $y(t) = y_0 \exp(t)$ from initial value $y(0) = y_0$.
For all perturbed initial states $y_0 \not= 0$, $\norm{y(t)}$ approaches $\infty$ as $t \to \infty$ so this ODE is not stable (nor attractive).
However, the Lyapunov function $\lterm \mnodefeq 1$ trivially satisfies all of the premises of rule~\irref{Lyap} except the omitted premise $\lterm(0)=0$.
\end{counterexample}

\fi

\end{document}